\documentclass[11pt]{article}
\usepackage[margin=1in]{geometry}
\usepackage{amsmath,amssymb,amsthm}
\usepackage{algorithm,algpseudocode}
\usepackage{booktabs}
\usepackage{hyperref}
\usepackage{cleveref}
\usepackage{graphicx}
\usepackage{xcolor}
\usepackage{enumitem}
\usepackage{bm}
\usepackage{mathtools}
\usepackage{quantikz}
\usepackage{bbm}
\usepackage{subfig}
\usepackage{float}
\newcommand{\SU}{\mathrm{SU}}
\newcommand{\su}{\mathfrak{su}}
\newcommand{\CNOT}{\mathrm{CNOT}}

\newcommand{\cM}{\mathcal{M}}
\newcommand{\cN}{\mathcal{N}}
\newcommand{\cH}{\mathcal{H}}
\newcommand{\cU}{\mathcal{U}}
\newcommand{\cV}{\mathcal{V}}
\newcommand{\bbR}{\mathbb{R}}
\newcommand{\bbC}{\mathbb{C}}

\newtheorem{theorem}{Theorem}[section]
\newtheorem{lemma}[theorem]{Lemma}
\newtheorem{proposition}[theorem]{Proposition}
\newtheorem{corollary}[theorem]{Corollary}
\theoremstyle{plain}
\newtheorem{definition}[theorem]{Definition}
\newtheorem{assumption}[theorem]{Assumption}
\theoremstyle{plain}
\newtheorem{remark}[theorem]{Remark}

\newcommand{\bbN}{\mathbb{N}}
\newcommand{\bbE}{\mathbb{E}}
\newcommand{\bbP}{\mathbb{P}}

\title{%
  Identity-Paired Progressive Depth Training:\\
  When Trainability Persists Beyond Expressibility%
}
\author{%
  Athanasios Hadjidimoulas, Tirthak Patel, and Anastasios Kyrillidis\\
  Computer Science, Rice University \\ 
  QuanTAS research cluster, Ken Kennedy Institute at Rice University\\
  \texttt{\{th81, tirthak.patel, anastasios\}@rice.edu}
}
\date{}
\begin{document}
\maketitle
\begin{abstract}
Variational Quantum Algorithms (VQAs) are a leading paradigm for
near-term quantum computing, yet their training suffers from
sensitivity to circuit depth, initialization, and landscape pathologies
such as barren plateaus.  We study \emph{progressive depth training}
(PDT)---a layerwise curriculum that trains a shallow circuit before
appending new layers---and identify a fundamental obstacle: fixed
entangling gates (CNOTs) in hardware-efficient ans\"atze cause
\emph{initialization shock}, an energy spike when new layers are
added.  We propose \emph{identity-paired progressive depth training}
(IP-PDT), which appends forward/inverse block pairs---each consisting
of a standard rotation$+$CNOT block followed by its reverse---that
compose to the identity at initialization.  Because the adjacent CNOT
rings cancel, the effective circuit retains only \textit{a single
entangling layer} surrounded by \textit{overparameterized local
rotations}.  We prove a simple \textit{Reachable Set Saturation Theorem}:
under this construction the variational manifold expands exactly once
(when post-entangler rotations are first introduced) and then
\emph{saturates}; all subsequent depth increases provide pure
overparameterization of single-qubit unitaries.  Despite this
saturation, progressive addition of rotation parameters can continue to
improve optimization outcomes---a phenomenon we term
\emph{trainability beyond expressibility}.  We formalize IP-PDT as a
continuation method on nested manifolds, prove monotone energy
guarantees under an acceptance rule, and connect energy error to
ground-state fidelity through spectral-gap inequalities.
A detailed resource analysis shows that IP-PDT achieves lower total gate cost
than both baselines by eliminating most CNOT gates.  

Experiments span
six-qubit Transverse-Field Ising, Tilted Ising, and Random Ising
Hamiltonians, a broader nine-Hamiltonian benchmark, and system sizes up to $n=16$
qubits.  IP-PDT matches or outperforms
full-depth baselines on Hamiltonians whose ground states are
well-approximated by the single-entangler reachable set, with
particularly strong gains in limited-budget regimes.  At $n=16$ the
hardware-efficient baseline degrades sharply while the
single-entangler methods remain trainable.  
\end{abstract}
\section{Introduction}\label{sec:intro}
Variational quantum algorithms (VQAs) convert limited-depth quantum
evolution into an optimization problem driven by classical
computation~\cite{peruzzo2014vqe,cerezo2021vqa}.  At their core lies
the minimization of an energy expectation
$E(\bm\theta)=\langle\boldsymbol{\psi}(\bm\theta)|\mathbf{H}|\boldsymbol{\psi}(\bm\theta)\rangle$, where
$\mathbf{H}$ is a problem Hamiltonian and $|\boldsymbol{\psi}(\bm\theta)\rangle$ is produced
by a parameterized quantum circuit (PQC).  As circuit depth grows,
\emph{expressibility} improves---the set of reachable quantum states
expands---yet \emph{trainability} can simultaneously degrade through
vanishing gradients~\cite{mcclean2018barren}, rugged
landscapes~\cite{anschuetz2022beyond}, and sensitivity to
initialization.

A natural remedy is to treat depth as a \emph{curriculum}: optimize a
shallow circuit, then grow.  This \emph{progressive depth training}
(PDT) strategy is employed in ADAPT-VQE~\cite{grimsley2019adaptive},
layerwise training~\cite{skolik2021layerwise}, and initialization
transfer schemes~\cite{grant2019initialization}.  However, a practical
obstacle remains relatively less-addressed.  In hardware-efficient
ans\"atze (HEA)~\cite{kandala2017hardware}, each layer contains fixed
entangling gates (e.g., CNOTs) that \emph{cannot} be set to zero.
Appending a new layer often disrupts the trained quantum state,
producing an ``initialization shock''---a spike in energy that
can erase the progress of prior stages.
The core difficulty is that entangling gates in the HEA are
non-parameterized and therefore cannot be continuously deactivated:
there is no parameter setting that makes a CNOT ring act as the
identity, so every depth expansion irreversibly commits to additional
entanglement.

\medskip\noindent\textbf{Our approach.}
We propose \emph{identity-paired progressive depth training} (IP-PDT),
in which each depth expansion appends a \emph{pair} of blocks---one
standard (rotations $\to$ CNOT ring) and one reverse (reversed CNOT
ring $\to$ inverse rotations)---that compose to the identity at
initialization.  The key mechanism is that the CNOT ring and its
reverse, being adjacent in the circuit, cancel exactly
($\mathbf{U}_{\mathrm{ent}} \mathbf{U}_{\mathrm{ent}}^{-1}=\mathbf{I}$).  The resulting effective
circuit therefore contains only \textit{a single entangling layer}
(from the initial depth-1 block) surrounded by
\textit{overparameterized local rotations}.  All depth expansions
beyond the first provide no additional entanglement---they are pure
overparameterization of single-qubit unitaries. 

Yet, as we show empirically, progressive
training on these redundant parameters continues to yield consistent
optimization improvements.
This observation challenges the prevailing intuition that deeper
entanglement is needed for better VQA performance and points to a regime
where \emph{trainability gains can persist beyond expressibility
gains}.
Our contributions are summarized as follows:\vspace{-0.2cm}
\begin{enumerate}[leftmargin=*,itemsep=2pt]
  \item \textit{Identity-Paired PDT (IP-PDT):} A layerwise growth
    strategy that eliminates initialization shock by appending
    forward/reverse block pairs whose CNOT gates cancel, leaving only
    local rotations around a fixed entangling layer
    (\Cref{sec:method}).  Identity pairing is the \emph{architectural}
    ingredient---it produces the single-entangler circuit with its
    gate-count savings and guarantees exact energy preservation at each
    expansion; the optimization improvement itself is driven by the
    progressive curriculum, as our ablations show.\vspace{-0.1cm}
  \item \textit{Reachable Set Saturation Theorem:} We prove that the
    variational manifold of IP-PDT circuits expands \emph{exactly once}
    (at the first depth increase) and then saturates; subsequent
    expansions are strictly overparameterization of $\SU(2)$ rotations
    (\Cref{thm:saturation}).\vspace{-0.1cm}
  \item \textit{Continuation-theoretic framework:} We formalize IP-PDT
    as a homotopy continuation method on nested manifolds and prove
    monotone energy guarantees under an acceptance rule
    (\Cref{sec:theory}).\vspace{-0.1cm}
  \item \textit{Basin preservation and fidelity bounds:} We track the
    shallow-stage minimizer under identity expansion via the implicit
    function theorem (\Cref{thm:continuation}), and \emph{separately}
    prove that gradient descent on the previous-stage parameters---new
    layers held fixed---contracts at a geometric rate
    within the inherited basin (\Cref{thm:basin}); joint optimization is
    left to future work.  We connect energy error to ground-state
    fidelity through spectral-gap inequalities
    (\Cref{sec:basin,sec:gap_fidelity}).\vspace{-0.1cm}
  \item \textit{Resource estimation:} We provide a comparison
    of gate counts, per-step costs, and cumulative training costs for
    the baseline, na\"ive PDT, and IP-PDT, showing that IP-PDT achieves
    the lowest total gate cost (\Cref{sec:resource}).\vspace{-0.1cm}
  \item \textit{Systematic experiments:} On Ising
    Hamiltonians (TFIM, Tilted Ising, Random Ising), IP-PDT matches or
    outperforms full-depth baselines and improves on naive PDT across
    most instances, with the largest gains in limited-budget regimes;
    ablations indicate that the progressive curriculum drives these gains
    (\Cref{sec:experiments}).\vspace{-0.1cm}
\end{enumerate}

\medskip\noindent\textbf{Related work.}
Our work intersects three research threads.
\textit{Progressive and layerwise VQE methods.}
ADAPT-VQE~\cite{grimsley2019adaptive} grows circuits by greedily
selecting operators from a pool; it is adaptive but expensive and, to our knowledge,
lacks closed-form convergence guarantees.
Skolik et al.~\cite{skolik2021layerwise} propose layerwise training
with zero initialization of new parameters, avoiding the operator
selection cost but not addressing the CNOT-induced initialization
shock.
Closest in spirit, \v{Z}unkovi\v{c} et al.~\cite{zunkovic2026adiabatic}
run VQE iteratively along an \emph{adiabatic path} of Hamiltonians,
warm-starting each optimization from the previous solution and deriving
conditions under which gradient-based training avoids barren plateaus and
local optima---continuation in Hamiltonian space, complementary to our
depth-space curriculum.
In the classical setting, Szlendak et al.~\cite{szlendak2024rpt}
cast progressive training as randomized coordinate descent and prove
convergence for smooth objectives via a unified quantity~$L_\mathbf{P}$.
Our work differs from this thread in that we exploit the specific
structure of quantum circuits to provide lossless depth growth
with provable energy monotonicity.

\textit{Identity-block and warm-start initialization.}
Grant et al.~\cite{grant2019initialization} propose identity blocks
to mitigate barren plateaus by initializing new layers near the
identity.  Their construction applies to individual parameterized
gates; ours pairs entire HEA blocks (including CNOTs) and
exploits the resulting cancellation structure, which eliminates
entangling gates rather than merely initializing them.
Layer-wise pretraining is a long-standing idea in classical deep
learning; the quantum layerwise scheme of Skolik et
al.~\cite{skolik2021layerwise} likewise uses progressive
initialization, but, to our knowledge, neither faces the entanglement barrier specific
to quantum circuits that we address here.

\textit{Overparameterization in quantum circuits.}
Larocca et al.~\cite{larocca2023theory} characterize the onset of
overparameterization in quantum neural networks, showing that
gradient-free regions vanish above a critical parameter count.
Holmes et al.~\cite{holmes2022connecting} connect expressibility to
gradient magnitudes.  Our results address a \emph{different} invariant
than the dynamical Lie algebra (DLA) that sets Larocca's
overparameterization onset.  The Saturation Theorem shows the SE
reachable \emph{set} stops growing at $D=2$ (\Cref{thm:saturation}):
each qubit's post-entangler factor is confined to $\SU(2)$, so the
ansatz adds only $O(n)$ parameters per layer rather than ranging over the
$4^n-1$ dimensions of $\mathfrak{su}(2^n)$ that a deep HEA can explore.
We do not equate the two invariants---computing the SE DLA exactly is
left open---but this gap between $O(n)$ parameter growth and the
exponential-dimensional algebra is the structural reason
behind both effects we observe: the single-qubit expressibility ceiling
we prove (\Cref{thm:saturation}).

\medskip\noindent\textbf{Paper organization.}
\Cref{sec:prelim} establishes notation and background.
\Cref{sec:method} describes IP-PDT and the circuit families.
\Cref{sec:theory} presents the theoretical analysis.
\Cref{sec:experiments} reports experiments.
\Cref{sec:conclusion} concludes.
\section{Preliminaries}\label{sec:prelim}
We collect the mathematical background needed to follow the
theoretical developments.  Readers familiar with quantum computing
may proceed directly to \Cref{sec:method}; the notation table
(\Cref{tab:notation}) suffices as a reference.

\medskip\noindent\textbf{Hilbert space and Dirac notation.}
A system of $n$ qubits lives in the Hilbert space
$\cH=(\bbC^{2})^{\otimes n}\cong\bbC^{2^n}$, where $\otimes$ denotes
the tensor (Kronecker) product.  We write quantum states in
\emph{Dirac notation}: a column vector $\boldsymbol{\psi}\in\bbC^{2^n}$ with
$\|\boldsymbol{\psi}\|=1$ is denoted $|\boldsymbol{\psi}\rangle$ (a ``ket''), and its conjugate
transpose is $\langle\boldsymbol{\psi}|$ (a ``bra'').  The inner product of two
states is $\langle\boldsymbol{\phi}|\boldsymbol{\psi}\rangle\in\bbC$, and the expected value of
an observable $\mathbf{H}$ (a Hermitian matrix) in state $|\boldsymbol{\psi}\rangle$ is
$\langle\boldsymbol{\psi}|\mathbf{H}|\boldsymbol{\psi}\rangle\in\bbR$.
The standard computational-basis states for $n$ qubits are
$|b_1 b_2\cdots b_n\rangle$ with $b_q\in\{0,1\}$, forming an
orthonormal basis of $\cH$.  The all-zeros state
$|0\rangle^{\otimes n}=|0\cdots 0\rangle$ is abbreviated
$|0^n\rangle$ and serves as the default initial state of a quantum
circuit.

\medskip\noindent\textbf{Tensor products.}
If $\mathbf{A}$ is an operator on qubit~$q_1$ and $\mathbf{B}$ on qubit~$q_2$, the
joint operator $\mathbf{A}\otimes \mathbf{B}$ acts on the tensor-product space
$\bbC^2\otimes\bbC^2=\bbC^4$.  The notation
$\bigotimes_{q=1}^n \mathbf{A}_q = \mathbf{A}_1\otimes \mathbf{A}_2\otimes\cdots\otimes \mathbf{A}_n$
describes independent (i.e., non-entangling) operations on each
qubit.  A crucial algebraic identity for tensor products of
square matrices is:
\begin{equation}\label{eq:tensor_mult}
  \Bigl(\bigotimes_{q=1}^n \mathbf{A}_q\Bigr)
  \Bigl(\bigotimes_{q=1}^n \mathbf{B}_q\Bigr)
  = \bigotimes_{q=1}^n (\mathbf{A}_q \mathbf{B}_q),
\end{equation}
which follows from the mixed-product property of the Kronecker
product.  This identity is used repeatedly in our analysis.

\medskip\noindent\textbf{Pauli matrices.}
The three Pauli matrices:
\[
  \mathbf{X}=\begin{pmatrix}0&1\\1&0\end{pmatrix},\quad
  \mathbf{Y}=\begin{pmatrix}0&-i\\i&0\end{pmatrix},\quad
  \mathbf{Z}=\begin{pmatrix}1&0\\0&-1\end{pmatrix}
\]
are Hermitian, unitary, traceless $2\times 2$ matrices.  Together with
the identity $\mathbf{I}_2$, they form a basis for all $2\times 2$ complex
matrices.  Any $n$-qubit Hamiltonian can be decomposed as a real
linear combination of $n$-fold tensor products of Pauli matrices
(called \emph{Pauli strings}):
$\mathbf{H}=\sum_k c_k \mathbf{P}_{k,1}\otimes\cdots\otimes \mathbf{P}_{k,n}$, where each
$\mathbf{P}_{k,q}\in\{\mathbf{I},\mathbf{X},\mathbf{Y},\mathbf{Z}\}$.

\medskip\noindent\textbf{Pauli rotation gates.}
For $\sigma\in\{\mathbf{X},\mathbf{Y},\mathbf{Z}\}$ and angle $\theta\in\bbR$, the Pauli
rotation is defined as:
\begin{equation}\label{eq:pauli_rot}
  \mathbf{R}_{\sigma}(\theta)
  = \exp \bigl(-i\tfrac{\theta}{2}\sigma\bigr)
  = \cos \bigl(\tfrac{\theta}{2}\bigr) \mathbf{I}_2
    - i\sin \bigl(\tfrac{\theta}{2}\bigr) \sigma.
\end{equation}
Key properties: $\mathbf{R}_{\sigma}(\theta)$ is unitary,
$\mathbf{R}_{\sigma}(\theta)^{\dagger}=\mathbf{R}_{\sigma}(-\theta)$, and
$\mathbf{R}_{\sigma}(\theta_1)\mathbf{R}_{\sigma}(\theta_2)
=\mathbf{R}_{\sigma}(\theta_1+\theta_2)$.

\medskip\noindent\textbf{Parameterized rotations.}
For angles $(\alpha,\beta,\gamma)\in\bbR^3$, define the forward
single-qubit rotation:
\begin{equation}\label{eq:fwd_rot}
  \mathbf{R}^{\mathrm{fwd}}(\alpha,\beta,\gamma)
  = \mathbf{R}_Y(\alpha) \mathbf{R}_Z(\beta) \mathbf{R}_X(\gamma)
\end{equation}
and its inverse:
\begin{equation}\label{eq:inv_rot}
  \mathbf{R}^{\mathrm{inv}}(\alpha,\beta,\gamma)
  = \mathbf{R}_X(-\gamma) \mathbf{R}_Z(-\beta) \mathbf{R}_Y(-\alpha).
\end{equation}
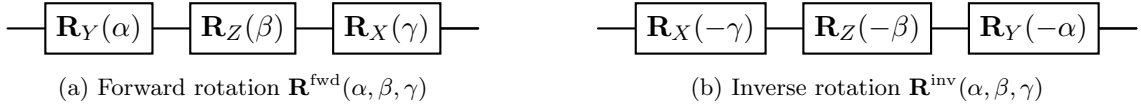
\begin{figure}[t]
\centering
\subfloat[Forward rotation $\mathbf{R}^{\mathrm{fwd}}(\alpha,\beta,\gamma)$]{%
\begin{quantikz}[row sep=0.3cm, column sep=0.5cm]
  \qw & \gate{\mathbf{R}_Y(\alpha)} & \gate{\mathbf{R}_Z(\beta)} & \gate{\mathbf{R}_X(\gamma)} & \qw
\end{quantikz}%
\label{fig:single_fwd}%
}
\hspace{1.2cm}
\subfloat[Inverse rotation $\mathbf{R}^{\mathrm{inv}}(\alpha,\beta,\gamma)$]{%
\begin{quantikz}[row sep=0.3cm, column sep=0.5cm]
  \qw & \gate{\mathbf{R}_X(-\gamma)} & \gate{\mathbf{R}_Z(-\beta)} & \gate{\mathbf{R}_Y(-\alpha)} & \qw
\end{quantikz}%
\label{fig:single_inv}%
}
\caption{Single-qubit rotation gates.  (a)~The forward rotation
applies $\mathbf{R}_Y$, $\mathbf{R}_Z$, $\mathbf{R}_X$ in sequence.  (b)~The inverse rotation
negates all angles \emph{and} reverses the gate order, yielding the
adjoint: $\mathbf{R}^{\mathrm{inv}}=[\mathbf{R}^{\mathrm{fwd}}]^{\dagger}$.
Composing~(a) followed by~(b) yields $\mathbf{I}_2$
(\Cref{lem:single_qubit_inv}).}
\label{fig:single_rotations}
\end{figure}
\begin{lemma}[Single-qubit inverse cancellation]
\label{lem:single_qubit_inv}
  For all $(\alpha,\beta,\gamma)\in\bbR^3$,
  \[
    \mathbf{R}^{\mathrm{fwd}}(\alpha,\beta,\gamma) 
    \mathbf{R}^{\mathrm{inv}}(\alpha,\beta,\gamma) = \mathbf{I}_2.
  \]
\end{lemma}
\begin{proof}
Substitute the definitions~\eqref{eq:fwd_rot} of
$\mathbf{R}^{\mathrm{fwd}}$ and~\eqref{eq:inv_rot} of
$\mathbf{R}^{\mathrm{inv}}$, and concatenate the two products into a single
chain of five factors:
\begin{equation}\label{eq:fwd_inv_chain0}
  \mathbf{R}^{\mathrm{fwd}}(\alpha,\beta,\gamma) 
  \mathbf{R}^{\mathrm{inv}}(\alpha,\beta,\gamma)
  = \mathbf{R}_Y(\alpha) \mathbf{R}_Z(\beta) 
    \mathbf{R}_X(\gamma) \mathbf{R}_X(-\gamma) 
    \mathbf{R}_Z(-\beta) \mathbf{R}_Y(-\alpha).
\end{equation}
The composition law $\mathbf{R}_{\sigma}(\theta_1)\mathbf{R}_{\sigma}(\theta_2)
=\mathbf{R}_{\sigma}(\theta_1+\theta_2)$ stated
after~\eqref{eq:pauli_rot} (with $\mathbf{R}_{\sigma}(0)=\mathbf{I}_2$) gives, for
each single axis $\sigma\in\{\mathbf{X},\mathbf{Y},\mathbf{Z}\}$, the inversion identity
$\mathbf{R}_{\sigma}(\theta)\mathbf{R}_{\sigma}(-\theta)
=\mathbf{R}_{\sigma}(0)=\mathbf{I}_2$. We apply it to the innermost adjacent pair
and work outward; at each step the surviving outer factors are
unchanged because the cancelled pair contracts to $\mathbf{I}_2$:
\begin{align}
  \mathbf{R}^{\mathrm{fwd}} \mathbf{R}^{\mathrm{inv}}
  &= \mathbf{R}_Y(\alpha) \mathbf{R}_Z(\beta) 
     \underbrace{\mathbf{R}_X(\gamma) \mathbf{R}_X(-\gamma)}_{= \mathbf{R}_X(0)=\mathbf{I}_2} 
     \mathbf{R}_Z(-\beta) \mathbf{R}_Y(-\alpha) \notag\\
  &= \mathbf{R}_Y(\alpha) 
     \underbrace{\mathbf{R}_Z(\beta) \mathbf{R}_Z(-\beta)}_{= \mathbf{R}_Z(0)=\mathbf{I}_2} 
     \mathbf{R}_Y(-\alpha) \notag\\
  &= \underbrace{\mathbf{R}_Y(\alpha) \mathbf{R}_Y(-\alpha)}_{= \mathbf{R}_Y(0)=\mathbf{I}_2}
   = \mathbf{I}_2. \label{eq:fwd_inv_chain}
\end{align}
The cancellation telescopes because $\mathbf{R}^{\mathrm{inv}}$ reverses
\emph{both} the angles \emph{and} the gate order relative to
$\mathbf{R}^{\mathrm{fwd}}$. Since the chain~\eqref{eq:fwd_inv_chain0} has no global
phase and each contraction is exact, the identity holds for every
$(\alpha,\beta,\gamma)\in\bbR^3$.
\end{proof}

\noindent\textbf{The special unitary group $\SU(2)$.}
The group $\SU(2)$ consists of all $2\times 2$ unitary matrices with
determinant~$1$:
$\SU(2)=\{\mathbf{U}\in\bbC^{2\times 2}:\mathbf{U}^{\dagger}\mathbf{U}=\mathbf{I}_2, \det(\mathbf{U})=1\}$.
It is a three-dimensional Lie group with Lie algebra $\su(2)$
spanned by $\{iX, iY, iZ\}$.  Every $\mathbf{U}\in\SU(2)$ (up to a global
phase) admits a three-parameter Euler decomposition; in particular,
our $\mathbf{R}^{\mathrm{fwd}}$ parameterizes all of $\SU(2)$ as
$(\alpha,\beta,\gamma)$ range over $\bbR^3$.

\medskip\noindent\textbf{Quantum system.}
We work on the $n$-qubit Hilbert space
$\cH\cong(\bbC^{2})^{\otimes n}$ of dimension $d=2^n$.  A Hermitian
operator $\mathbf{H}=\mathbf{H}^\dagger\in\bbC^{d\times d}$ (the \emph{Hamiltonian})
has eigen-decomposition:
\begin{equation}\label{eq:eigdecomp}
  \mathbf{H} = \sum_{i=1}^{d}\lambda_i|\mathbf{v}_i\rangle\langle \mathbf{v}_i|,
  \qquad
  \lambda_1\le\lambda_2\le\cdots\le\lambda_d,
\end{equation}
with orthonormal eigenbasis $\{|\mathbf{v}_i\rangle\}$.  The \emph{spectral gap}
is $\Delta:=\lambda_2-\lambda_1\ge 0$.

\medskip\noindent\textbf{Variational Quantum Eigensolver (VQE).}
A parameterized quantum circuit $\mathbf{U}(\bm\theta)$ prepares a trial state
$|\boldsymbol{\psi}(\bm\theta)\rangle=\mathbf{U}(\bm\theta)|0^n\rangle$, and a classical
optimizer minimizes the \emph{energy function}:
\begin{equation}\label{eq:vqe_energy}
  E(\bm\theta)
  = \langle 0^n|\mathbf{U}(\bm\theta)^{\dagger} \mathbf{H} \mathbf{U}(\bm\theta)|0^n\rangle
  \ge\lambda_1
  \quad\text{(variational principle)}.
\end{equation}

\medskip\noindent\textbf{CNOT ring.}
A controlled-NOT (CNOT) gate flips a target qubit conditioned on a
control qubit:
$\CNOT_{c\to t}|a\rangle_c|b\rangle_t=|a\rangle_c|a\oplus b\rangle_t$.
We arrange $n$ CNOT gates in a ring topology:
\begin{equation}\label{eq:cnot_ring}
  \mathbf{U}_{\mathrm{ent}}
  =\prod_{q=1}^{n}\CNOT_{q\to(q\bmod n)+1},
\end{equation}
where $\CNOT_{1\to 2}$ is applied first (rightmost) and
$\CNOT_{n\to 1}$ last.  The \emph{reverse CNOT ring} applies the same
gates in the opposite order:
\begin{equation}\label{eq:cnot_ring_rev}
  \mathbf{U}_{\mathrm{ent}}^{-1}
  =\prod_{q=n}^{1}\CNOT_{q\to(q\bmod n)+1}.
\end{equation}
Since each CNOT is self-adjoint ($\CNOT^2=\mathbf{I}_4$) and the reverse-order
product is the inverse, we have
\begin{equation}\label{eq:cnot_cancel}
  \mathbf{U}_{\mathrm{ent}} \mathbf{U}_{\mathrm{ent}}^{-1}
  =\mathbf{U}_{\mathrm{ent}}^{-1} \mathbf{U}_{\mathrm{ent}}=\mathbf{I}_{2^n}.
\end{equation}
This cancellation is the engine behind the identity-pair construction.
\begin{figure}[t]
\centering
\subfloat[Forward CNOT ring $\mathbf{U}_{\mathrm{ent}}$]{%
\begin{quantikz}[row sep=0.35cm]
  \lstick{$q_1$} & \ctrl{1} & \qw      & \qw      & \targ{}   & \qw \\
  \lstick{$q_2$} & \targ{}  & \ctrl{1} & \qw      & \qw       & \qw \\
  \lstick{$q_3$} & \qw      & \targ{}  & \ctrl{1} & \qw       & \qw \\
  \lstick{$q_4$} & \qw      & \qw      & \targ{}  & \ctrl{-3} & \qw
\end{quantikz}%
\label{fig:cnot_fwd}%
}\hspace{0.8cm}
\subfloat[Reverse CNOT ring $\mathbf{U}_{\mathrm{ent}}^{-1}$]{%
\begin{quantikz}[row sep=0.35cm]
  \lstick{$q_1$} & \targ{}   & \qw      & \qw      & \ctrl{1} & \qw \\
  \lstick{$q_2$} & \qw       & \qw      & \ctrl{1} & \targ{}  & \qw \\
  \lstick{$q_3$} & \qw       & \ctrl{1} & \targ{}  & \qw      & \qw \\
  \lstick{$q_4$} & \ctrl{-3} & \targ{}  & \qw      & \qw      & \qw
\end{quantikz}%
\label{fig:cnot_rev}%
}
\caption{CNOT rings for $n=4$ qubits.  (a)~Forward ring: gates are
applied in order $\CNOT_{1\to 2},\ldots,\CNOT_{4\to 1}$.
(b)~Reverse ring: same gates in opposite order.  Placing (a) immediately
followed by (b) yields $\mathbf{U}_{\mathrm{ent}} \mathbf{U}_{\mathrm{ent}}^{-1}=\mathbf{I}_{16}$.}
\label{fig:cnot_ring}
\end{figure}
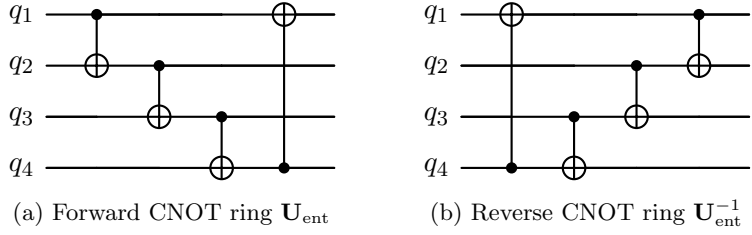
\begin{remark}
  Although the CNOT ring is parameter-free (its matrix is fixed), it
  is the \emph{only} operation that couples different qubits.  Without
  it, the ansatz reduces to a product of independent single-qubit
  rotations and can only represent product states.
\end{remark}

\medskip\noindent\textbf{Rotation layers.}
For a parameter block
$\bm\theta^{(d)}=(\alpha_1^{(d)},\beta_1^{(d)},\gamma_1^{(d)},
\ldots,\alpha_n^{(d)},\beta_n^{(d)},\gamma_n^{(d)})\in\bbR^{3n}$,
define the \emph{forward rotation layer}:
\begin{equation}\label{eq:rot_fwd_layer}
  \mathcal{R}^{\mathrm{fwd}}(\bm\theta^{(d)})
  =\bigotimes_{q=1}^{n}
   \mathbf{R}^{\mathrm{fwd}}(\alpha_q^{(d)},\beta_q^{(d)},\gamma_q^{(d)}),
\end{equation}
and the \emph{inverse rotation layer}
$\mathcal{R}^{\mathrm{inv}}(\bm\theta^{(d)})
=[\mathcal{R}^{\mathrm{fwd}}(\bm\theta^{(d)})]^\dagger$.
\begin{figure}[t]
\centering
\begin{quantikz}[row sep=0.3cm, column sep=0.4cm]
  \lstick{$q_1$} & \gate{\mathbf{R}_Y(\alpha_1)} & \gate{\mathbf{R}_Z(\beta_1)}
                  & \gate{\mathbf{R}_X(\gamma_1)} & \qw \\
  \lstick{$q_2$} & \gate{\mathbf{R}_Y(\alpha_2)} & \gate{\mathbf{R}_Z(\beta_2)}
                  & \gate{\mathbf{R}_X(\gamma_2)} & \qw \\
  \lstick{$q_3$} & \gate{\mathbf{R}_Y(\alpha_3)} & \gate{\mathbf{R}_Z(\beta_3)}
                  & \gate{\mathbf{R}_X(\gamma_3)} & \qw \\
  \lstick{$q_4$} & \gate{\mathbf{R}_Y(\alpha_4)} & \gate{\mathbf{R}_Z(\beta_4)}
                  & \gate{\mathbf{R}_X(\gamma_4)} & \qw
\end{quantikz}
\caption{Parallel forward rotation layer
$\mathcal{R}^{\mathrm{fwd}}(\bm\theta)$ for $n=4$ qubits.  Each
qubit undergoes an independent three-gate rotation; no entanglement
occurs.  The inverse layer $\mathcal{R}^{\mathrm{inv}}$ negates all
angles and reverses the gate order on each wire.}
\label{fig:rotation_layer}
\end{figure}
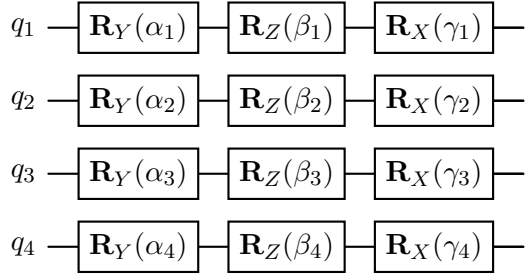

\medskip\noindent\textbf{Notation conventions.}
Throughout, boldface symbols ($\bm\theta$, $\bm\phi$) denote
parameter vectors, $\nabla_{\bm\theta}$ is the gradient, and
$\nabla^2_{\bm\theta\bm\theta}$ is the Hessian.  We use $\|\cdot\|$
for both the $\ell_2$ vector norm and the operator (spectral) norm.
The notation $\mathbf{A}\succeq\mu \mathbf{I}$ means $\mathbf{A}$ has all eigenvalues at least
$\mu$.  \Cref{tab:notation} collects the principal symbols.
\begin{table}[t]
\centering\small
\caption{Notation summary.}\label{tab:notation}
\begin{tabular}{@{}lll@{}}
\toprule
\textbf{Symbol} & \textbf{Type} & \textbf{Definition} \\
\midrule
$n$               & $\bbN$       & Number of qubits \\
$d=2^n$           & $\bbN$       & Hilbert space dimension \\
$D$               & $\bbN$       & Number of rotation layers \\
$\bm\theta^{(k)}$ & $\bbR^{3n}$ & Rotation parameters for layer $k$ \\
$\mathbf{H}$               & Hermitian    & Hamiltonian \\
$\lambda_i$,$\Delta$ & $\bbR$    & Eigenvalues; spectral gap \\
$E_D(\bm\theta)$  & $\bbR$      & Energy at depth $D$ \\
$\mathbf{U}_{\mathrm{ent}}$ & Unitary     & CNOT ring (fixed) \\
$\cM_D$           & $\subseteq\cH$ & Reachable set at depth $D$ \\
$T$               & $\bbN$       & Optimization steps per stage \\
$\eta$            & $\bbR_{>0}$  & Learning rate \\
$F(\boldsymbol{\psi})$         & $[0,1]$      & Ground-state fidelity \\
$\sigma$          & $\bbR_{>0}$  & Jitter scale for identity-pair init. \\
\bottomrule
\end{tabular}
\end{table}
\section{Identity-Paired Progressive Depth Training}\label{sec:method}
We describe two circuit architectures---the standard hardware-efficient
ansatz (HEA) and the single-entangler circuit that arises from
identity-pair CNOT cancellation---and three training protocols.
\subsection{Standard hardware-efficient ansatz (HEA)}
\label{sec:hea}
The standard HEA 
stacks $D$ identical \emph{blocks}, each consisting of a forward
rotation layer followed by a CNOT ring.  A single block at layer~$d$
is the unitary:
\begin{equation}\label{eq:hea_block}
  \mathbf{B}_d(\bm\theta^{(d)})
  = \mathbf{U}_{\mathrm{ent}} \mathcal{R}^{\mathrm{fwd}}(\bm\theta^{(d)}).
\end{equation}
The depth-$D$ HEA circuit is
\begin{equation}\label{eq:hea_circuit}
  \mathbf{U}_D^{\mathrm{HEA}}(\bm\theta)
  = \prod_{d=D-1}^{0}\mathbf{B}_d(\bm\theta^{(d)})
  = \prod_{d=D-1}^{0}
    \bigl[\mathbf{U}_{\mathrm{ent}} 
    \mathcal{R}^{\mathrm{fwd}}(\bm\theta^{(d)})\bigr],
\end{equation}
with parameter vector
$\bm\theta=(\bm\theta^{(0)},\ldots,\bm\theta^{(D-1)})\in\bbR^{3nD}$.
The descending product $\prod_{d=D-1}^{0}$ is time-ordered:
block~$0$ acts first on $|0^n\rangle$ (rightmost), and block~$D-1$
acts last (leftmost), following the standard right-to-left operator
composition convention.
\begin{figure}[t]
\centering
\begin{quantikz}[row sep=0.3cm, column sep=0.3cm]
  \lstick{$q_1$}
    & \gate{\mathbf{R}_Y} & \gate{\mathbf{R}_Z} & \gate{\mathbf{R}_X}
    & \ctrl{1} & \qw      & \qw      & \targ{}
    & \gate{\mathbf{R}_Y} & \gate{\mathbf{R}_Z} & \gate{\mathbf{R}_X}
    & \ctrl{1} & \qw      & \qw      & \targ{}  & \qw \\
  \lstick{$q_2$}
    & \gate{\mathbf{R}_Y} & \gate{\mathbf{R}_Z} & \gate{\mathbf{R}_X}
    & \targ{}  & \ctrl{1} & \qw      & \qw
    & \gate{\mathbf{R}_Y} & \gate{\mathbf{R}_Z} & \gate{\mathbf{R}_X}
    & \targ{}  & \ctrl{1} & \qw      & \qw      & \qw \\
  \lstick{$q_3$}
    & \gate{\mathbf{R}_Y} & \gate{\mathbf{R}_Z} & \gate{\mathbf{R}_X}
    & \qw      & \targ{}  & \ctrl{1} & \qw
    & \gate{\mathbf{R}_Y} & \gate{\mathbf{R}_Z} & \gate{\mathbf{R}_X}
    & \qw      & \targ{}  & \ctrl{1} & \qw      & \qw \\
  \lstick{$q_4$}
    & \gate{\mathbf{R}_Y} & \gate{\mathbf{R}_Z} & \gate{\mathbf{R}_X}
    & \qw      & \qw      & \targ{}  & \ctrl{-3}
    & \gate{\mathbf{R}_Y} & \gate{\mathbf{R}_Z} & \gate{\mathbf{R}_X}
    & \qw      & \qw      & \targ{}  & \ctrl{-3}& \qw
\end{quantikz}
\caption{Standard HEA circuit $\mathbf{U}_2^{\mathrm{HEA}}$ for $n=4$, $D=2$.}
\label{fig:hea_circuit}
\end{figure}
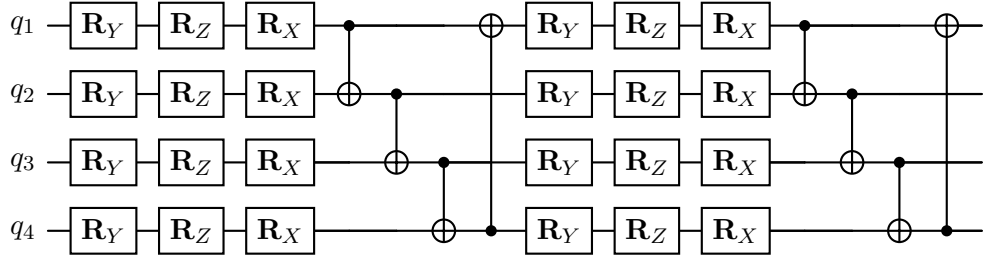
Because $\mathbf{U}_{\mathrm{ent}}$ is fixed and non-parameterized, its action
cannot be ``turned off''.  Adding a new block to the HEA
\emph{always} introduces an additional entangling layer, which
generically disrupts the state prepared by the preceding blocks.
This is the source of \emph{initialization shock} in na\"ive PDT.
\subsection{Identity-pair expansion and CNOT cancellation}
\label{sec:identity_pair}
To grow the circuit without disrupting the trained state, we append
not one but \emph{two} blocks that together compose to the identity.
The first is a standard forward block; the second is its
\emph{reverse}---the CNOT ring applied in reverse order followed by
inverse rotations:
\begin{equation}\label{eq:reverse_block}
  \mathbf{B}_d^{\mathrm{rev}}(\bm\theta^{(d)})
  = \mathcal{R}^{\mathrm{inv}}(\bm\theta^{(d)}) 
    \mathbf{U}_{\mathrm{ent}}^{-1}.
\end{equation}
\begin{proposition}[Identity-pair cancellation]
\label{prop:ip_cancel}
  For any $\bm\theta\in\bbR^{3n}$, the forward block followed by the
  reverse block composes to the identity:
  \begin{equation}\label{eq:ip_cancel}
    \mathbf{B}^{\mathrm{rev}}(\bm\theta) \mathbf{B}(\bm\theta)
    = \mathcal{R}^{\mathrm{inv}}(\bm\theta) 
      \underbrace{\mathbf{U}_{\mathrm{ent}}^{-1} \mathbf{U}_{\mathrm{ent}}}_{= \mathbf{I}_{2^n}}
      \mathcal{R}^{\mathrm{fwd}}(\bm\theta)
    = \underbrace{\mathcal{R}^{\mathrm{inv}}(\bm\theta) 
      \mathcal{R}^{\mathrm{fwd}}(\bm\theta)}_{= \mathbf{I}_{2^n}}
    = \mathbf{I}_{2^n}.
  \end{equation}
\end{proposition}
\begin{proof}
Recall the forward block $\mathbf{B}(\bm\theta)
=\mathbf{U}_{\mathrm{ent}} \mathcal{R}^{\mathrm{fwd}}(\bm\theta)$
from~\eqref{eq:hea_block} and the reverse block $\mathbf{B}^{\mathrm{rev}}(\bm\theta)
=\mathcal{R}^{\mathrm{inv}}(\bm\theta) \mathbf{U}_{\mathrm{ent}}^{-1}$
from~\eqref{eq:reverse_block}. Composing them and grouping the two
central entangler factors,
\begin{equation}\label{eq:ip_cancel_step1}
  \mathbf{B}^{\mathrm{rev}}(\bm\theta) \mathbf{B}(\bm\theta)
  = \mathcal{R}^{\mathrm{inv}}(\bm\theta) 
    \bigl[\mathbf{U}_{\mathrm{ent}}^{-1} \mathbf{U}_{\mathrm{ent}}\bigr] 
    \mathcal{R}^{\mathrm{fwd}}(\bm\theta).
\end{equation}
By the CNOT cancellation~\eqref{eq:cnot_cancel}, the bracketed factor is
$\mathbf{U}_{\mathrm{ent}}^{-1}\mathbf{U}_{\mathrm{ent}}=\mathbf{I}_{2^n}$, so
\begin{equation}\label{eq:ip_cancel_step2}
  \mathbf{B}^{\mathrm{rev}}(\bm\theta) \mathbf{B}(\bm\theta)
  = \mathcal{R}^{\mathrm{inv}}(\bm\theta) 
    \mathcal{R}^{\mathrm{fwd}}(\bm\theta).
\end{equation}
It remains to show $\mathcal{R}^{\mathrm{inv}}(\bm\theta) 
\mathcal{R}^{\mathrm{fwd}}(\bm\theta)=\mathbf{I}_{2^n}$. Writing the layers in
the tensor-product form~\eqref{eq:rot_fwd_layer}, namely
$\mathcal{R}^{\mathrm{fwd}}(\bm\theta)=\bigotimes_{q=1}^n
\mathbf{R}^{\mathrm{fwd}}_q$ and
$\mathcal{R}^{\mathrm{inv}}(\bm\theta)=\bigotimes_{q=1}^n
\mathbf{R}^{\mathrm{inv}}_q$ (where $\mathbf{R}^{\mathrm{fwd}}_q$ and
$\mathbf{R}^{\mathrm{inv}}_q$ abbreviate the forward and inverse single-qubit
rotations on qubit~$q$), the tensor-product multiplication
rule~\eqref{eq:tensor_mult} gives
\begin{equation}\label{eq:ip_cancel_step3}
  \mathcal{R}^{\mathrm{inv}}(\bm\theta) 
  \mathcal{R}^{\mathrm{fwd}}(\bm\theta)
  = \Bigl(\bigotimes_{q=1}^n \mathbf{R}^{\mathrm{inv}}_q\Bigr)
    \Bigl(\bigotimes_{q=1}^n \mathbf{R}^{\mathrm{fwd}}_q\Bigr)
  = \bigotimes_{q=1}^n
    \bigl(\mathbf{R}^{\mathrm{inv}}_q \mathbf{R}^{\mathrm{fwd}}_q\bigr).
\end{equation}
\Cref{lem:single_qubit_inv} establishes
$\mathbf{R}^{\mathrm{fwd}}_q \mathbf{R}^{\mathrm{inv}}_q=\mathbf{I}_2$; taking the inverse of
both sides (equivalently, applying the lemma with the roles of the two
unitaries exchanged, valid since each $\mathbf{R}^{\mathrm{fwd}}_q\in\SU(2)$ is
invertible with inverse $\mathbf{R}^{\mathrm{inv}}_q$) yields
$\mathbf{R}^{\mathrm{inv}}_q \mathbf{R}^{\mathrm{fwd}}_q=\mathbf{I}_2$ for every $q$. Hence
the right-hand side of~\eqref{eq:ip_cancel_step3} equals
$\bigotimes_{q=1}^n\mathbf{I}_2=\mathbf{I}_{2^n}$, and combining
with~\eqref{eq:ip_cancel_step2} proves the chain~\eqref{eq:ip_cancel}.
\end{proof}
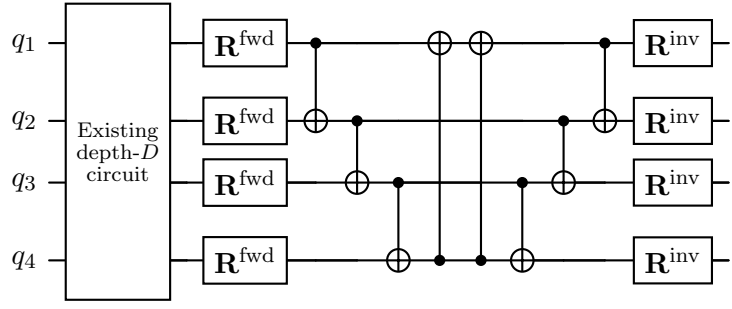
\begin{figure}[t]
\centering
\begin{quantikz}[row sep=0.2cm, column sep=0.22cm]
  \lstick{$q_1$}
    & \gate[4,nwires={2}]{\substack{\text{Existing}\\\text{depth-}D\\
      \text{circuit}}} & \qw
    & \gate{\mathbf{R}^{\mathrm{fwd}}}
    & \ctrl{1} & \qw      & \qw      & \targ{}
    & \targ{}  & \qw      & \qw      & \ctrl{1}
    & \gate{\mathbf{R}^{\mathrm{inv}}} & \qw \\
  \lstick{$q_2$}
    &          & \qw
    & \gate{\mathbf{R}^{\mathrm{fwd}}}
    & \targ{}  & \ctrl{1} & \qw      & \qw
    & \qw      & \qw      & \ctrl{1} & \targ{}
    & \gate{\mathbf{R}^{\mathrm{inv}}} & \qw \\
  \lstick{$q_3$}
    &          & \qw
    & \gate{\mathbf{R}^{\mathrm{fwd}}}
    & \qw      & \targ{}  & \ctrl{1} & \qw
    & \qw      & \ctrl{1} & \targ{}  & \qw
    & \gate{\mathbf{R}^{\mathrm{inv}}} & \qw \\
  \lstick{$q_4$}
    &          & \qw
    & \gate{\mathbf{R}^{\mathrm{fwd}}}
    & \qw      & \qw      & \targ{}  & \ctrl{-3}
    & \ctrl{-3}& \targ{}  & \qw      & \qw
    & \gate{\mathbf{R}^{\mathrm{inv}}} & \qw
\end{quantikz}
\caption{Identity-pair expansion from depth~$D$ to depth~$D{+}2$ for
$n=4$.}
\label{fig:identity_pair}
\end{figure}

\medskip\noindent\textbf{CNOT cancellation.}
Since the reverse CNOT ring immediately follows the forward CNOT ring,
these $2n$ CNOT gates cancel exactly and \emph{need not be physically
implemented}.  After cancellation, the two new blocks reduce to:
\[
  \mathcal{R}^{\mathrm{inv}}(\bm\theta^{(d+1)}) 
  \mathcal{R}^{\mathrm{fwd}}(\bm\theta^{(d)}),
\]
which is a product of two tensor-product rotation layers---purely local
gates with no entanglement.
\subsection{Effective single-entangler circuit (IP-PDT architecture)}
\label{sec:circuit}
After all CNOT cancellations, the IP-PDT circuit at ``depth''~$D$
(meaning $D$ rotation layers) contains the CNOT ring only at layer~0.
All subsequent layers are purely local rotations:
\begin{equation}\label{eq:UD}
  \mathbf{U}_D^{\mathrm{SE}}(\bm\theta)
  = \left(\prod_{k=D-1}^{1}\mathbf{W}_k(\bm\theta^{(k)})\right)
    \cdot \mathbf{U}_{\mathrm{ent}}
    \cdot\mathcal{R}^{\mathrm{fwd}}(\bm\theta^{(0)}),
\end{equation}
where
\begin{equation}\label{eq:Wk}
  \mathbf{W}_k(\bm\theta^{(k)})=
  \begin{cases}
    \mathcal{R}^{\mathrm{inv}}(\bm\theta^{(k)}), & k\text{ odd},\\
    \mathcal{R}^{\mathrm{fwd}}(\bm\theta^{(k)}), & k\ge 2\text{ even}.
  \end{cases}
\end{equation}
We call~\eqref{eq:UD} the \emph{single-entangler} (SE) circuit.
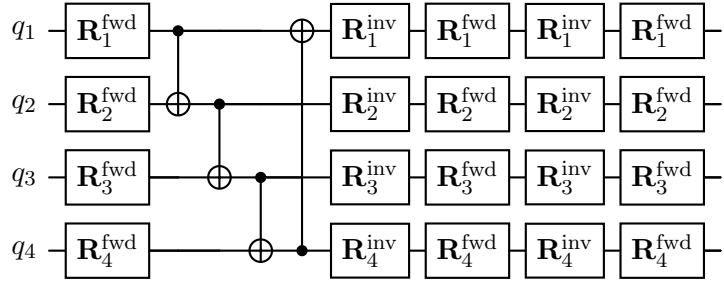
\begin{figure}[t]
\centering
\begin{quantikz}[row sep=0.25cm, column sep=0.22cm]
  \lstick{$q_1$}
    & \gate{\mathbf{R}^{\mathrm{fwd}}_1} & \ctrl{1} & \qw      & \qw
    & \targ{}
    & \gate{\mathbf{R}^{\mathrm{inv}}_1} & \gate{\mathbf{R}^{\mathrm{fwd}}_1}
    & \gate{\mathbf{R}^{\mathrm{inv}}_1} & \gate{\mathbf{R}^{\mathrm{fwd}}_1} & \qw \\
  \lstick{$q_2$}
    & \gate{\mathbf{R}^{\mathrm{fwd}}_2} & \targ{}  & \ctrl{1} & \qw
    & \qw
    & \gate{\mathbf{R}^{\mathrm{inv}}_2} & \gate{\mathbf{R}^{\mathrm{fwd}}_2}
    & \gate{\mathbf{R}^{\mathrm{inv}}_2} & \gate{\mathbf{R}^{\mathrm{fwd}}_2} & \qw \\
  \lstick{$q_3$}
    & \gate{\mathbf{R}^{\mathrm{fwd}}_3} & \qw      & \targ{}  & \ctrl{1}
    & \qw
    & \gate{\mathbf{R}^{\mathrm{inv}}_3} & \gate{\mathbf{R}^{\mathrm{fwd}}_3}
    & \gate{\mathbf{R}^{\mathrm{inv}}_3} & \gate{\mathbf{R}^{\mathrm{fwd}}_3} & \qw \\
  \lstick{$q_4$}
    & \gate{\mathbf{R}^{\mathrm{fwd}}_4} & \qw      & \qw      & \targ{}
    & \ctrl{-3}
    & \gate{\mathbf{R}^{\mathrm{inv}}_4} & \gate{\mathbf{R}^{\mathrm{fwd}}_4}
    & \gate{\mathbf{R}^{\mathrm{inv}}_4} & \gate{\mathbf{R}^{\mathrm{fwd}}_4} & \qw
\end{quantikz}
\caption{Effective single-entangler (SE) circuit at depth $D=5$ for
$n=4$ qubits, \emph{after} all CNOT cancellations.}
\label{fig:effective_circuit}
\end{figure}
\begin{definition}[Identity pairing]\label{def:id_pair}
  Layers $k$ (odd) and $k+1$ (even) form an \emph{identity pair} if
  $\bm\theta^{(k)}=\bm\theta^{(k+1)}$, so that
  \begin{equation}\label{eq:id_pair}
    \mathbf{W}_{k+1}(\bm\theta^{(k+1)}) \mathbf{W}_k(\bm\theta^{(k)})
    =\mathcal{R}^{\mathrm{fwd}}(\bm\theta^{(k)}) 
     \mathcal{R}^{\mathrm{inv}}(\bm\theta^{(k)})
    =\mathbf{I}_{2^n}.
  \end{equation}
\end{definition}
\begin{proposition}[Rotation-layer inverse]\label{prop:block_inv}
  For any $\bm\theta^{(d)}\in\bbR^{3n}$,
  $\mathcal{R}^{\mathrm{fwd}}(\bm\theta^{(d)}) 
  \mathcal{R}^{\mathrm{inv}}(\bm\theta^{(d)})=\mathbf{I}_{2^n}$.
\end{proposition}
\begin{proof}
By the definition~\eqref{eq:rot_fwd_layer} of the forward rotation
layer and of $\mathcal{R}^{\mathrm{inv}}$ as its adjoint, both layers are
tensor products over the $n$ qubits:
$\mathcal{R}^{\mathrm{fwd}}(\bm\theta^{(d)})=\bigotimes_{q=1}^n
\mathbf{R}^{\mathrm{fwd}}_q$ and
$\mathcal{R}^{\mathrm{inv}}(\bm\theta^{(d)})=\bigotimes_{q=1}^n
\mathbf{R}^{\mathrm{inv}}_q$, where $\mathbf{R}^{\mathrm{fwd}}_q$ and
$\mathbf{R}^{\mathrm{inv}}_q$ denote the forward and inverse single-qubit
rotations~\eqref{eq:fwd_rot}--\eqref{eq:inv_rot} acting on qubit~$q$
with that qubit's three angles from $\bm\theta^{(d)}$. The
tensor-product multiplication rule~\eqref{eq:tensor_mult} then gives
\begin{equation}\label{eq:block_inv_chain}
  \mathcal{R}^{\mathrm{fwd}}(\bm\theta^{(d)}) 
  \mathcal{R}^{\mathrm{inv}}(\bm\theta^{(d)})
  = \Bigl(\bigotimes_{q=1}^{n}\mathbf{R}^{\mathrm{fwd}}_q\Bigr)
    \Bigl(\bigotimes_{q=1}^{n}\mathbf{R}^{\mathrm{inv}}_q\Bigr)
  = \bigotimes_{q=1}^{n}
    \bigl(\mathbf{R}^{\mathrm{fwd}}_q \mathbf{R}^{\mathrm{inv}}_q\bigr)
  = \bigotimes_{q=1}^{n}\mathbf{I}_2
  = \mathbf{I}_{2^n},
\end{equation}
where the third equality applies \Cref{lem:single_qubit_inv} to each
factor (each $\mathbf{R}^{\mathrm{fwd}}_q \mathbf{R}^{\mathrm{inv}}_q=\mathbf{I}_2$,
independently of the qubit's angles), and the last uses
$\mathbf{I}_2^{\otimes n}=\mathbf{I}_{2^n}$.
\end{proof}

\subsection{Identity-paired initialization}\label{sec:identity_init}
When expanding the SE circuit from depth $D$ to $D+2$, we append layers
$D$ and $D+1$ initialized as an identity pair by copying a reference
parameter vector $\bm\theta_{\mathrm{ref}}$.  This guarantees:
\begin{equation}\label{eq:energy_preserve_exact}
  E_{D+2}\bigl(\hat{\bm\theta}^{(0:D-1)},
               \bm\theta_{\mathrm{ref}},
               \bm\theta_{\mathrm{ref}}\bigr)
  = E_D\bigl(\hat{\bm\theta}^{(0:D-1)}\bigr),
\end{equation}
where $\hat{\bm\theta}^{(0:D-1)}$ are the optimized parameters from
the previous stage.  We call this \emph{zero-shock expansion}.
\begin{definition}[Symmetry-breaking jitter]\label{def:jitter}
  To break the symmetry at the identity, we perturb the new
  parameters:
  \begin{equation}\label{eq:jitter}
    \bm\theta^{(D)}=\bm\theta_{\mathrm{ref}}+\bm\xi^{(1)},
    \quad
    \bm\theta^{(D+1)}=\bm\theta_{\mathrm{ref}}+\bm\xi^{(2)},
    \quad
    \bm\xi^{(i)}\sim\mathcal{N}(\bm 0,\sigma^2 \mathbf{I}_{3n}),
  \end{equation}
  with $\sigma\ll 1$.
\end{definition}
\subsection{Training protocols}\label{sec:protocols}
We compare three training protocols, and the comparison is set up to be
deliberately fair: in practice, in our experiments, all three end with the same number of trainable
parameters ($3nD_\star$ rotation angles, with $D_\star=5$) and all three
spend the same total optimization budget ($3T$ gradient-descent steps).
Our main study uses $n=6$ qubits; the scaling study of
\Cref{sec:exp_n16} repeats the comparison at $n=16$.  Because the
parameter count and the step budget are held fixed, any difference in
the final energy must come from the two ingredients we actually vary:
the \emph{circuit architecture}---how many entangling (CNOT) layers the
circuit physically contains---and the \emph{growth schedule}---whether
the circuit is trained at full depth from the outset or grown one stage
at a time, and how each newly appended layer is initialized.  To rule
out luck in the starting point, all three protocols seed their first
block from the \emph{same} random draw
$\bm\theta_0^{(0)}\sim\mathrm{Unif}([0,2\pi]^{3n})$.
\Cref{tab:protocol_summary} places the three side by side; we describe
each in turn.

\medskip\noindent\textbf{(a) Baseline: full-depth HEA, trained from
scratch.}
This is the conventional recipe: build the full depth-$D_\star$
hardware-efficient circuit~\eqref{eq:hea_circuit} up front and optimize
all of its angles simultaneously.  Every block starts as a copy of the
shared first block $\bm\theta_0^{(0)}$, and all $3nD_\star$ parameters
descend together,
\begin{equation}\label{eq:baseline_opt}
  \bm\theta^{(t+1)}_{\mathrm{base}}
  =\bm\theta^{(t)}_{\mathrm{base}}
   -\eta \nabla E_{D_\star}^{\mathrm{HEA}}
    \bigl(\bm\theta^{(t)}_{\mathrm{base}}\bigr),
  \qquad t=0,\ldots,3T{-}1,
\end{equation}
for the entire budget.  This circuit is the most expressive of the
three. 

\medskip\noindent\textbf{(b) Na\"ive progressive depth training.}
Na\"ive PDT keeps the hardware-efficient architecture but grows it in
stages instead of all at once.  
E.g., assuming the $1\to3\to5$ schedule, it trains a single block for $T$ steps,
appends two freshly randomized HEA blocks and trains for another $T$
steps at depth~$3$, then appends two more and trains a final $T$ steps
at depth~$5$. 
The motivation is sound---warm
start each deeper circuit from the optimum of the shallower one---but
the hardware-efficient architecture ``sabotages'' it.  Each appended block
carries its own CNOT ring that \emph{cannot} be switched off, so every
expansion injects entanglement; in practice, we observe that the
carefully trained state is scrambled and the energy jumps upward.  This
is the \emph{initialization shock} we document in \Cref{sec:exp_shock}.

\medskip\noindent\textbf{(c) Identity-paired PDT (this work).}
IP-PDT grows depth on the single-entangler circuit~\eqref{eq:UD}, where
only the first block ever contains a CNOT ring and every later layer is
a purely local rotation.  It uses the same schedule as
Na\"ive PDT, but the \emph{way} it appends layers is what eliminates the
shock.  At each expansion it adds the two new rotation layers as an
\emph{identity pair}---initialized so that they undo one another---so
that, \emph{before any further training}, the deeper circuit prepares
exactly the state the shallower circuit had already reached.  Concretely,
expanding sets:
\begin{equation}\label{eq:ippdt_expand1}
  \bm\theta_3^{(0)}
  =\Bigl(\hat{\bm\theta}_1,
    \underbrace{\hat{\bm\theta}_1+\bm\epsilon}_{\text{inverse layer}},
    \underbrace{\hat{\bm\theta}_1+\bm\epsilon'}_{\text{forward layer}}\Bigr)\in\bbR^{9n},
\end{equation}
where the small jitter
$\bm\epsilon,\bm\epsilon'\sim\mathcal{N}(0,\sigma^2 \mathbf{I})$, $\sigma\ll 1$,
breaks the exact symmetry just enough to hand the optimizer a usable
descent direction (\Cref{def:jitter}).  At zero jitter the appended pair
is \emph{exactly} the identity, so the energy carries across the
expansion unchanged,
\begin{equation}\label{eq:energy_preserve}
  E_3\bigl(\bm\theta_3^{(0)}\bigr)\big|_{\bm\epsilon=\bm\epsilon'=\bm 0}
  =E_1\bigl(\hat{\bm\theta}_1\bigr),
\end{equation}
and \emph{no new CNOT gate is introduced}.  IP-PDT thus climbs to
the target depth through a sequence of energy-preserving expansions
while paying for a single CNOT ring throughout.  \Cref{alg:ippdt} states
the procedure for an arbitrary stage schedule
$\mathcal{D}=\{D_1,\dots,D_M\}$.
\begin{table}[t]
\centering\small
\caption{Comparison of the three training protocols at final depth.
$D_\star=5$, $n=6$.}
\label{tab:protocol_summary}
\begin{tabular}{lccc}
\toprule
 & \textbf{Baseline} & \textbf{Na\"ive PDT} & \textbf{IP-PDT} \\
\midrule
Circuit architecture & HEA & HEA & Single-entangler \\
Growth schedule   & None & $1\to 3\to 5$ & $1\to 3\to 5$ \\
New-block init    & Random & Random (shock) & Identity pair \\
CNOT rings at final depth & $D_\star=5$ & $D_\star=5$ & $1$ \\
Rotation params   & $3nD_\star=90$ & $3nD_\star=90$ & $3nD_\star=90$ \\
Energy at expansion & N/A & Unpredictable & $E_D^\star+O(\sigma)$ \\
Total grad.\ steps & $3T$ & $3T$ & $3T$ \\
\bottomrule
\end{tabular}
\end{table}
\begin{algorithm}[t]
\caption{Identity-Paired Progressive Depth Training
(IP-PDT)}\label{alg:ippdt}
\begin{algorithmic}[1]
\Require Target depth $D_\star$; stage depths
  $\mathcal{D}=\{D_1,D_2,\ldots,D_M\}$ with $D_M=D_\star$;
  budget $T$ per stage; jitter scale $\sigma$; learning rate $\eta$.
\State Sample $\bm\theta_{\mathrm{ref}}\sim
  \mathrm{Unif}([0,2\pi]^{3n})$.
\State Initialize $\bm\theta\gets(\bm\theta_{\mathrm{ref}})$
  at depth $D_1$.
\For{$m=1$ to $M$}
  \For{$t=1$ to $T$}
    \State $\bm\theta\gets\bm\theta
      -\eta \nabla_{\bm\theta}E_{D_m}(\bm\theta)$.
  \EndFor
  \If{$m<M$}
    \State \textbf{Expand:} Append two rotation layers initialized as
      identity pair with jitter (\Cref{def:jitter}).
  \EndIf
\EndFor
\State\Return $\bm\theta$.
\end{algorithmic}
\end{algorithm}
\subsection{Resource estimation}\label{sec:resource}
Comparing the three protocols purely by ``number of optimization steps''
is misleading, because one step does different amounts of work
depending on which circuit is being differentiated.  We therefore
account for \emph{gate-weighted} cost, which captures both the classical
backpropagation effort in our simulations and---more importantly for a
hardware deployment---the number of noisy two-qubit gates that must be
physically executed.

Each forward (or backward) pass through a depth-$d$ circuit applies one
rotation gate per angle and one CNOT per entangling edge.  The
hardware-efficient circuit places a CNOT ring in \emph{every} layer, so
at depth $d$ it uses $3nd$ single-qubit rotations and $nd$ CNOTs---$4nd$
gates in all.  The single-entangler circuit uses the same $3nd$
rotations, but because every CNOT ring after the first cancels by
construction it carries only $n$ CNOTs \emph{regardless of depth}, for
$n(3d+1)$ gates.  \Cref{tab:gate_counts} collects these counts.  The key
asymmetry is already visible: as the circuit deepens, the
hardware-efficient cost grows linearly in $d$ through \emph{both} terms,
whereas the single-entangler entangling cost stays constant at $n$.
\begin{table}[h]
\centering\small
\caption{Gate counts per circuit evaluation at depth~$d$.}
\label{tab:gate_counts}
\begin{tabular}{lccc}
\toprule
& \textbf{Baseline/Na\"ive} (HEA) & \textbf{IP-PDT} (SE) \\
\midrule
Single-qubit rotation gates & $3nd$ & $3nd$ \\
CNOT gates                  & $nd$  & $n$   \\
Total gates                 & $4nd$ & $n(3d+1)$ \\
\bottomrule
\end{tabular}
\end{table}

\medskip\noindent\textbf{Cumulative training cost.}
The protocols also differ in \emph{when} they spend their steps.  The
baseline runs all $3T$ steps on the full depth circuit; the
progressive methods spend their first $T$ steps on a cheap depth-$1$
circuit and only reach full depth in the final stage.  Summing the
per-step gate cost over the three stages (here, depths $\{1,3,5\}$, $T$ steps
each, based on our experiments) gives the totals:
\begin{align}
  C_{\mathrm{total}}^{\mathrm{base}}
  &=3T\times O(4n\times 5)=60  nT, \label{eq:cost_base}\\
  C_{\mathrm{total}}^{\mathrm{naive}}
  &=T\times O(4n)+T\times O(12n)+T\times O(20n)=36  nT, \label{eq:cost_naive}\\
  C_{\mathrm{total}}^{\mathrm{IP}}
  &=T\times O(4n)+T\times O(10n)+T\times O(16n)=30  nT. \label{eq:cost_ip}
\end{align}
The gap widens with the
target depth $D_\star$: the baseline pays for full depth on every step,
whereas IP-PDT's entangling cost does not grow with depth at all.  On
hardware, where two-qubit gates dominate both the error budget and the
wall-clock, the fivefold reduction in CNOT count (\Cref{tab:gate_counts})
is the more consequential of the two savings.
\section{Theoretical Analysis}\label{sec:theory}
Our experiments will show that progressive depth training keeps helping
even after the circuit can no longer represent any new states.  Our analysis answers four questions in sequence:

\emph{(i) What can the circuit represent, and when does that stop
growing?}  We first characterize the reachable set of the
single-entangler ansatz and show that it \emph{saturates} right after
the first post-entangler layer (\Cref{thm:saturation}): beyond depth
two, additional layers add parameters but \emph{no new reachable
states}.  This is exactly what turns the paper's central question---why
does training keep improving?---into a question about \emph{optimization}
rather than expressibility.

\emph{(ii) Why doesn't growing the circuit destroy the solution we
already found?}  Saturation alone does not guarantee that a good
depth-$D$ solution survives the jump to depth $D{+}2$.  We recast each
expansion as a \emph{homotopy} that smoothly switches the new layers on,
and use the implicit function theorem to track the minimizer as it
deforms (\Cref{thm:continuation}): a nondegenerate minimum moves by only
$O(\sigma)$ under the jittered identity-pair expansion.

\emph{(iii) Can we guarantee that training never moves backwards?}  The
continuation argument is local and presumes we actually find the
continued minimum.  We complement it with an \emph{unconditional}
guarantee---under the acceptance rule the energy produced at successive
stages is monotonically non-increasing (\Cref{thm:monotone}), with no
assumption on the landscape---and we show that gradient descent on the
expanded circuit contracts at a geometric rate, inherited up to
$O(\sigma)$ from the previous stage (\Cref{thm:basin}).

\emph{(iv) Does low energy actually mean we found the ground state?}
Finally we connect the quantity we optimize (energy) to the quantity we
ultimately care about (ground-state fidelity) through a spectral-gap
sandwich (\Cref{thm:gap_fidelity}), and we are explicit about when this
certificate is informative and when the expressibility ceiling makes it
vacuous.

Together the four results tell one story: \emph{the single-entangler
construction freezes expressibility early~(i) so that added depth can be
spent safely~(ii) and monotonically~(iii) on optimization, with a
certificate that ties the optimized energy back to fidelity~(iv).}  Complete proofs are deferred
to \Cref{app:proofs,app:theory_extended}.
\begin{lemma}[State factorization]\label{lem:factor}
  For any depth $D\ge 1$, every state
  $|\boldsymbol{\psi}\rangle\in\cM_D$ can be written as
  \begin{equation}\label{eq:state_form}
    |\boldsymbol{\psi}\rangle
    = \Bigl(\bigotimes_{q=1}^n \mathbf{V}_q\Bigr) 
      \mathbf{U}_{\mathrm{ent}} 
      \Bigl(\bigotimes_{q=1}^n r_q\Bigr)|0^n\rangle,
  \end{equation}
  where $r_q\in\SU(2)$ is the pre-entangler rotation on qubit~$q$, and
  $\mathbf{V}_q = \mathbf{W}_{D{-}1,q}\cdots \mathbf{W}_{1,q}\in\SU(2)$ is the composition of all
  post-entangler rotation layers on qubit~$q$.
\end{lemma}
\begin{proof}
Fix $D\ge 1$ and a parameter vector $\bm\theta$, and let
$|\boldsymbol{\psi}\rangle=\mathbf{U}_D^{\mathrm{SE}}(\bm\theta)|0^n\rangle\in\cM_D$
be the corresponding reachable state. The SE circuit
definition~\eqref{eq:UD} factors as
\begin{equation}\label{eq:factor_split}
  \mathbf{U}_D^{\mathrm{SE}}(\bm\theta)
  = \underbrace{\Bigl(\prod_{k=D-1}^{1}\mathbf{W}_k(\bm\theta^{(k)})\Bigr)}_{
      \text{post-entangler}}
    \cdot \mathbf{U}_{\mathrm{ent}}\cdot
    \underbrace{\mathcal{R}^{\mathrm{fwd}}(\bm\theta^{(0)})}_{
      \text{pre-entangler}},
\end{equation}
where the empty product for $D=1$ is read as $\mathbf{I}_{2^n}$.

\emph{Pre-entangler factor.} By~\eqref{eq:rot_fwd_layer},
$\mathcal{R}^{\mathrm{fwd}}(\bm\theta^{(0)})=\bigotimes_{q=1}^n r_q$ with
$r_q=\mathbf{R}^{\mathrm{fwd}}(\alpha_q^{(0)},\beta_q^{(0)},\gamma_q^{(0)})$.
Each $r_q$ is a product of three Pauli rotations~\eqref{eq:pauli_rot};
every $\mathbf{R}_\sigma(\theta)=\exp(-i\tfrac{\theta}{2}\sigma)$ is unitary
with unit determinant (since $\sigma$ is traceless,
$\det\mathbf{R}_\sigma(\theta)=e^{-i\frac{\theta}{2}\operatorname{tr}\sigma}=1$),
so $r_q\in\SU(2)$ as $\SU(2)$ is closed under matrix multiplication.

\emph{Post-entangler factor.} For each $k\in\{1,\dots,D-1\}$, the
definition~\eqref{eq:Wk} sets $\mathbf{W}_k$ equal to either
$\mathcal{R}^{\mathrm{inv}}(\bm\theta^{(k)})$ or
$\mathcal{R}^{\mathrm{fwd}}(\bm\theta^{(k)})$; in both cases
\eqref{eq:rot_fwd_layer} (and its adjoint) expresses it as a tensor
product $\mathbf{W}_k=\bigotimes_{q=1}^n \mathbf{W}_{k,q}$ with each single-qubit
factor $\mathbf{W}_{k,q}\in\SU(2)$ by the same determinant computation. We
collapse the ordered product over $k$ qubit-by-qubit by iterating the
tensor-product multiplication rule~\eqref{eq:tensor_mult}:
\begin{equation}\label{eq:factor_post}
  \prod_{k=D-1}^{1}\mathbf{W}_k(\bm\theta^{(k)})
  = \prod_{k=D-1}^{1}\Bigl(\bigotimes_{q=1}^n \mathbf{W}_{k,q}\Bigr)
  = \bigotimes_{q=1}^n
    \underbrace{\bigl(\mathbf{W}_{D-1,q}\cdots\mathbf{W}_{1,q}\bigr)}_{=:~\mathbf{V}_q}.
\end{equation}
Here the second equality follows because~\eqref{eq:tensor_mult} merges
two qubit-aligned tensor products into one, and a finite induction on
the number of factors extends this from two layers to all $D-1$ layers
(the order on each wire is preserved, matching
$\mathbf{V}_q=\mathbf{W}_{D-1,q}\cdots\mathbf{W}_{1,q}$ in the statement). Each $\mathbf{V}_q$ is a
product of $\SU(2)$ matrices, hence $\mathbf{V}_q\in\SU(2)$ by closure.

Substituting the two tensor-product factors into~\eqref{eq:factor_split}
gives
$\mathbf{U}_D^{\mathrm{SE}}(\bm\theta)
=\bigl(\bigotimes_{q=1}^n \mathbf{V}_q\bigr) \mathbf{U}_{\mathrm{ent}} 
\bigl(\bigotimes_{q=1}^n r_q\bigr)$, and applying it to $|0^n\rangle$
yields the claimed form~\eqref{eq:state_form} with $r_q,\mathbf{V}_q\in\SU(2)$.
\end{proof}
\begin{theorem}[Reachable Set Saturation]\label{thm:saturation}
  For the single-entangler circuit family~\eqref{eq:UD}:
  \begin{enumerate}[label=(\alph*)]
    \item $\cM_1\subseteq\cM_2$.
    \item For all $D\ge 2$, $\cM_D=\cM_2$.
    \item If $\mathbf{U}_{\mathrm{ent}}$ does not normalize the local-unitary
      group $\mathcal{L}=\{\bigotimes_q\mathbf{V}_q:\mathbf{V}_q\in\SU(2)\}$
      (as the CNOT ring does not), then $\cM_1\subsetneq\cM_2$.
  \end{enumerate}
\end{theorem}
\begin{proof}
Throughout, $\cM_D=\{\mathbf{U}_D^{\mathrm{SE}}(\bm\theta)|0^n\rangle:
\bm\theta\in\bbR^{3nD}\}$ is the set of states reachable by the
depth-$D$ SE circuit over all parameter choices.

\textit{Part~(a): $\cM_1\subseteq\cM_2$.}
Let $|\boldsymbol{\psi}\rangle\in\cM_1$, so
$|\boldsymbol{\psi}\rangle=\mathbf{U}_1^{\mathrm{SE}}(\bm\theta^{(0)})|0^n\rangle$ for
some $\bm\theta^{(0)}\in\bbR^{3n}$. By~\eqref{eq:UD} the depth-$1$
circuit has an empty post-entangler product, so
$\mathbf{U}_1^{\mathrm{SE}}(\bm\theta^{(0)})
=\mathbf{U}_{\mathrm{ent}} \mathcal{R}^{\mathrm{fwd}}(\bm\theta^{(0)})$.
Consider the depth-$2$ circuit with parameters
$(\bm\theta^{(0)},\bm 0)$. Its single post-entangler layer is
$\mathbf{W}_1(\bm 0)=\mathcal{R}^{\mathrm{inv}}(\bm 0)$ by~\eqref{eq:Wk}
($k=1$ is odd). Each single-qubit factor of
$\mathcal{R}^{\mathrm{inv}}(\bm 0)$ is
$\mathbf{R}^{\mathrm{inv}}(0,0,0)=\mathbf{R}_X(0)\mathbf{R}_Z(0)\mathbf{R}_Y(0)=\mathbf{I}_2$
by~\eqref{eq:inv_rot} and $\mathbf{R}_\sigma(0)=\mathbf{I}_2$, so
$\mathcal{R}^{\mathrm{inv}}(\bm 0)=\bigotimes_{q=1}^n\mathbf{I}_2=\mathbf{I}_{2^n}$.
Hence
\begin{equation}\label{eq:sat_a}
  \mathbf{U}_2^{\mathrm{SE}}(\bm\theta^{(0)},\bm 0)
  = \mathbf{I}_{2^n} \mathbf{U}_{\mathrm{ent}} 
    \mathcal{R}^{\mathrm{fwd}}(\bm\theta^{(0)})
  = \mathbf{U}_1^{\mathrm{SE}}(\bm\theta^{(0)}),
\end{equation}
so $|\boldsymbol{\psi}\rangle=\mathbf{U}_2^{\mathrm{SE}}(\bm\theta^{(0)},\bm 0)|0^n\rangle
\in\cM_2$. As $|\boldsymbol{\psi}\rangle$ was arbitrary, $\cM_1\subseteq\cM_2$.

\textit{Part~(b): $\cM_D=\cM_2$ for all $D\ge 2$.}
We prove the two inclusions separately.

\emph{($\cM_2\subseteq\cM_D$).} Fix $D\ge 2$ and let
$|\boldsymbol{\psi}\rangle\in\cM_2$ be produced by parameters
$(\bm\theta^{(0)},\bm\theta^{(1)})$. Define depth-$D$ parameters by
keeping $\bm\theta^{(0)},\bm\theta^{(1)}$ and setting
$\bm\theta^{(k)}=\bm 0$ for all $2\le k\le D-1$. For each such $k$,
$\mathbf{W}_k(\bm 0)$ equals $\mathcal{R}^{\mathrm{inv}}(\bm 0)$ or
$\mathcal{R}^{\mathrm{fwd}}(\bm 0)$ by~\eqref{eq:Wk}; either way every
single-qubit factor is $\mathbf{R}^{\mathrm{fwd}}(0,0,0)=\mathbf{I}_2$ or
$\mathbf{R}^{\mathrm{inv}}(0,0,0)=\mathbf{I}_2$, so $\mathbf{W}_k(\bm 0)=\mathbf{I}_{2^n}$. The
post-entangler product therefore collapses to
$\prod_{k=D-1}^{1}\mathbf{W}_k=\mathbf{W}_1(\bm\theta^{(1)})$, giving
$\mathbf{U}_D^{\mathrm{SE}}(\bm\theta^{(0)},\bm\theta^{(1)},\bm 0,\dots,\bm 0)
=\mathbf{U}_2^{\mathrm{SE}}(\bm\theta^{(0)},\bm\theta^{(1)})$ and hence
$|\boldsymbol{\psi}\rangle\in\cM_D$.

\emph{($\cM_D\subseteq\cM_2$).} Let $|\boldsymbol{\psi}\rangle\in\cM_D$ with
$D\ge 2$. By \Cref{lem:factor} it has the form~\eqref{eq:state_form},
\begin{equation}\label{eq:sat_b_form}
  |\boldsymbol{\psi}\rangle
  = \Bigl(\bigotimes_{q=1}^n \mathbf{V}_q\Bigr) \mathbf{U}_{\mathrm{ent}} 
    \Bigl(\bigotimes_{q=1}^n r_q\Bigr)|0^n\rangle,
  \qquad r_q,\mathbf{V}_q\in\SU(2),
\end{equation}
where $\mathbf{V}_q=\mathbf{W}_{D-1,q}\cdots\mathbf{W}_{1,q}$ is the composition of all
post-entangler single-qubit rotations on qubit~$q$ and lies in $\SU(2)$
by closure of $\SU(2)$ under multiplication. At depth $2$ the
post-entangler factor on qubit~$q$ is the single $\SU(2)$ element
$\mathbf{W}_{1,q}=\mathbf{R}^{\mathrm{inv}}(\alpha_q^{(1)},\beta_q^{(1)},\gamma_q^{(1)})$,
and by \Cref{lem:su2_cover} the map
$(\alpha,\beta,\gamma)\mapsto\mathbf{R}^{\mathrm{inv}}(\alpha,\beta,\gamma)$ is
surjective onto $\SU(2)$. Hence for each $q$ there exist angles
$(\alpha_q^{(1)},\beta_q^{(1)},\gamma_q^{(1)})$ with
$\mathbf{R}^{\mathrm{inv}}(\alpha_q^{(1)},\beta_q^{(1)},\gamma_q^{(1)})=\mathbf{V}_q$;
likewise, surjectivity of the forward map---the forward Euler
sequence of \Cref{lem:su2_cover}, Part~(a)---supplies
$\bm\theta^{(0)}$ with
$\mathcal{R}^{\mathrm{fwd}}(\bm\theta^{(0)})=\bigotimes_q r_q$. With this
choice $\mathbf{U}_2^{\mathrm{SE}}(\bm\theta^{(0)},\bm\theta^{(1)})|0^n\rangle$
equals the right-hand side of~\eqref{eq:sat_b_form}, so
$|\boldsymbol{\psi}\rangle\in\cM_2$. The crux is that one post-entangler layer
already attains \emph{every} per-qubit factor $\mathbf{V}_q\in\SU(2)$; the extra
layers present when $D>2$ compose additional $\SU(2)$ elements but,
by closure, cannot leave $\SU(2)$, so they enlarge neither the per-qubit
factor nor $\cM_D$. Combining both inclusions, $\cM_D=\cM_2$.

\textit{Part~(c): $\cM_1\subsetneq\cM_2$ when $\mathbf{U}_{\mathrm{ent}}$ does
not normalize $\mathcal{L}$.}
By Part~(a), $\cM_1\subseteq\cM_2$, so it suffices to exhibit one state
in $\cM_2\setminus\cM_1$. We first give an explicit witness for $n=2$
with $\mathbf{U}_{\mathrm{ent}}=\CNOT_{1\to2}$, the entangling gate from which
the CNOT ring~\eqref{eq:cnot_ring} is built, and then give the general
non-normalization argument. Write
$c_8=\cos(\pi/8)$, $s_8=\sin(\pi/8)$, and recall from~\eqref{eq:pauli_rot}
that $\mathbf{R}_Y(\theta)|0\rangle=\cos\tfrac\theta2|0\rangle
+\sin\tfrac\theta2|1\rangle$ and
$\mathbf{R}_Y(\theta)|1\rangle=-\sin\tfrac\theta2|0\rangle
+\cos\tfrac\theta2|1\rangle$.

\emph{The target depth-$2$ state.} Take the depth-$2$ parameters
realizing $r_1=\mathbf{R}_Y(\pi/2)$, $r_2=\mathbf{I}_2$, $\mathbf{V}_1=\mathbf{I}_2$,
$\mathbf{V}_2=\mathbf{R}_Y(\pi/4)$ in~\eqref{eq:sat_b_form} (each is in $\SU(2)$, hence
attainable by Part~(b)). Applying the gates right-to-left:
\begin{align}
  (\mathbf{R}_Y(\tfrac\pi2)\otimes\mathbf{I}_2)|00\rangle
  &= \bigl(\tfrac{1}{\sqrt2}|0\rangle+\tfrac{1}{\sqrt2}|1\rangle\bigr)
     \otimes|0\rangle
   = \tfrac{1}{\sqrt2}\bigl(|00\rangle+|10\rangle\bigr),
     \label{eq:satc_pre}\\
  \CNOT_{1\to2} \tfrac{1}{\sqrt2}\bigl(|00\rangle+|10\rangle\bigr)
  &= \tfrac{1}{\sqrt2}\bigl(|00\rangle+|11\rangle\bigr),
     \label{eq:satc_ent}\\
  (\mathbf{I}_2\otimes\mathbf{R}_Y(\tfrac\pi4)) 
     \tfrac{1}{\sqrt2}\bigl(|00\rangle+|11\rangle\bigr)
  &= \tfrac{1}{\sqrt2}\bigl(c_8|00\rangle+s_8|01\rangle
     -s_8|10\rangle+c_8|11\rangle\bigr)
   =: |\boldsymbol{\psi}^\star\rangle,
     \label{eq:satc_target}
\end{align}
where~\eqref{eq:satc_ent} uses
$\CNOT_{1\to2}|a\rangle|b\rangle=|a\rangle|a\oplus b\rangle$, and the
second qubit in~\eqref{eq:satc_target} was rotated by
$\mathbf{R}_Y(\pi/4)$ (half-angle $\pi/8$). Thus
$|\boldsymbol{\psi}^\star\rangle\in\cM_2$.

\emph{General form of a $\cM_1$ state.} By~\eqref{eq:UD} with $D=1$,
every $|\boldsymbol{\phi}\rangle\in\cM_1$ is
$\CNOT_{1\to2} \mathcal{R}^{\mathrm{fwd}}(\bm\theta^{(0)})|00\rangle$,
i.e.\ a CNOT applied to a product state. Writing the product state as
$(a|0\rangle+b|1\rangle)\otimes(c|0\rangle+d|1\rangle)$ with
$a,b,c,d\in\bbC$ (the most general single-qubit states; any global
phase is absorbed), expand and apply $\CNOT_{1\to2}$, which fixes
$|00\rangle,|01\rangle$ and swaps the control-$1$ pair
$|10\rangle\leftrightarrow|11\rangle$:
\begin{align}
  (a|0\rangle{+}b|1\rangle)\otimes(c|0\rangle{+}d|1\rangle)
  &= ac |00\rangle+ad |01\rangle+bc |10\rangle+bd |11\rangle,
     \label{eq:satc_prod}\\
  \CNOT_{1\to2}\bigl[ \cdot \bigr]
  &= ac |00\rangle+ad |01\rangle+bd |10\rangle+bc |11\rangle
   =: |\boldsymbol{\phi}\rangle.
     \label{eq:satc_phi}
\end{align}

\emph{The contradiction.} Suppose, for contradiction, that
$|\boldsymbol{\psi}^\star\rangle\in\cM_1$, i.e.\
$|\boldsymbol{\phi}\rangle=|\boldsymbol{\psi}^\star\rangle$ for some $a,b,c,d$. Matching
the four computational-basis amplitudes
of~\eqref{eq:satc_phi} and~\eqref{eq:satc_target} gives
\begin{equation}\label{eq:satc_amps}
  ac=\tfrac{c_8}{\sqrt2},\qquad
  ad=\tfrac{s_8}{\sqrt2},\qquad
  bd=-\tfrac{s_8}{\sqrt2},\qquad
  bc=\tfrac{c_8}{\sqrt2}.
\end{equation}
Since $c_8=\cos(\pi/8)>0$ and $s_8=\sin(\pi/8)>0$, the right-hand sides
$\tfrac{c_8}{\sqrt2}$ and $\pm\tfrac{s_8}{\sqrt2}$ are all nonzero;
hence none of $a,b,c,d$ can vanish (a zero factor would force two of
the products in~\eqref{eq:satc_amps} to be $0$). With $a,b,c,d\neq0$ we
may divide. The first two equations give
\begin{equation}\label{eq:satc_ratio1}
  \frac{ad}{ac}=\frac{d}{c}
  =\frac{s_8/\sqrt2}{c_8/\sqrt2}
  =\tan \Bigl(\frac\pi8\Bigr)>0,
\end{equation}
while the last two give
\begin{equation}\label{eq:satc_ratio2}
  \frac{bd}{bc}=\frac{d}{c}
  =\frac{-s_8/\sqrt2}{c_8/\sqrt2}
  =-\tan \Bigl(\frac\pi8\Bigr)<0.
\end{equation}
Equations~\eqref{eq:satc_ratio1} and~\eqref{eq:satc_ratio2} assign the
single ratio $d/c$ two values of opposite sign,
$+\tan(\pi/8)$ and $-\tan(\pi/8)$, which are distinct because
$\tan(\pi/8)\neq0$. This contradiction shows
$|\boldsymbol{\psi}^\star\rangle\notin\cM_1$, hence
$|\boldsymbol{\psi}^\star\rangle\in\cM_2\setminus\cM_1$ and
$\cM_1\subsetneq\cM_2$.

\emph{General argument (any non-normalizing entangler).}
By the depth-$1$ and depth-$2$ forms,
$\cM_1=\{\mathbf{U}_{\mathrm{ent}}\,|\pi\rangle:|\pi\rangle\ \text{a product state}\}$
and, using the arbitrary post-entangler product layer of Part~(b),
$\cM_2=\{\mathbf{L}\,m:\mathbf{L}\in\mathcal{L},\ m\in\cM_1\}=\mathcal{L}\,\cM_1$,
where $\mathcal{L}$ is the local-unitary group.  Hence $\cM_1=\cM_2$ iff
$\mathcal{L}\,\cM_1=\cM_1$, i.e.\ iff for every $\mathbf{L}\in\mathcal{L}$ and
product $|\pi\rangle$ the state
$\mathbf{U}_{\mathrm{ent}}^{\dagger}\mathbf{L}\,\mathbf{U}_{\mathrm{ent}}|\pi\rangle$ is again a
product state.  The unitaries preserving every product state are exactly
$\mathcal{L}\rtimes S_n$ (local unitaries composed with qubit permutations);
since $\mathcal{L}=\SU(2)^{\otimes n}$ is connected and the conjugation
$\mathbf{L}\mapsto\mathbf{U}_{\mathrm{ent}}^{\dagger}\mathbf{L}\,\mathbf{U}_{\mathrm{ent}}$ fixes
$\mathbf{I}$, continuity confines its image to the identity component
$\mathcal{L}$, so this holds precisely when $\mathbf{U}_{\mathrm{ent}}$
\emph{normalizes} $\mathcal{L}$
($\mathbf{U}_{\mathrm{ent}}^{\dagger}\mathcal{L}\,\mathbf{U}_{\mathrm{ent}}\subseteq\mathcal{L}$).
Contrapositively, if $\mathbf{U}_{\mathrm{ent}}$ does not normalize $\mathcal{L}$
then $\cM_1\subsetneq\cM_2$.  The CNOT ring does not normalize
$\mathcal{L}$: conjugating the local rotation
$\mathbf{R}_X(\alpha)\otimes\mathbf{I}\in\mathcal{L}$ by the entangler gives
$\CNOT_{1\to2}^{\dagger}(\mathbf{R}_X(\alpha)\otimes\mathbf{I})\CNOT_{1\to2}
=e^{-i\alpha\mathbf{X}_1\mathbf{X}_2/2}$, an entangling $XX$-rotation that for
generic $\alpha$ lies outside $\mathcal{L}$, and applying it to a generic
product state produces a member of $\cM_2\setminus\cM_1$.  (By contrast a SWAP
gate, though not a product unitary, \emph{does} normalize $\mathcal{L}$
via $\mathrm{SWAP}\,(\mathbf{u}\otimes\mathbf{v})\,\mathrm{SWAP}=\mathbf{v}\otimes\mathbf{u}$, and
indeed gives $\cM_1=\cM_2$; this is why the correct hypothesis is
non-normalization, not mere non-productness.)
\end{proof}

\subsection{Continuation-theoretic framework}
\label{sec:continuation}
\begin{definition}[Identity-embedding property (IEP)]
\label{def:iep}
  The layered family $\{\mathbf{U}_D^{\mathrm{SE}}\}_{D\ge 1}$ satisfies the
  IEP if, for every $D\ge 1$, there exists a map
  $\iota_D:\bbR^{3nD}\to\bbR^{3n(D+2)}$ such that
  $\mathbf{U}_{D+2}^{\mathrm{SE}}(\iota_D(\bm\theta))
  = \mathbf{U}_D^{\mathrm{SE}}(\bm\theta)$.
\end{definition}
The IEP is the foundational mechanism that makes progressive depth
training safe: it guarantees that the deeper circuit can reproduce
every state that the shallower circuit can produce.  We provide an
extended discussion of the IEP and its consequences in
\Cref{app:iep_extended}.
\begin{proposition}[Nested reachable sets]\label{prop:nested}
  Under the identity-pair expansion rule,
  $\cM_D\subseteq\cM_{D+2}$ for all $D\ge 1$.
\end{proposition}
\begin{proof}
Let $|\boldsymbol{\psi}\rangle\in\cM_D$, so
$|\boldsymbol{\psi}\rangle=\mathbf{U}_D^{\mathrm{SE}}(\bm\theta^\star)|0^n\rangle$ for
some $\bm\theta^\star\in\bbR^{3nD}$. We exhibit depth-$(D{+}2)$
parameters reaching the same state, which proves
$|\boldsymbol{\psi}\rangle\in\cM_{D+2}$.

Apply the identity-embedding map of \Cref{def:iep} in the concrete form
provided by the identity pairing (\Cref{def:id_pair}): set
\begin{equation}\label{eq:nested_iota}
  \iota_D(\bm\theta^\star)
  = \bigl(\bm\theta^\star, \bm\theta_{\mathrm{ref}}, 
          \bm\theta_{\mathrm{ref}}\bigr)
  = \bigl(\bm\theta^{\star(0)},\dots,\bm\theta^{\star(D-1)}, 
          \bm\theta_{\mathrm{ref}}, \bm\theta_{\mathrm{ref}}\bigr),
\end{equation}
i.e.\ append two new rotation layers $D$ and $D{+}1$, both with the same
reference parameters $\bm\theta_{\mathrm{ref}}$.  By~\eqref{eq:Wk} the two
appended layers carry opposite orientations according to the parity of
$D$---for the odd-$D$ schedule ($D\in\{1,3,5,\dots\}$), $\mathbf{W}_D
=\mathcal{R}^{\mathrm{inv}}(\bm\theta_{\mathrm{ref}})$ ($D$ odd) and
$\mathbf{W}_{D+1}=\mathcal{R}^{\mathrm{fwd}}(\bm\theta_{\mathrm{ref}})$
($D{+}1$ even), the identity pair of~\eqref{eq:id_pair}; the even-$D$
case merely swaps the two labels.  In either parity their composition in
the leftmost (last-applied) position of the depth-$(D{+}2)$
post-entangler product collapses,
\begin{equation}\label{eq:nested_cancel}
  \mathbf{W}_{D+1}(\bm\theta_{\mathrm{ref}}) 
  \mathbf{W}_D(\bm\theta_{\mathrm{ref}})
  = \mathbf{I}_{2^n},
\end{equation}
because $\mathcal{R}^{\mathrm{fwd}}(\bm\theta_{\mathrm{ref}})$ and
$\mathcal{R}^{\mathrm{inv}}(\bm\theta_{\mathrm{ref}})$ are mutual
two-sided inverses---each is the adjoint of the other, so
\Cref{prop:block_inv} gives both
$\mathcal{R}^{\mathrm{fwd}}\mathcal{R}^{\mathrm{inv}}=\mathbf{I}_{2^n}$ and
$\mathcal{R}^{\mathrm{inv}}\mathcal{R}^{\mathrm{fwd}}=\mathbf{I}_{2^n}$. Substituting the
product form~\eqref{eq:UD} for the depth-$(D{+}2)$ circuit and using
that the two new layers occupy the leftmost factors,
\begin{equation}\label{eq:nested_collapse}
  \mathbf{U}_{D+2}^{\mathrm{SE}}\bigl(\iota_D(\bm\theta^\star)\bigr)
  = \underbrace{\mathbf{W}_{D+1}(\bm\theta_{\mathrm{ref}}) 
      \mathbf{W}_D(\bm\theta_{\mathrm{ref}})}_{= \mathbf{I}_{2^n}
      \text{ by }\eqref{eq:nested_cancel}}
     \Bigl(\prod_{k=D-1}^{1}\mathbf{W}_k(\bm\theta^{\star(k)})\Bigr)
    \mathbf{U}_{\mathrm{ent}} 
    \mathcal{R}^{\mathrm{fwd}}(\bm\theta^{\star(0)})
  = \mathbf{U}_D^{\mathrm{SE}}(\bm\theta^\star).
\end{equation}
Applying~\eqref{eq:nested_collapse} to $|0^n\rangle$ gives
$\mathbf{U}_{D+2}^{\mathrm{SE}}(\iota_D(\bm\theta^\star))|0^n\rangle
=\mathbf{U}_D^{\mathrm{SE}}(\bm\theta^\star)|0^n\rangle=|\boldsymbol{\psi}\rangle$,
so $|\boldsymbol{\psi}\rangle\in\cM_{D+2}$. Since $|\boldsymbol{\psi}\rangle\in\cM_D$
was arbitrary, $\cM_D\subseteq\cM_{D+2}$ for every $D\ge 1$.
\end{proof}
\begin{corollary}[Monotonicity of optimal energy]
\label{cor:energy_monotone}
  $E_{D+2}^\star\le E_D^\star$ for all $D\ge 1$.
\end{corollary}
\begin{proof}
By~\eqref{eq:vqe_energy} the energy depends on the parameters only
through the prepared state
$|\boldsymbol{\psi}(\bm\theta)\rangle=\mathbf{U}_D^{\mathrm{SE}}(\bm\theta)|0^n\rangle$,
namely $E_D(\bm\theta)=\langle\boldsymbol{\psi}(\bm\theta)|\mathbf{H}|\boldsymbol{\psi}(\bm\theta)\rangle$.
Hence the optimal depth-$D$ energy is the minimum of the
state functional $\boldsymbol{\psi}\mapsto\langle\boldsymbol{\psi}|\mathbf{H}|\boldsymbol{\psi}\rangle$
over the reachable set: 
\begin{equation}\label{eq:Estar_def}
  E_D^\star
  = \min_{\bm\theta}E_D(\bm\theta)
  = \min_{|\boldsymbol{\psi}\rangle\in\cM_D}
    \langle\boldsymbol{\psi}|\mathbf{H}|\boldsymbol{\psi}\rangle,
\end{equation}
and likewise for $E_{D+2}^\star$ over $\cM_{D+2}$. By
\Cref{prop:nested}, $\cM_D\subseteq\cM_{D+2}$, so the minimization
defining $E_{D+2}^\star$ ranges over a superset of the feasible set
defining $E_D^\star$. The minimum of a fixed real-valued functional
over a larger set is no larger than over a subset: explicitly, a minimizer
$|\boldsymbol{\psi}_D^\star\rangle\in\cM_D$ of the depth-$D$ problem is also a
feasible point of the depth-$(D{+}2)$ problem, whence:
\begin{equation}\label{eq:Estar_mono}
  E_{D+2}^\star
  = \min_{|\boldsymbol{\psi}\rangle\in\cM_{D+2}}
    \langle\boldsymbol{\psi}|\mathbf{H}|\boldsymbol{\psi}\rangle
  \le \langle\boldsymbol{\psi}_D^\star|\mathbf{H}|\boldsymbol{\psi}_D^\star\rangle
  = E_D^\star.
\end{equation}
Since $D\ge 1$ was arbitrary, $E_{D+2}^\star\le E_D^\star$ for all
$D\ge 1$.
\end{proof}

\medskip\noindent\textbf{Homotopy formulation.}
Define the homotopy unitary:
\begin{equation}\label{eq:homotopy}
  \mathbf{U}_{D+2}^{(\lambda)}(\bm\theta,\bm\phi)
  =\exp\bigl({-}i\lambda \mathbf{G}(\bm\phi)\bigr) 
   \mathbf{U}_D^{\mathrm{SE}}(\bm\theta),
\end{equation}
where $\mathbf{G}(\bm\phi)$ is the Hermitian generator of the two new rotation
layers and $\mathbf{W}(\bm\phi)=\bigotimes_q \mathbf{W}_q(\bm\phi_q)$ is a tensor
product of single-qubit unitaries (constructed explicitly in
\Cref{app:homotopy_extended}, \eqref{eq:W_tensor_app}).
The generator decomposes as
$\mathbf{G}(\bm\phi)=\sum_{j=1}^{6n}\phi_j \mathbf{G}_j+O(\|\bm\phi\|^2)$,
with $\|\mathbf{G}_j\|=1/2$ for standard Pauli rotations (each
$\mathbf{G}_j=\tfrac12\mathbf{P}_q$ for a single-qubit Pauli $\mathbf{P}_q$;
constructed in \Cref{app:homotopy_extended}).
The homotopy energy is:
\begin{equation}\label{eq:F_homotopy}
  F(\bm\theta,\bm\phi;\lambda)
  =\bigl\langle 0^n\bigr|
   \bigl[\mathbf{U}_{D+2}^{(\lambda)}\bigr]^{\dagger}
   \mathbf{H} 
   \mathbf{U}_{D+2}^{(\lambda)}
   \bigl|0^n\bigr\rangle.
\end{equation}
At $\lambda=0$, $F=E_D(\bm\theta)$; at $\lambda=1$,
$F=E_{D+2}(\bm\theta,\bm\phi)$.
The detailed construction of the generator $\mathbf{G}(\bm\phi)$, its
tensor-product decomposition, and the physical interpretation of the
homotopy are given in \Cref{app:homotopy_extended}.
\begin{assumption}[Nondegenerate depth-$D$ minimum]
\label{ass:nondegen}
  Let $\bm\theta_D^\star$ be a local minimizer of $E_D$ with
  $\nabla^2_{\bm\theta\bm\theta}E_D(\bm\theta_D^\star)
  \succeq\mu \mathbf{I}$ for some $\mu>0$.
\end{assumption}
\begin{theorem}[Local continuation under identity expansion]
\label{thm:continuation}
  Suppose:
  \begin{enumerate}[nosep]
    \item \textbf{Nondegeneracy} (\Cref{ass:nondegen}).
    \item \textbf{Smoothness}: $F$ is $C^2$ near
      $(\bm\theta_D^\star,\bm\phi^\star,0)$.
    \item \textbf{Identity embedding}: $F(\bm\theta,\bm\phi;0)=E_D(\bm\theta)$.
    \item \textbf{Fixed target parameters}: $\bm\phi^\star\neq\bm 0$.
    \item \textbf{Non-trivial coupling}:
      $\partial_\lambda\nabla_{\bm\theta}
      F(\bm\theta_D^\star,\bm\phi^\star;\lambda)
      \big|_{\lambda=0}\neq\bm 0$.
  \end{enumerate}
  Then there exist $\lambda_{\max}\in(0,1]$ and a unique $C^1$ curve
  $\bm\theta:[0,\lambda_{\max}]\to\bbR^{3nD}$ with
  $\bm\theta(0)=\bm\theta_D^\star$ and
  $\nabla_{\bm\theta}
  F\bigl(\bm\theta(\lambda),\bm\phi^\star;\lambda\bigr)=\bm 0$.
  The initial velocity is
  \begin{equation}\label{eq:curve_deriv}
    \bm\theta'(0)
    =-\bigl[\nabla^2_{\bm\theta\bm\theta}
      E_D(\bm\theta_D^\star)\bigr]^{-1} 
     \partial_\lambda\nabla_{\bm\theta}
     F(\bm\theta_D^\star,\bm\phi^\star;\lambda)
     \big|_{\lambda=0}
    \neq\bm 0.
  \end{equation}
\end{theorem}
Conditions 1-3 are standard for IFT applications.
Condition~4 is ensured by the jitter ($\sigma>0$).
Condition~5 is generically satisfied: the coupling map
$\bm\phi^\star\mapsto\partial_\lambda\nabla_{\bm\theta}
F(\bm\theta_D^\star,\bm\phi^\star;\lambda)\big|_{\lambda=0}$ is real-analytic
in $\bm\phi^\star$ and, absent a Pauli-commutation symmetry of $\mathbf{H}$, is
not identically zero, so its zero set has Lebesgue measure zero (see
\Cref{rem:nontrivial_genericity_app}).
\begin{proof}[Proof Sketch]
Apply the Implicit Function Theorem to the stationarity map
$\Phi(\bm\theta,\lambda)
\coloneqq\nabla_{\bm\theta}F(\bm\theta,\bm\phi^\star;\lambda)$.
At $(\bm\theta_D^\star,0)$ the two IFT hypotheses hold:
(a)~by the identity embedding~(iii),
$\Phi(\bm\theta_D^\star,0)=\nabla E_D(\bm\theta_D^\star)=\bm 0$
(first-order optimality of the depth-$D$ minimizer);
(b)~the $\bm\theta$-Jacobian
$\partial_{\bm\theta}\Phi(\bm\theta_D^\star,0)
=\nabla^2 E_D(\bm\theta_D^\star)\succeq\mu \mathbf{I}\succ 0$ is invertible
with inverse norm $\le 1/\mu$.  The IFT then yields a unique $C^1$
curve $\bm\theta(\lambda)$ with $\bm\theta(0)=\bm\theta_D^\star$ and
$\Phi(\bm\theta(\lambda),\lambda)=\bm0$.  Differentiating this identity
at $\lambda=0$ and inverting the Hessian gives the
velocity~\eqref{eq:curve_deriv}, which is nonzero by the non-trivial
coupling~(v).  Finally, since $\nabla^2_{\bm\theta\bm\theta}F$ is
continuous and $\succeq\mu\mathbf{I}$ at $\lambda=0$, Weyl's inequality keeps
it $\succeq(\mu/2)\mathbf{I}$ for small $\lambda$, so the tracked point stays
a nondegenerate minimum.
See \Cref{app:continuation_complete} for the complete five-step proof
with explicit verification of all IFT hypotheses, the non-triviality
check, and discussion of the degenerate case $\bm\phi=\bm 0$.
\end{proof}

\begin{remark}
\Cref{thm:continuation} is the formal counterpart of the intuition
behind IP-PDT: an identity-pair expansion does not drop the optimizer
into a fresh, unrelated landscape---it gently \emph{deforms} the one it
has already solved.  The tracked minimum moves at speed
$\|\bm\theta'(0)\|=O(\|\bm\phi^\star\|/\mu)=O(\sigma/\mu)$, so a smaller
jitter $\sigma$ keeps the new stage nearer the previous solution.  The
statement is deliberately local: it guarantees that a good minimum
\emph{persists} under expansion, not that gradient descent will
\emph{find} it.  That remaining gap---persistence versus
findability---is exactly what the next two subsections close, from two
complementary directions: an unconditional monotonicity guarantee
(\Cref{thm:monotone}) and a convergence-rate guarantee
(\Cref{thm:basin}).
\end{remark}

\medskip\noindent\textbf{Perturbation bound on energy deviation.}
Applying the mean-value inequality to the unitary conjugation
$t\mapsto e^{it\mathbf{G}(\bm\phi)}\mathbf{H}e^{-it\mathbf{G}(\bm\phi)}$
(\Cref{app:perturb}, with \emph{no} truncation of the series):
\begin{equation}\label{eq:perturb}
  |E_{D+2}(\bm\theta,\bm\phi)-E_D(\bm\theta)|
  \le 2\|\mathbf{H}\| \|\mathbf{G}(\bm\phi)\|
  \le 2\|\mathbf{H}\|\Bigl(\sum_{j=1}^{6n}|\phi_j| \|\mathbf{G}_j\|
     +\|\mathbf{R}(\bm\phi)\|\Bigr).
\end{equation}
The first inequality is \emph{exact}.  The second uses the generator
decomposition $\mathbf{G}(\bm\phi)=\sum_j\phi_j\mathbf{G}_j+\mathbf{R}(\bm\phi)$, where the
second-order remainder is bounded by $\|\mathbf{R}(\bm\phi)\|\le C_R\|\bm\phi\|^2$
for a finite generator-dependent constant $C_R$ (\Cref{app:perturb}); we
\emph{retain} $\mathbf{R}$ as an additive term in every downstream bound rather
than discarding it.
\emph{Application to jittered identity pairs.}  Under our
initialization, the new-layer parameters are
$\phi_j\sim\mathcal{N}(0,\sigma^2)$, so $\bbE[|\phi_j|]=\sigma\sqrt{2/\pi}$,
$\bbE\|\bm\phi\|^2=6n\sigma^2$, and $\|\mathbf{G}_j\|=\tfrac12$.  Taking
expectations in~\eqref{eq:perturb},
\[
  \bbE\bigl[|E_{D+2}-E_D|\bigr]
  \le 6n\sigma\|\mathbf{H}\|\sqrt{2/\pi}
    +  2\|\mathbf{H}\| C_R 6n\sigma^2 ,
\]
i.e.\ a deviation linear in $\sigma$ with an explicitly bounded quadratic
correction (not dropped).  The complete derivation appears in
\Cref{app:perturb}.
\subsection{From theoretical monotonicity to practical monotonicity}
\label{sec:practical_mono}
The continuation result above is conditional---it presumes the deeper
minimum is actually reached.  We now give a guarantee that needs no such
presumption.  The acceptance rule makes the \emph{produced} energies
monotone by construction: if a stage fails to improve, IP-PDT falls back
to the identity-embedded previous solution, which by the IEP has
\emph{exactly} the previous stage's energy.
\begin{theorem}[Monotone non-increase of produced energies]
\label{thm:monotone}
  Under the acceptance rule (accept if energy improved, else fall back
  to the embedded parameters $\iota_D(\hat{\bm\theta}^{(D)})$), for
  all stages~$D$:
  \begin{equation}\label{eq:monotone_iter}
    E_{D+2}\bigl(\hat{\bm\theta}^{(D+2)}\bigr)
    \le E_D\bigl(\hat{\bm\theta}^{(D)}\bigr).
  \end{equation}
\end{theorem}
\begin{proof}
Fix a stage and let $\hat{\bm\theta}^{(D)}$ be the accepted depth-$D$
parameters with produced energy $E_D(\hat{\bm\theta}^{(D)})$.  Expanding
to depth $D{+}2$ and running the optimizer from the identity-pair
initialization yields a candidate $\tilde{\bm\theta}^{(D+2)}$.  The
acceptance rule produces the depth-$(D{+}2)$ parameters
\begin{equation}\label{eq:accept_def}
  \hat{\bm\theta}^{(D+2)}=
  \begin{cases}
    \tilde{\bm\theta}^{(D+2)},
      & \text{if } E_{D+2}\bigl(\tilde{\bm\theta}^{(D+2)}\bigr)
        \le E_D\bigl(\hat{\bm\theta}^{(D)}\bigr)
        \quad(\text{accept}),\\[2pt]
    \iota_D\bigl(\hat{\bm\theta}^{(D)}\bigr),
      & \text{otherwise}\quad(\text{fall back}).
  \end{cases}
\end{equation}
We verify~\eqref{eq:monotone_iter} in each branch.

\emph{Accept branch.}  Acceptance is triggered exactly when
$E_{D+2}(\tilde{\bm\theta}^{(D+2)})\le E_D(\hat{\bm\theta}^{(D)})$.
Since in this branch
$\hat{\bm\theta}^{(D+2)}=\tilde{\bm\theta}^{(D+2)}$, substituting gives
\begin{equation}\label{eq:accept_branch}
  E_{D+2}\bigl(\hat{\bm\theta}^{(D+2)}\bigr)
  =E_{D+2}\bigl(\tilde{\bm\theta}^{(D+2)}\bigr)
  \le E_D\bigl(\hat{\bm\theta}^{(D)}\bigr),
\end{equation}
which is the conclusion~\eqref{eq:monotone_iter}.

\emph{Fall-back branch.}  Here
$\hat{\bm\theta}^{(D+2)}=\iota_D(\hat{\bm\theta}^{(D)})
=(\hat{\bm\theta}^{(D)},\bm\theta_{\mathrm{ref}},\bm\theta_{\mathrm{ref}})$
is the embedding of the depth-$D$ solution with the two appended layers
frozen at the identity pair (\Cref{def:id_pair}).  The
identity-embedding property---\Cref{prop:nested} in reachable-set form,
or its energy specialization~\eqref{eq:energy_preserve_exact}---states
that this embedding reproduces the depth-$D$ state \emph{exactly}, hence
preserves the energy exactly:
\begin{equation}\label{eq:fallback_branch}
  E_{D+2}\bigl(\hat{\bm\theta}^{(D+2)}\bigr)
  =E_{D+2}\bigl(\iota_D(\hat{\bm\theta}^{(D)})\bigr)
  =E_D\bigl(\hat{\bm\theta}^{(D)}\bigr)
  \le E_D\bigl(\hat{\bm\theta}^{(D)}\bigr),
\end{equation}
where the middle equality is~\eqref{eq:energy_preserve_exact} and the
final inequality is equality.  Thus~\eqref{eq:monotone_iter} holds with
equality on this branch.

Since the two branches in~\eqref{eq:accept_def} are exhaustive
and~\eqref{eq:monotone_iter} holds in each,
$E_{D+2}(\hat{\bm\theta}^{(D+2)})\le E_D(\hat{\bm\theta}^{(D)})$ for
every stage.  No property of the Hamiltonian, the energy landscape, or
the optimizer that produced $\tilde{\bm\theta}^{(D+2)}$ enters the
argument: the fall-back branch caps the produced energy unconditionally
through the exact embedding~\eqref{eq:energy_preserve_exact}, so the
guarantee is independent of \Cref{ass:local_reg} and of whether any
deeper minimum is reached.
\end{proof}
\begin{remark}
This theorem holds unconditionally: no assumptions on the landscape,
Hamiltonian, or optimizer quality are needed.  Na\"ive PDT cannot
implement such a lossless fallback because the HEA architecture has no
identity embedding.  See \Cref{app:monotone_extended} for the extended
discussion and contrast with na\"ive PDT.
\end{remark}
\subsection{Basin preservation under gradient descent}
\label{sec:basin}
\begin{assumption}[Local regularity]\label{ass:local_reg}
  There exists a bounded convex neighbourhood $\cN_D$ of
  $\bm\theta_D^\star$ on which $E_D$ is $L_D$-smooth and $\mu_D$-strongly
  convex; equivalently,
  $\mu_D\mathbf{I}\preceq\nabla^2 E_D(\bm\theta)\preceq L_D\mathbf{I}$ for all
  $\bm\theta\in\cN_D$.  (Boundedness is necessary: the globally bounded,
  $4\pi$-periodic energy $E_D$ cannot be strongly convex on an unbounded
  set, so strong convexity is a genuinely local property here.)
\end{assumption}
We motivate local strong convexity in \Cref{app:pl_motivation}.

\begin{theorem}[Basin preservation for old parameters]
\label{thm:basin}
  Adopt \Cref{ass:local_reg} and the IEP, with new-layer parameters
  fixed at $\bm\phi^\star$ ($\|\bm\phi^\star\|=O(\sigma)$); let
  $\bm\theta_{D+2}^\star$ minimize $E_{D+2}(\cdot,\bm\phi^\star)$ on
  $\cN_D$.  Assume the warm start $\bm\theta^{(0)}$ is close enough to the
  minimizer that the Euclidean ball
  $\overline{B}\bigl(\bm\theta_{D+2}^\star,\|\bm\theta^{(0)}
  -\bm\theta_{D+2}^\star\|\bigr)\subseteq\cN_D$, and let the jitter obey
  $\sigma\le\min(\sigma_0,\sigma_1)$ for the explicit curvature threshold
  $\sigma_0=\mu_D/(2c)>0$---with $c=O(n^{3/2}D\|\mathbf{H}\|)$ the
  Hessian-perturbation constant of~\eqref{eq:basin_sigma_threshold}---and the
  interiority threshold $\sigma_1$ of \Cref{app:basin_complete} (ensuring
  $\bm\theta_{D+2}^\star$ is interior to $\cN_D$).  Then $\mu_{D+2}\ge\mu_D/2>0$,
  $L_{D+2}\le L_D+c\sigma$, and for every step size
  $\eta\in(0,1/L_{D+2}]$ the gradient-descent iterates remain in $\cN_D$
  and satisfy
  \begin{equation}\label{eq:linear_rate}
    E_{D+2}(\bm\theta^+,\bm\phi^\star)-E_{D+2}(\bm\theta_{D+2}^\star,\bm\phi^\star)
    \le(1-\eta\mu_{D+2})\bigl(E_{D+2}(\bm\theta,\bm\phi^\star)
    -E_{D+2}(\bm\theta_{D+2}^\star,\bm\phi^\star)\bigr).
  \end{equation}
\end{theorem}
\begin{proof}
Write $\widetilde{E}(\bm\theta)\coloneqq
E_{D+2}(\bm\theta,\bm\phi^\star)$.
By the perturbation bound~\eqref{eq:perturb}:
$|\widetilde{E}(\bm\theta)-E_D(\bm\theta)|=O(\sigma)$.
Differentiating:
\begin{align}
  \|\nabla\widetilde{E}(\bm\theta)
  -\nabla E_D(\bm\theta)\|
  &=O(\sigma),\label{eq:grad_diff}\\
  \|\nabla^2\widetilde{E}(\bm\theta)
  -\nabla^2 E_D(\bm\theta)\|
  &=O(\sigma).\label{eq:hess_diff}
\end{align}
To see why~\eqref{eq:grad_diff} holds, note that
$\widetilde{E}(\bm\theta) - E_D(\bm\theta)
= \langle\boldsymbol{\psi}_D(\bm\theta)|(\mathbf{W}^\dagger \mathbf{H} \mathbf{W} - \mathbf{H})|\boldsymbol{\psi}_D(\bm\theta)
\rangle$,
where $\|\mathbf{W}^\dagger \mathbf{H} \mathbf{W} - \mathbf{H}\| = O(\sigma)$
by~\eqref{eq:perturb}.  Each parameter derivative
$\partial_{\theta_j}|\boldsymbol{\psi}_D\rangle$ has norm $\le\tfrac12$ (a
single half-Pauli generator; \Cref{app:grad_hess}, \eqref{eq:dpsi_bound}),
so by the product rule and Cauchy--Schwarz each component obeys
$|\partial_{\theta_j}(\widetilde{E} - E_D)|\le\|\mathbf{W}^\dagger \mathbf{H} \mathbf{W} - \mathbf{H}\|
= O(\sigma)$, and aggregating the $3nD$ components,
$\|\nabla(\widetilde{E} - E_D)\| \le \sqrt{3nD}\,\|\mathbf{W}^\dagger \mathbf{H} \mathbf{W} - \mathbf{H}\|
= O(\sqrt{nD}\,\sigma)$.  The Hessian bound~\eqref{eq:hess_diff}
follows by a second differentiation of the same identity
(\Cref{app:grad_hess}, Step~4).
A complete derivation appears in \Cref{app:grad_hess}.

\emph{Constant transfer.}  The smoothness constant can only grow by the
Hessian deviation: for $\bm\theta,\bm\theta'\in\cN_D$, the mean-value
inequality applied to $\bm\theta\mapsto(\nabla\widetilde{E}-\nabla E_D)(\bm\theta)$
(whose Jacobian is the Hessian difference bounded by $c\sigma$
in~\eqref{eq:hess_diff}), together with the $L_D$-smoothness of $E_D$, gives
$\|\nabla\widetilde{E}(\bm\theta)-\nabla\widetilde{E}(\bm\theta')\|
\le(L_D+c\sigma)\|\bm\theta-\bm\theta'\|$, so $L_{D+2}\le L_D+c\sigma$
(the same Hessian-perturbation constant $c$ as for the curvature, since
both come from the Weyl bound $\|\nabla^2\widetilde{E}-\nabla^2 E_D\|\le c\sigma$).
For the curvature constant, \Cref{ass:local_reg} gives
$\nabla^2 E_D\succeq\mu_D\mathbf{I}$ on $\cN_D$---this is precisely the role
of strong convexity, since a Hessian lower bound, unlike the PL
inequality, is preserved under the $O(\sigma)$ perturbation.  Weyl's
eigenvalue perturbation inequality applied to the Hermitian
difference~\eqref{eq:hess_diff} shifts every eigenvalue by at most
$\|\nabla^2\widetilde{E}-\nabla^2 E_D\|=O(\sigma)$, so
$\nabla^2\widetilde{E}\succeq(\mu_D-c\sigma)\mathbf{I}$ for a constant
$c>0$.  Writing $\mu_{D+2}\coloneqq\mu_D-c\sigma$, the smallness
condition $\sigma\le\sigma_0\coloneqq\mu_D/(2c)$ gives
\begin{equation}\label{eq:sigma_threshold_main}
  \mu_{D+2}=\mu_D-c\sigma\ge\mu_D/2>0,
\end{equation}
the regime in which the statement holds; the $\mu_{D+2}$-strong convexity
of $\widetilde{E}$ then yields the $\mu_{D+2}$-PL inequality.
Applying the descent lemma ($L_{D+2}$-smoothness) with the step
restriction $\eta\le 1/L_{D+2}$ (which makes the coefficient
$1-L_{D+2}\eta/2\ge 1/2$), followed by the PL inequality,
\begin{align}
  \widetilde{E}(\bm\theta^+)
  &\le \widetilde{E}(\bm\theta)
  -\eta\Bigl(1-\tfrac{L_{D+2}\eta}{2}\Bigr)
   \|\nabla \widetilde{E}(\bm\theta)\|^2
  \le \widetilde{E}(\bm\theta)
  -\frac{\eta}{2}\|\nabla \widetilde{E}(\bm\theta)\|^2 \notag\\
  &\le \widetilde{E}(\bm\theta)
   -\eta\mu_{D+2}\bigl(\widetilde{E}(\bm\theta)
   -\widetilde{E}(\bm\theta_{D+2}^\star)\bigr).
\end{align}
Subtracting $\widetilde{E}(\bm\theta_{D+2}^\star)$ from both sides
gives~\eqref{eq:linear_rate} with contraction
$(1-\eta\mu_{D+2})\in[0,1)$ (the lower bound $\ge 0$ follows from
$\eta\le 1/L_{D+2}\le 1/\mu_{D+2}$).
The complete five-step proof with all intermediate inequalities,
smoothness and PL transfer, and energy proximity analysis appears in
\Cref{app:basin_complete}.
\end{proof}

\medskip\noindent\textbf{Convergence rate.}
Iterating over $T$ steps:
\begin{equation}\label{eq:T_step_convergence}
  E_{D+2}(\bm\theta^{(T)},\bm\phi^\star)-E_{D+2}(\bm\theta_{D+2}^\star,\bm\phi^\star)
  \le(1-\eta\mu_{D+2})^T
   \bigl[E_{D+2}(\bm\theta^{(0)},\bm\phi^\star)-E_{D+2}(\bm\theta_{D+2}^\star,\bm\phi^\star)\bigr].
\end{equation}
\begin{remark}
\Cref{eq:T_step_convergence} says each stage shrinks its energy gap
geometrically, at a rate set by $\mu_{D+2}$.  The substance of
\Cref{thm:basin} is that this rate is \emph{inherited}: the regularity
constants of the expanded landscape differ from those of the previous
stage by only $O(\sigma)$.  Identity-pair initialization is what makes
this work---it places the optimizer \emph{inside} the basin it has
already descended, rather than at a random point of a deeper landscape,
so the favourable geometry of stage~$D$ carries over to stage~$D{+}2$.
This is the precise sense in which a shock-free expansion ``preserves''
trainability.  The ball-containment hypothesis is mild for the same
reason: the warm start $\bm\theta^{(0)}=\iota_D(\hat{\bm\theta}^{(D)})$
differs from the new minimizer $\bm\theta_{D+2}^\star$ by only
$O(\sigma/\mu_D)$ (the displacement bound established in
\Cref{app:basin_complete}), so the ball
$\overline{B}(\bm\theta_{D+2}^\star,\|\bm\theta^{(0)}-\bm\theta_{D+2}^\star\|)$
has radius $O(\sigma)$ and sits inside $\cN_D$ whenever the previous
stage has converged ($\hat{\bm\theta}^{(D)}\approx\bm\theta_D^\star$) and
$\sigma\le\min(\sigma_0,\sigma_1)$, with $\sigma_1$ the interiority threshold
of \Cref{app:basin_complete}---so the hypothesis follows from warm-starting
under these mild conditions rather than being imposed separately.
\end{remark}

\begin{remark}[Scope limitation: joint dynamics]
\label{rem:joint_dynamics}
  \Cref{thm:basin} establishes convergence for the old parameters
  $\bm\theta$ with new-layer parameters fixed.  Joint optimization
  requires additional assumptions and is left to future work.
\end{remark}
\begin{remark}[What \Cref{thm:basin} does \emph{not} claim]
\label{rem:basin_scope}
  \Cref{thm:basin} is an \emph{inheritance} result, not a
  barren-plateau-avoidance result.  \emph{If} a stage is
  well-conditioned---locally $\mu_D$-strongly convex on $\cN_D$---a
  shock-free expansion transfers that conditioning to the next stage up
  to $O(\sigma)$.  The hypothesis can fail: in the barren regime $\mu_D$
  is exponentially small in $n$, so both the admissible jitter
  $\sigma_0=\mu_D/(2c)$ and the per-step progress $\eta\mu_{D+2}$ shrink
  accordingly and the guarantee becomes vacuous.  What the theorem buys
  is a \emph{per-stage} decoupling---for the old parameters with the new
  layers held at $\bm\phi^\star$ (\Cref{rem:joint_dynamics} fences the
  joint case): each warm-started expansion re-optimizes within an
  \emph{inherited} basin rather than a fresh deep landscape, so the
  per-stage difficulty need not compound with depth---\emph{provided} the
  shallow first stage is itself well-conditioned, where $\mu_D$ is most
  favourable and the single entangling layer limits the expressibility
  that drives plateaus.  Whether that shallow stage is itself
  plateau-free is a separate question.
\end{remark}
\subsection{Spectral-gap bound on ground-state fidelity}
\label{sec:gap_fidelity}
The results so far control the \emph{energy} that IP-PDT produces.  The
physical target, however, is the ground \emph{state}.  The final result
closes that loop: it shows that a small energy error certifies high
ground-state fidelity, with the exchange rate between the two set by the
spectral gap~$\Delta$.
\begin{theorem}[Gap--fidelity sandwich]\label{thm:gap_fidelity}
  For a spectral gap $\Delta=\lambda_2-\lambda_1>0$ (so the ground state
  $\mathbf{v}_1$ is non-degenerate) and fidelity
  $F(\boldsymbol{\psi})=|\langle \mathbf{v}_1|\boldsymbol{\psi}\rangle|^2$:
  \begin{equation}\label{eq:sandwich}
    \Delta\bigl(1-F(\boldsymbol{\psi})\bigr)
    \le\langle\boldsymbol{\psi}|\mathbf{H}|\boldsymbol{\psi}\rangle-\lambda_1
    \le(\lambda_d-\lambda_1)\bigl(1-F(\boldsymbol{\psi})\bigr).
  \end{equation}
\end{theorem}
\begin{proof}
Expand $|\boldsymbol{\psi}\rangle=\sum_i c_i|\mathbf{v}_i\rangle$.
Then $\langle\boldsymbol{\psi}|\mathbf{H}|\boldsymbol{\psi}\rangle-\lambda_1
=\sum_{i\ge 2}(\lambda_i-\lambda_1)|c_i|^2$.
Lower bound: $\lambda_i-\lambda_1\ge\Delta$ for $i\ge 2$.
Upper bound: $\lambda_i-\lambda_1\le\lambda_d-\lambda_1$.
A detailed step-by-step proof and applications to our benchmark
Hamiltonians appear in \Cref{app:gap_fidelity_extended}.
\end{proof}

\medskip\noindent\textbf{Energy--fidelity certificate.}
If $\langle\boldsymbol{\psi}|\mathbf{H}|\boldsymbol{\psi}\rangle-\lambda_1\le\varepsilon$, then
$F(\boldsymbol{\psi})\ge 1-\varepsilon/\Delta$.
Apply this with $\boldsymbol{\psi}=\boldsymbol{\psi}(\bm\theta^{(T)})$ at depth $D{+}2$, so that
$\langle\boldsymbol{\psi}|\mathbf{H}|\boldsymbol{\psi}\rangle=E_{D+2}(\bm\theta^{(T)},\bm\phi^\star)$, and split
the suboptimality through the depth-$(D{+}2)$ minimizer:
$\varepsilon=\bigl[E_{D+2}(\bm\theta^{(T)},\bm\phi^\star)
-E_{D+2}(\bm\theta_{D+2}^\star,\bm\phi^\star)\bigr]+r_{D+2}$, where
$r_{D+2}\coloneqq E_{D+2}(\bm\theta_{D+2}^\star,\bm\phi^\star)-\lambda_1\ge 0$
is the \emph{reachability residual}: the excess of the local minimizer's
energy over the true ground energy~$\lambda_1$, itself at least the
representability floor $E_{D+2}^\star-\lambda_1$ (positive whenever
$|\mathbf{v}_1\rangle\notin\cM_2$).  Combining with the linear
rate~\eqref{eq:T_step_convergence}:
\begin{equation}\label{eq:fidelity_convergence}
  F(\boldsymbol{\psi}(\bm\theta^{(T)}))
  \ge 1-\frac{1}{\Delta}\Bigl[
    (1-\eta\mu_{D+2})^T
    \bigl(E_{D+2}(\bm\theta^{(0)},\bm\phi^\star)
    -E_{D+2}(\bm\theta_{D+2}^\star,\bm\phi^\star)\bigr)
    +r_{D+2}\Bigr].
\end{equation}
\begin{remark}[Reachability caveat]\label{rem:reachability}
  The optimization term in~\eqref{eq:fidelity_convergence} vanishes as
  $T\to\infty$, but the reachability residual $r_{D+2}$ does not: the
  guarantee saturates at $1-r_{D+2}/\Delta$, useful only when
  $E_{D+2}(\bm\theta_{D+2}^\star,\bm\phi^\star)\approx\lambda_1$.
  By the Saturation Theorem, $\cM_D=\cM_2$ for $D\ge 2$; if
  $|\mathbf{v}_1\rangle\notin\cM_2$ then $r_{D+2}>0$ is bounded away from zero and
  the fidelity guarantee cannot reach~$1$.  The certificate is likewise
  uninformative when the spectral gap is small relative to the energy
  error: as $\Delta\to0$ (e.g.\ near criticality) the exchange rate
  $1/\Delta\to\infty$, so the bound $F\ge1-\varepsilon/\Delta$ is useful
  only for $\varepsilon\ll\Delta$.
\end{remark}
\section{Experiments}\label{sec:experiments}
We now put the theory to the test, and we organize the experiments as a
single line of inquiry rather than a checklist.  We first confirm the
structural fact the whole story rests on---that the single-entangler
manifold saturates (\Cref{sec:exp_saturation}).  We then show that this
saturation is precisely what lets added depth help without hurting: it
removes the initialization shock that plagues na\"ive growth
(\Cref{sec:exp_shock}). 
Next we isolate the resulting
\emph{trainability} gain from any expressibility gain
(\Cref{sec:exp_trainability}) and ask what it costs
(\Cref{sec:exp_cost}).  We then stress-test generality---across
nine Hamiltonians, against the more expressive HEA, and against
established VQE optimizers
(\Cref{sec:exp_extended,sec:exp_ceiling,sec:exp_sota,sec:exp_sota_budget})---and
probe the design choices behind IP-PDT
(\Cref{sec:exp_schedules,sec:exp_hyperparams,sec:ablation}).
Finally, we scale the benchmark from $n=6$ to $n=16$ qubits
(\Cref{sec:exp_n16}).  Throughout, every method receives the same
parameter count and the same step budget, so each comparison isolates
the effect of architecture and growth schedule alone.

Unless stated otherwise, experiments use $n=6$ qubits with exact
statevector simulation (PennyLane), comparing IP-PDT against up to five
competing VQE strategies and reporting mean $\pm$ standard deviation over
$10$ random seeds; the scaling study of \Cref{sec:exp_n16} uses a
Pauli-sum Hamiltonian representation to reach $n=16$.
\subsection{Setup}\label{sec:exp_setup}
\subsubsection{Hamiltonians}\label{sec:hamiltonians}
We consider nine Hamiltonians spanning a range of entanglement
requirements and symmetry structures (\Cref{tab:hamiltonians_full}).
The first three serve as primary benchmarks throughout all experiments;
the remaining six test generality in \Cref{sec:exp_extended}.
\begin{table}[t]
\centering\small
\caption{Benchmark Hamiltonians ($n=6$).  $\lambda_1$ is the exact
ground-state energy; $\Delta$ is the spectral gap.}
\label{tab:hamiltonians_full}
\begin{tabular}{@{}llccc@{}}
\toprule
\textbf{Hamiltonian} & \textbf{Interactions} & $\lambda_1$ & $\Delta$ & \textbf{Difficulty for SE} \\
\midrule
\multicolumn{5}{@{}l}{\textit{Primary benchmarks}} \\
TFIM ($J = h = 1$, ring) & $\mathbf{Z}\mathbf{Z} + \mathbf{X}$ & $-7.727$ & $0.27$ & Hard \\
Tilted Ising (ring) & $\mathbf{Z}\mathbf{Z} + \mathbf{X} + \mathbf{Z}$ & $-10.214$ & $0.44$ & Easy \\
Random Ising (all-to-all) & $\mathbf{Z}\mathbf{Z} + \mathbf{Z}$ & $-3.405$ & $0.53$ & Medium \\
\midrule
\multicolumn{5}{@{}l}{\textit{Extended benchmarks}} \\
Mixed-field Ising (ring) & $\mathbf{Z}\mathbf{Z} + \mathbf{X} + \mathbf{Y} + \mathbf{Z}$ & $-10.324$ & --- & Easy \\
Spin glass + transverse & $\mathbf{Z}\mathbf{Z} + \mathbf{Z} + \mathbf{X}$ & $-5.543$ & --- & Medium \\
XXZ ($\Delta_z = 0.5$, ring) & $\mathbf{X}\mathbf{X} + \mathbf{Y}\mathbf{Y} + \mathbf{Z}\mathbf{Z}$ & $-6.819$ & --- & Medium \\
Random XY (all-to-all) & $\mathbf{X}\mathbf{X} + \mathbf{Y}\mathbf{Y} + \mathbf{Z}$ & $-4.493$ & --- & Medium \\
Power-law Ising ($\alpha = 2$) & long-range $\mathbf{Z}\mathbf{Z} + \mathbf{X}$ & $-9.041$ & --- & Hard \\
Heisenberg (dense) & $\mathbf{X}\mathbf{X} + \mathbf{Y}\mathbf{Y} + \mathbf{Z}\mathbf{Z} + \mathbf{X} + \mathbf{Z}$ & $-4.910$ & --- & Hard \\
\bottomrule
\end{tabular}
\par\smallskip
  {\footnotesize ``Difficulty for SE'' measures the proximity of the ground state
  to the single-entangler reachable set $\cM_2$ (one CNOT ring plus local
  rotations): \emph{Easy}~$=$ lies in or near $\cM_2$; \emph{Hard}~$=$ leaves
  $\cM_2$ (criticality, long-range, or dense couplings); \emph{Medium}~$=$
  intermediate. The labels summarize the experiments \emph{a posteriori}; they
  are not an input to them.}
\end{table}
\subsubsection{Protocol and budget matching}\label{sec:protocol}
All progressive methods (IP-PDT, Na\"ive PDT) use stage depths
$\{1,3,5\}$ with $T$ gradient-descent steps per stage ($3T$ total).
Baselines receive the same $3T$ total steps at a fixed depth.  The
learning rate is $\eta=0.02$ and the jitter scale is $\sigma=0.05$
unless stated otherwise.  The shared initial block
$\bm\theta_0^{(0)}$ is drawn from
$\mathrm{Unif}([0,2\pi]^{3n})$ with the same seed across all methods
to ensure fair comparison.
\textcolor{black}{Unless stated otherwise, ``Na\"ive PDT'' denotes the \emph{copy}-initialized
variant, which appends each new block as a copy of the shared initial
block $\bm\theta_0^{(0)}$ rather than drawing it afresh---the strongest
na\"ive warm start, and the one used in all benchmark comparisons below.
The \emph{random}-initialized variant (a fresh $\mathrm{Unif}([0,2\pi]^{3n})$
draw per appended block) appears only in the initialization-shock
experiment of \Cref{sec:exp_shock}, where it isolates the shock mechanism.}
\subsection{Reachable set saturation}\label{sec:exp_saturation}
The entire story rests on one structural claim: that adding depth to the
single-entangler ansatz buys nothing once a threshold is crossed.  The
Saturation Theorem (\Cref{thm:saturation}) makes this precise, predicting
that the reachable set $\mathcal{M}_D$ stabilizes at $D\geq 2$, and our
first task is to see it happen.  We train SE circuits from scratch at
depths $D\in\{1,2,3,5,7,11\}$, giving each $2000$ gradient steps and
recording the final energy per seed (\Cref{fig:saturation}).

On Tilted Ising the saturation is unambiguous.  A single layer is
expressibility-limited and stalls at $-9.000\pm 0.000$, whereas every
$D\geq 2$ reaches $-10.200\pm 0.007$, within $0.014$ of the exact ground
state, and the standard deviation collapses by two orders of magnitude
between $D = 1$ and $D = 2$.  The other two Hamiltonians tell the same
story with a caveat.  On TFIM all $D\geq 2$ land on similar means
($-7.29$ to $-7.43$) with overlapping error bars, but the ceiling
($\approx -7.43$, gap $0.30$) sits well above the exact energy, a first
glimpse of the SE ansatz's expressibility limit for highly entangled
ground states.  On Random Ising, $D = 5$ posts the best mean ($-3.191$)
yet with high variance ($\pm 0.361$), and deeper circuits do not improve
on it---again consistent with saturation, here confounded by local
minima.  In every case the reachable set has effectively closed by
$D = 2$, which is the fact the rest of the paper exploits.
\begin{figure}[t]
\centering
\includegraphics[width=\textwidth]{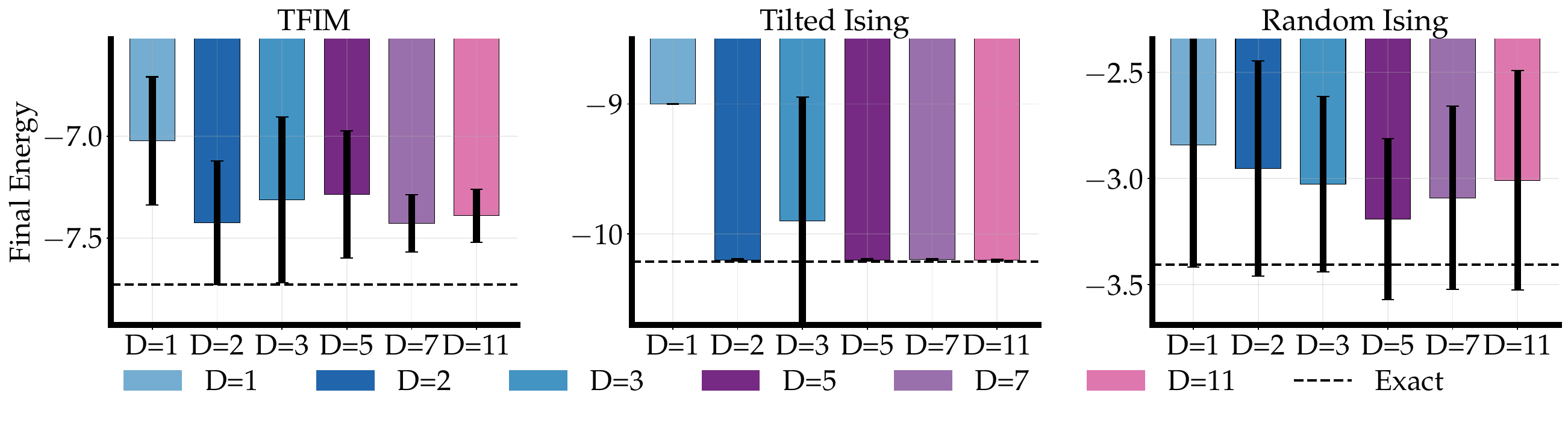}
\caption{Reachable set saturation: final energy vs.\ circuit depth
for the SE ansatz ($2000$ steps, $10$ seeds).  $D = 1$ is strictly
worse; $D\geq 2$ saturates, confirming \Cref{thm:saturation}.
Dashed line: exact ground state.}
\label{fig:saturation}
\end{figure}
\subsection{Initialization shock and monotonicity}\label{sec:exp_shock}
If saturation says depth eventually stops helping, the next question is
what happens at the moment depth is added.  Na\"ive expansion drops in
fresh CNOT gates that perturb the unitary, and the optimizer pays for it;
identity-paired initialization is designed to make the expanded circuit
implement exactly the same unitary it did before, so that the energy
cannot jump.  \Cref{thm:monotone} predicts a monotone descent, and
\Cref{fig:shock} lets us watch the two behaviors side by side.  We trace
the energy across stages for three methods at $T = 200$ steps per
stage---IP-PDT against Na\"ive PDT with copy initialization and Na\"ive
PDT with random initialization.

At each expansion point (the dotted vertical lines), Na\"ive PDT spikes
sharply as the optimizer scrambles to undo the disruption from the newly
introduced gates, while IP-PDT slides through the same points on a smooth,
monotonically non-increasing trajectory, exactly as
\Cref{thm:monotone} guarantees.  The contrast is starkest on Tilted
Ising, where IP-PDT descends steadily to $-10.194\pm 0.009$ while
Na\"ive PDT with random initialization oscillates between $-10$ and $0$
at every expansion.  \textcolor{black}{The copy-initialized na\"ive variant, however,
warm-starts each appended block from $\bm\theta_0^{(0)}$ rather than from
fresh randomness, so it perturbs the trained state far less and---like
IP-PDT---avoids the large spikes; the two nearly coincide.  What still
favors IP-PDT is not this energy trace but its \emph{structure}: the
cancelling forward/inverse pairs collapse the trained circuit to a single
entangling layer ($\sim\!5\times$ fewer two-qubit gates,
\Cref{sec:resource}) and preserve the energy \emph{exactly} at every
expansion (\Cref{thm:monotone}), whereas copy initialization only makes
the shock empirically small with no guarantee.}
\begin{figure}[t]
\centering
\includegraphics[width=\textwidth]{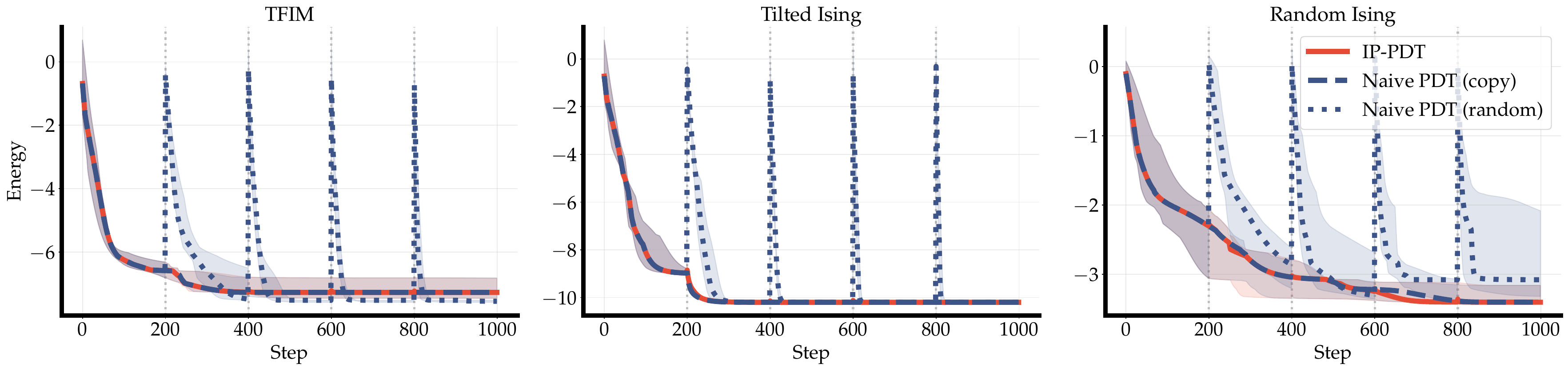}
\caption{Initialization shock across depth expansions ($T = 200$
per stage, $10$ seeds).  \textcolor{black}{Three methods are traced.  IP-PDT (red) shows
smooth monotone descent.  Na\"ive PDT with \emph{random} initialization
suffers dramatic energy spikes at each expansion point (dotted vertical
lines).  Na\"ive PDT with \emph{copy} initialization (the benchmark
default; it warm-starts each appended block from $\bm\theta_0^{(0)}$)
largely avoids the shock and nearly coincides with IP-PDT---so IP-PDT's
advantage over copy is not this trajectory but the $\sim\!5\times$
two-qubit-gate saving from its cancelling pairs (\Cref{sec:resource}) and
the \emph{exact} monotonicity guarantee of \Cref{thm:monotone}.}}
\label{fig:shock}
\end{figure}

\subsection{Trainability beyond expressibility}\label{sec:exp_trainability}
Saturation sets up the cleanest possible test of why progressive training
helps.  Once $\mathcal{M}_2=\mathcal{M}_5$, a depth-5 circuit trained from
scratch can reach exactly the same states as IP-PDT, so any gap between
them cannot be about expressibility---it must be about
\emph{trainability}, about how easily the optimizer finds those states.
We therefore pit IP-PDT ($1{\to}3{\to}5$) against SE circuits trained
from scratch at $D\in\{1,2,5\}$ under an identical budget of $3T$
gradient steps, and read off the residual advantage as a pure
optimization gain (\Cref{fig:trainability}).

Here, \emph{Scratch-$D = 5$} (``SE scratch'') is the single-entangler circuit at the
  full target depth $D = 5$ trained directly from one random initialization---all
  $15n$ parameters at once, with no progressive depth growth and no identity-pair
  warm start---under the same total step budget; since it shares IP-PDT's reachable
  set $\cM_2$ (\Cref{thm:saturation}), the IP-PDT-vs-scratch gap isolates the
  benefit of the curriculum from that of expressibility.
  
That gain is obvious on some Random Ising, where IP-PDT matches or wins every scratch
depth at $T = 100, 500$: at $T = 200$ (600 total steps) only $D=5$ wins over IP-PDT.
Note that IP-PDT holds overall lower variance---it is not just better on average but
markedly more reliable.  On Tilted Ising it matches scratch-$D = 5$ at
all budgets (both reach $-10.194$ by $T = 100$), comfortably ahead of
scratch-$D = 1$ at its $-9.0$ plateau and scratch-$D = 2$ with its
$\pm 0.91$ variance.  The exception is instructive: on TFIM,
scratch-$D = 2$ overtakes IP-PDT at large budgets ($T = 500$: $-7.425$
vs.\ $-7.048$), simply because it spends all $1500$ steps optimizing its
$36$ parameters whereas IP-PDT burns its first $500$ steps at depth~$1$,
where the expressibility ceiling leaves little to gain.
\begin{figure}[t]
\centering
\includegraphics[width=\textwidth]{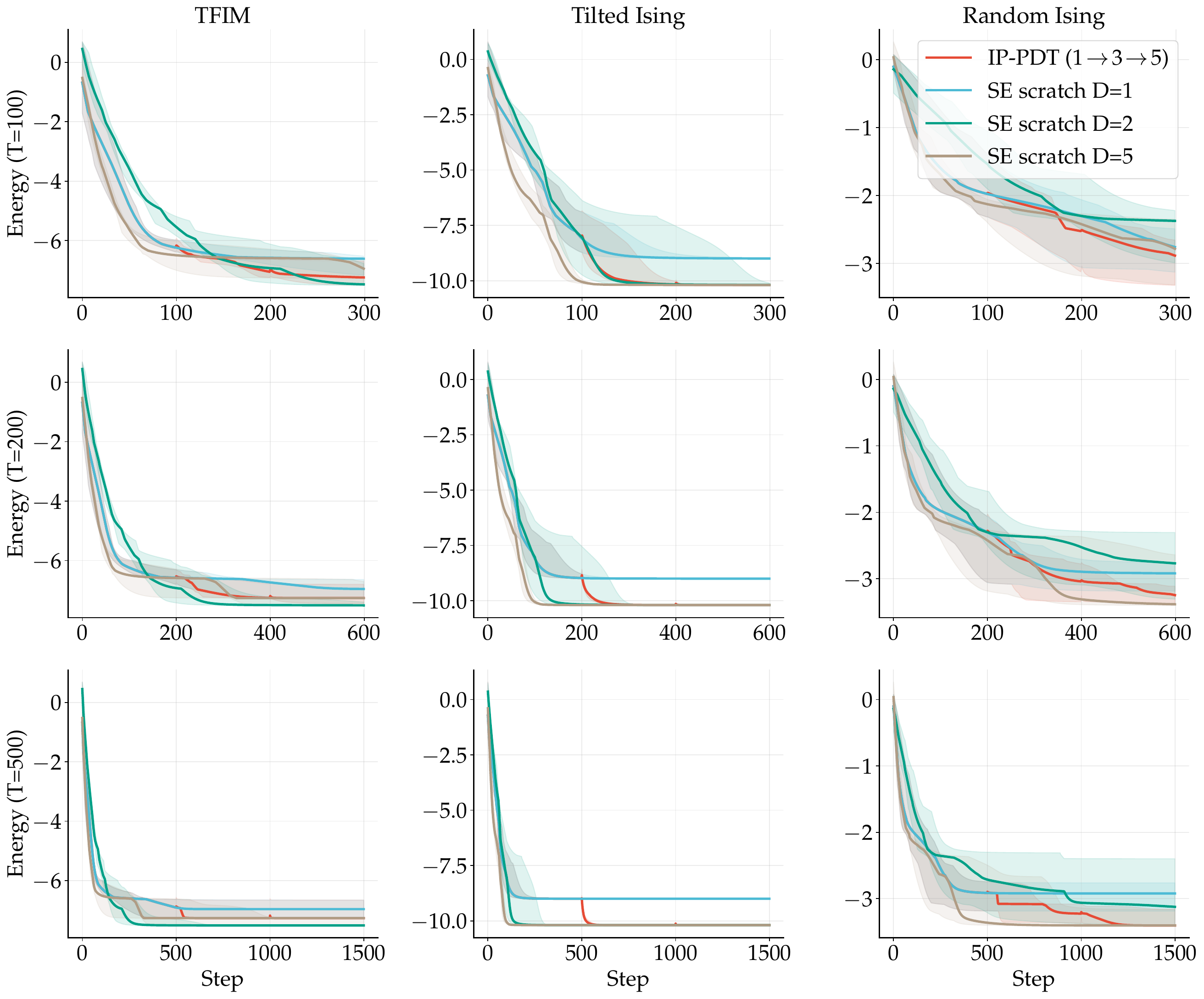}
\caption{Trainability beyond expressibility.  Rows: budget
$T\in\{100,200,500\}$.  Columns: Hamiltonians.  IP-PDT (red) matches
or outperforms all scratch depths on Random and Tilted Ising with
dramatically lower variance.  All methods share the same reachable
set $\mathcal{M}_2$.  Median $\pm$ IQR, $10$ seeds.}
\label{fig:trainability}
\end{figure}
\subsection{Cost-adjusted comparison}\label{sec:exp_cost}
Comparing methods by gradient steps flatters the SE ansatz, since its
single CNOT ring is five times cheaper than HEA's per layer. To level the field
we plot energy against cumulative single- and two-qubit gate count for
IP-PDT (SE), Baseline HEA, and Na\"ive PDT (HEA), each run for $300$
gradient steps (\Cref{fig:cost_adjusted}).

By that point HEA has spent $36{,}000$ gates to IP-PDT's $18{,}000$, yet
the cheaper method still reaches lower energy on all three primary
Hamiltonians---Tilted Ising ($-10.194$ vs.\ $-9.407$), Random Ising
($-2.919$ vs.\ $-2.770$), and TFIM ($-7.126$ vs.\ $-6.948$).  IP-PDT thus
wins on the metric that matters most for near-term devices: it does more
with half the gates.
\begin{figure}[t]
\centering
\includegraphics[width=\textwidth]{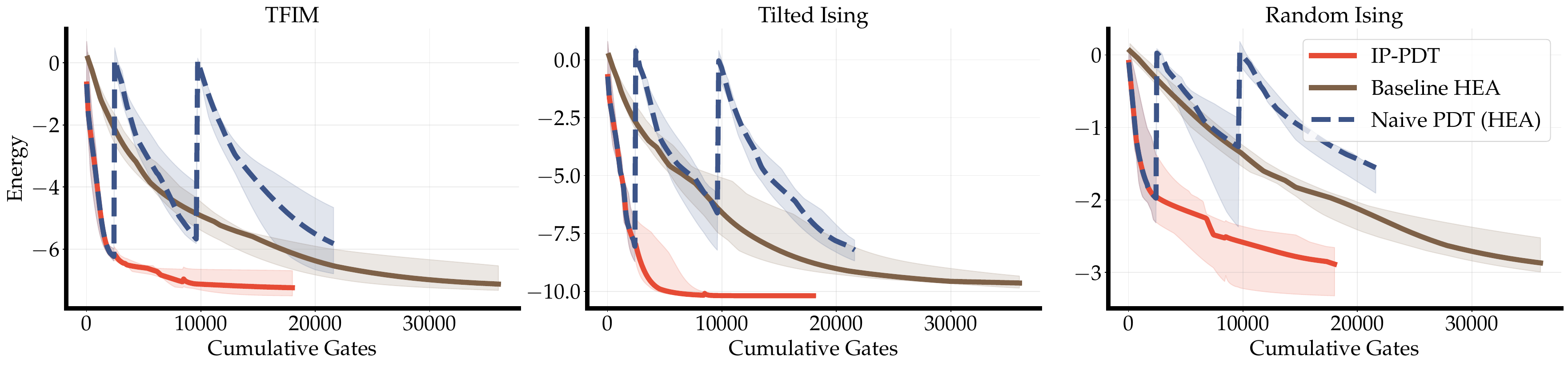}
\caption{Energy vs.\ cumulative gate cost.  IP-PDT (red, $18$K gates)
achieves lower energy than HEA (brown, $36$K gates) on all three
Hamiltonians.  Median $\pm$ IQR over $10$ seeds.}
\label{fig:cost_adjusted}
\end{figure}




\subsection{Extended Hamiltonian benchmark}\label{sec:exp_extended}
To see whether
the advantage is a property of IP-PDT rather than of a favorable test
problem, we run IP-PDT, Na\"ive PDT, Baseline HEA, and Baseline SE at
$T = 200$ ($3T = 600$ total steps) across all nine Hamiltonians of
\Cref{tab:hamiltonians_full}, reporting the energy gap $E-\lambda_1$ in
\Cref{tab:extended} and \Cref{fig:extended}.

IP-PDT takes the lowest or tied-lowest gap on eight of the nine.  The lone
loss is TFIM, whose highly entangled ground state once again exceeds the
SE ansatz's expressibility ceiling and lets HEA through.  Where the
ground state is within reach the margin can be wide: on Tilted Ising and
Mixed-field Ising IP-PDT drives the gap below $0.02$, against $0.3$ or
more for HEA and the scratch baseline.  The pattern across the table is
consistent---progressive depth growth helps broadly, and only genuine
entanglement demands, not the choice of Hamiltonian, hold it back.
\begin{table}[t]
\centering\small
\caption{Energy gap $E-\lambda_1$ (mean $\pm$ std, $10$ seeds) across
nine Hamiltonians.  \textbf{Bold} = best.  $\dagger$ = tied within
$1$ std.}
\label{tab:extended}
\begin{tabular}{@{}lcccc@{}}
\toprule
\textbf{Hamiltonian} & \textbf{IP-PDT} & \textbf{Na\"ive PDT} & \textbf{HEA} & \textbf{SE scratch} \\
\midrule
Tilted Ising     & $\mathbf{0.017}{\pm .006}^{\dagger}$ & $0.017{\pm .006}$ & $0.322{\pm .214}$ & $0.613{\pm 1.27}$ \\
Mixed-field      & $\mathbf{0.020}{\pm .008}^{\dagger}$ & $0.020{\pm .008}$ & $0.259{\pm .114}$ & $0.315{\pm .959}$ \\
Spin Glass       & $\mathbf{0.131}{\pm .066}$ & $0.156{\pm .101}$ & $0.556{\pm .279}$ & $0.472{\pm .774}$ \\
Random Ising     & $\mathbf{0.157}{\pm .135}$ & $0.166{\pm .135}$ & $0.333{\pm .256}$ & $0.369{\pm .456}$ \\
Power-law        & $\mathbf{0.186}{\pm .241}$ & $0.262{\pm .296}$ & $0.206{\pm .069}$ & $0.522{\pm .412}$ \\
Random XY        & $\mathbf{0.479}{\pm .295}$ & $0.628{\pm .466}$ & $0.640{\pm .224}$ & $0.540{\pm .215}$ \\
XXZ              & $\mathbf{0.525}{\pm .034}$ & $0.544{\pm .054}$ & $0.556{\pm .287}$ & $0.574{\pm .070}$ \\
Heisenberg       & $\mathbf{0.690}{\pm .229}$ & $0.760{\pm .322}$ & $0.769{\pm .278}$ & $0.943{\pm .338}$ \\
TFIM             & $0.574{\pm .369}$ & $0.574{\pm .369}$ & $\mathbf{0.304}{\pm .145}$ & $0.697{\pm .511}$ \\
\bottomrule
\end{tabular}
\end{table}
\begin{figure}[t]
\centering
\includegraphics[width=\textwidth]{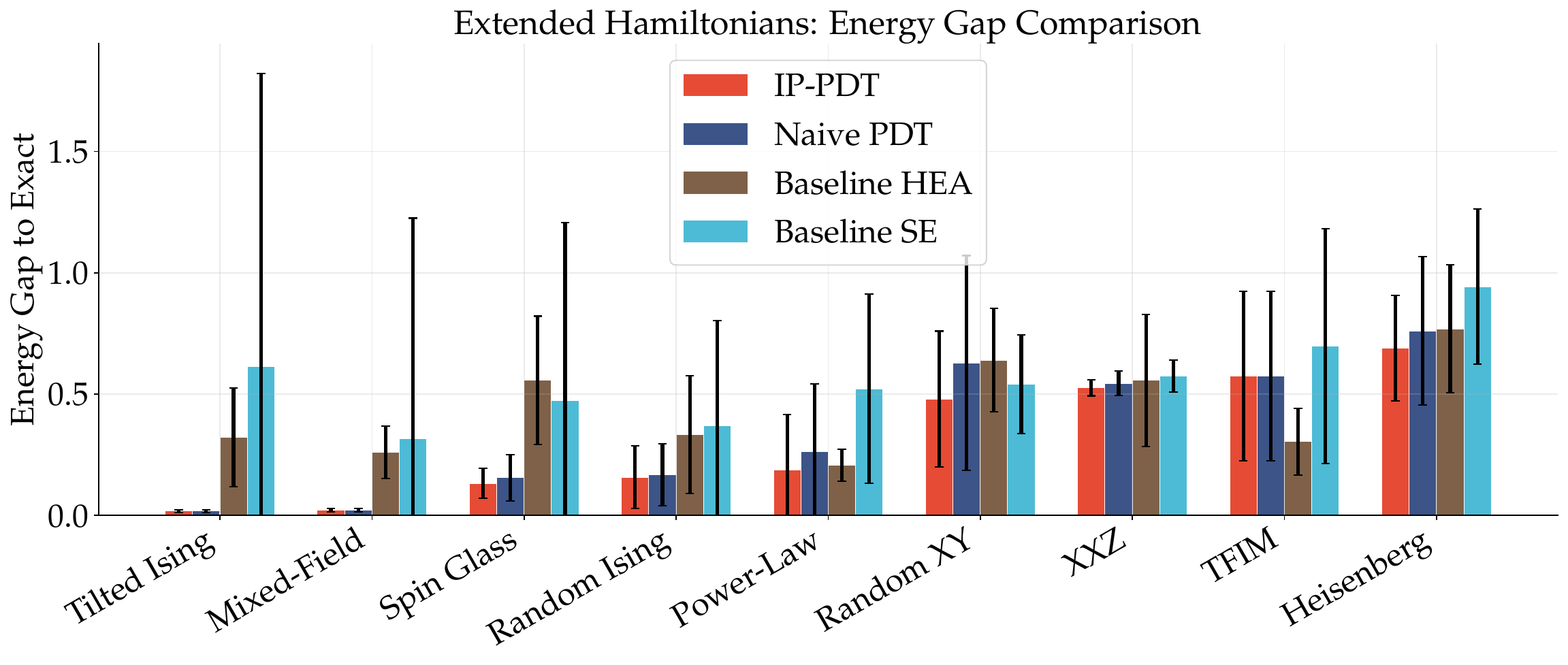}
\caption{Energy gap across nine Hamiltonians (lower is better).  IP-PDT
(red) wins or ties on $8/9$.  Error bars: $\pm 1$ std, $10$ seeds.}
\label{fig:extended}
\end{figure}
\subsection{Expressibility ceiling: SE vs.\ HEA}\label{sec:exp_ceiling}
The recurring TFIM exception raises a fair worry: is the SE ansatz's lone
CNOT ring a practical handicap against HEA's five?  One might expect the
more expressive ansatz to pull ahead whenever the ground state is hard to
reach, so we test exactly that, comparing IP-PDT (SE) against Baseline HEA
and Na\"ive PDT (HEA) on five Hamiltonians ranging from easy to demanding
in their entanglement (\Cref{fig:ceiling}, \Cref{tab:ceiling}).

IP-PDT outperforms HEA on four of the
five, ceding only TFIM, whose ground state truly lives beyond the SE
ceiling.  Even on Heisenberg---also nominally hard---IP-PDT edges HEA out,
$0.733$ vs.\ $0.769$.  The lesson we conjecture is that expressibility and trainability
trade against each other: a better-conditioned optimization landscape can
more than offset five times fewer entangling gates, except when the
target genuinely cannot be represented at all.
\begin{table}[h]
\centering\small
\caption{IP-PDT (SE) vs.\ HEA: energy gap (mean $\pm$ std, $10$
seeds).}
\label{tab:ceiling}
\begin{tabular}{@{}lccc@{}}
\toprule
\textbf{Hamiltonian} & \textbf{IP-PDT} & \textbf{HEA} & \textbf{Winner} \\
\midrule
Tilted Ising & $\mathbf{0.017\pm 0.006}$ & $0.322\pm 0.203$ & IP-PDT \\
Random Ising & $\mathbf{0.162\pm 0.127}$ & $0.333\pm 0.243$ & IP-PDT \\
Spin Glass   & $\mathbf{0.147\pm 0.090}$ & $0.556\pm 0.265$ & IP-PDT \\
Heisenberg   & $\mathbf{0.733\pm 0.287}$ & $0.769\pm 0.264$ & IP-PDT \\
TFIM         & $0.574\pm 0.350$ & $\mathbf{0.304\pm 0.137}$ & HEA \\
\bottomrule
\end{tabular}
\end{table}
\begin{figure}[t]
\centering
\includegraphics[width=\textwidth]{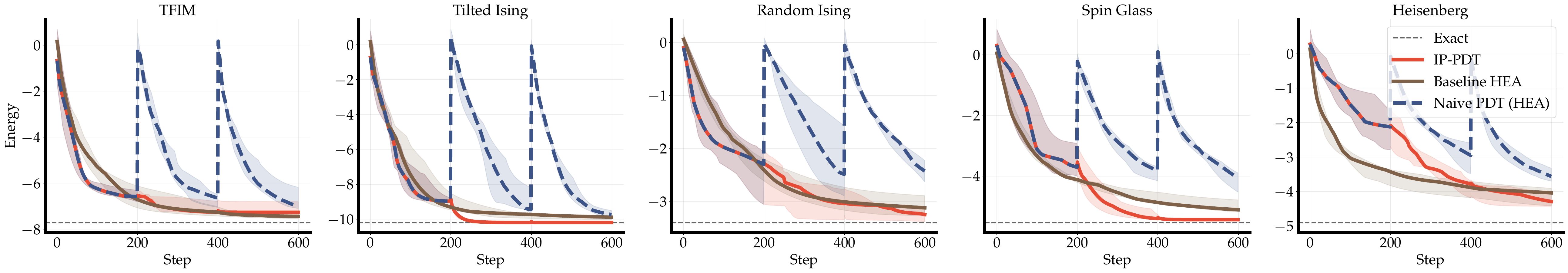}
\caption{IP-PDT (SE, $1$ CNOT ring) vs.\ HEA ($5$ CNOT rings) and
Na\"ive PDT (HEA).  IP-PDT (red) outperforms HEA (brown) on $4/5$
Hamiltonians despite $5\times$ fewer entangling gates.  Median $\pm$
IQR, $10$ seeds.}
\label{fig:ceiling}
\end{figure}

\subsection{Comparison with SOTA methods}\label{sec:exp_sota}
We
budget-match IP-PDT at $300$ function evaluations against five
alternatives: SE scratch ($D = 5$), HEA ($D = 5$),
Rotosolve~\cite{ostaszewski2021rotosolve}, Layerwise
training~\cite{skolik2021layerwise}, and a simplified
ADAPT-VQE~\cite{grimsley2019adaptive} (\Cref{fig:sota}).

IP-PDT finishes first on TFIM (gap $0.60$) and Random Ising ($0.49$) and a
close second on Tilted Ising ($0.020$ against SE scratch's $0.015$).  The
established VQE optimizers fare markedly worse: Rotosolve, Layerwise, and
ADAPT-VQE trail on all three Hamiltonians with energy gaps two to five
times larger than IP-PDT's.
\begin{figure}[t]
\centering
\includegraphics[width=\textwidth]{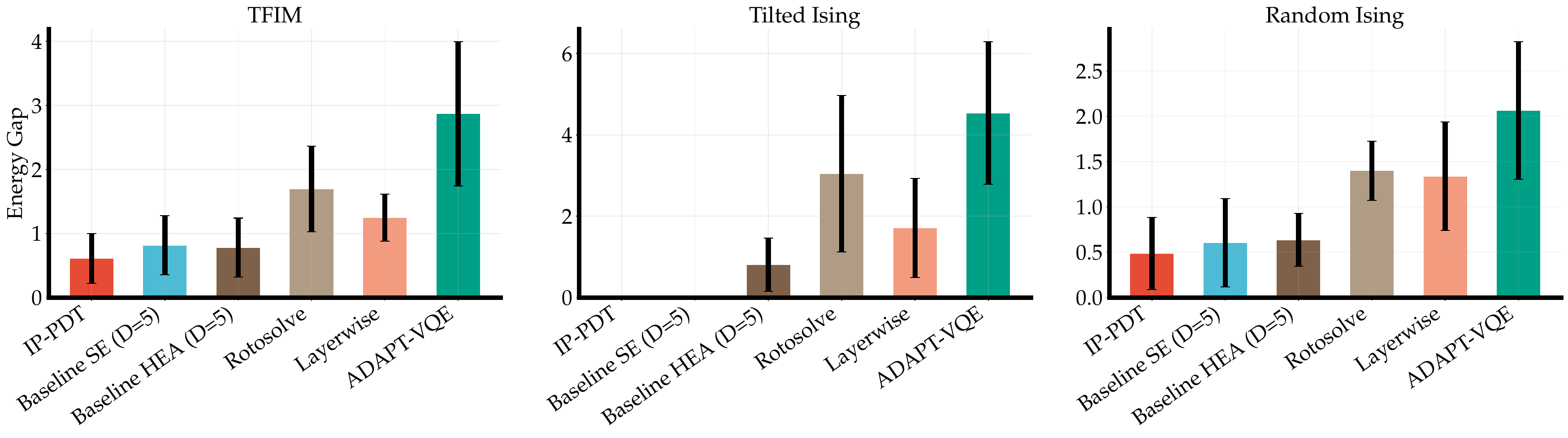}
\caption{SOTA comparison ($300$ evaluations, $10$ seeds).  IP-PDT
(red) is competitive with fixed-depth baselines and dominates
Rotosolve, Layerwise, and ADAPT-VQE on all three Hamiltonians.}
\label{fig:sota}
\end{figure}
\subsection{SOTA methods at multiple budgets}\label{sec:exp_sota_budget}
A single budget can hide a method that wins only when starved or only when
indulged, so we repeat the comparison across the full range, sweeping
$T\in\{50,100,200,500\}$ for all six methods of \Cref{sec:exp_sota}
(\Cref{fig:sota_budget}).

IP-PDT's edge is sharpest exactly where near-term hardware lives---at low
budgets.  At $T = 50$ it already posts the lowest energy on all three
Hamiltonians, and on Tilted Ising it reaches near-exact energy
($-10.19$) by $T = 100$ while Layerwise and Rotosolve are still stuck
above $-8.0$ at $T = 500$.  On Random Ising the advantage is consistent
across every budget.  Only on TFIM does a competitor catch up: Baseline
SE closes the gap at $T = 500$ as both methods press against the
expressibility ceiling.  Rotosolve and ADAPT-VQE lag at every budget
tested.
\begin{figure}[t]
\centering
\includegraphics[width=\textwidth]{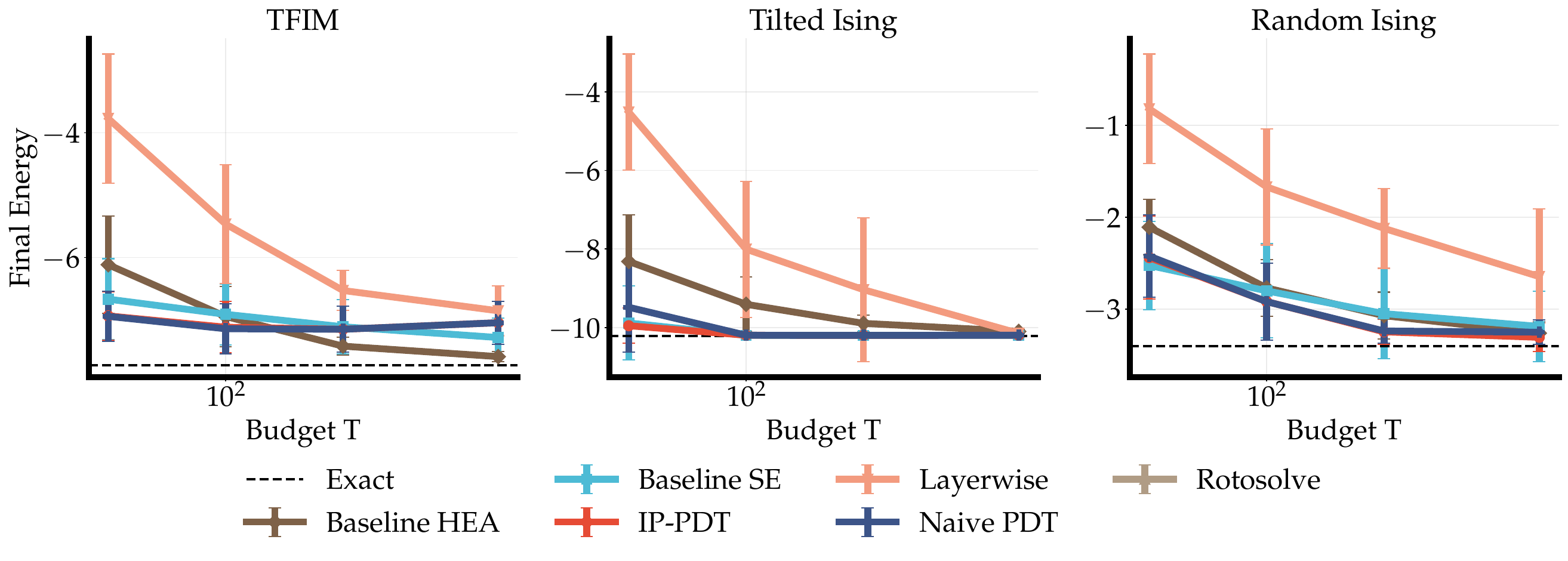}
\caption{SOTA methods at multiple budgets.  IP-PDT (red) dominates at
low budgets; Baseline SE catches up at $T = 500$ on TFIM.  Layerwise
and Rotosolve lag behind at all budgets.  Mean $\pm$ std, $10$
seeds.  Dashed line: exact ground state.}
\label{fig:sota_budget}
\end{figure}
\subsection{Growth schedule comparison}\label{sec:exp_schedules}
Having argued that gradual growth matters, we should ask how much the
particular schedule matters.  We compare four progressive
schedules---$\{1{\to}5\}$ (a single jump), $\{1{\to}3{\to}5\}$ (our
default), $\{1{\to}2{\to}3{\to}4{\to}5\}$ (gradual), and
$\{1{\to}3{\to}5{\to}7{\to}9\}$ (extended depth)---alongside a scratch
$D = 5$ baseline where available, all at $T = 200$ steps per stage
(\Cref{fig:schedules}).

When the landscape is benign the choice barely registers: on Tilted Ising
every progressive schedule converges to the same energy.  The
distinctions emerge on rugged problems.  On Random Ising and Spin Glass
the extended schedule $\{1{\to}3{\to}5{\to}7{\to}9\}$ reaches the lowest
final energy, evidence that extra overparameterization stages pay off
where the landscape is hard, while the single-jump $\{1{\to}5\}$ is the
worst performer on every Hamiltonian, confirming that gradual growth is
not incidental.  On Heisenberg all progressive schedules beat scratch
$D = 5$ (dashed), one more reminder that the trainability advantage
survives a change of schedule.
\begin{figure}[t]
\centering
\includegraphics[width=\textwidth]{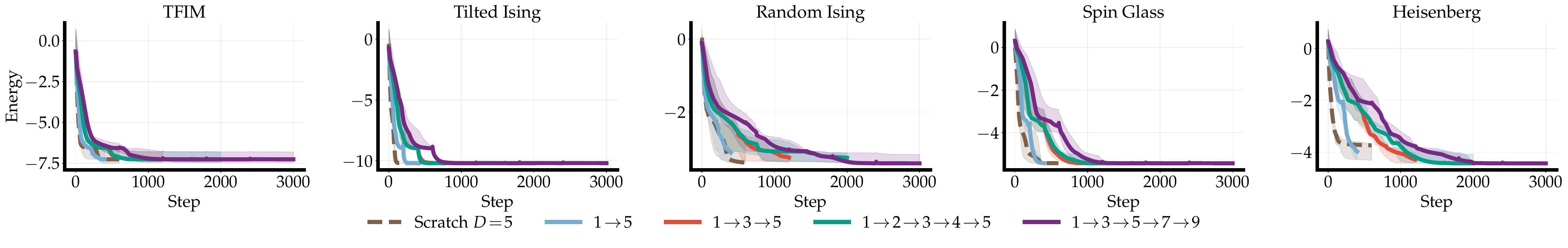}
\caption{Growth schedule comparison across five Hamiltonians
($T = 200$ per stage, $10$ seeds).  Gradual growth outperforms the
single-jump $1{\to}5$.  Extended depth
($1{\to}3{\to}5{\to}7{\to}9$) helps on rugged landscapes.  Scratch
$D = 5$ (dashed brown, where available) underperforms all progressive
schedules.  Median $\pm$ IQR.}
\label{fig:schedules}
\end{figure}



\subsection{Hyperparameter sensitivity}\label{sec:exp_hyperparams}
A method is only as useful as it is forgiving of its knobs, so we close
the main study by stress-testing the two that matter most, the learning
rate $\eta$ and the steps per stage $T$, sweeping both on a grid for each
growth schedule and recording the mean energy gap over seeds and
Hamiltonians (\Cref{fig:hyperparams}).

The heatmap shows a wide, flat optimum.  IP-PDT is robust across the whole
range $\eta\in[0.01,0.1]$ and degrades only at extreme learning rates,
while increasing $T$ improves performance monotonically with diminishing
returns past $T = 200$---precisely the saturating behavior the
exponential convergence rate of \Cref{thm:basin} would predict.
\begin{figure}[t]
\centering
\includegraphics[width=\textwidth]{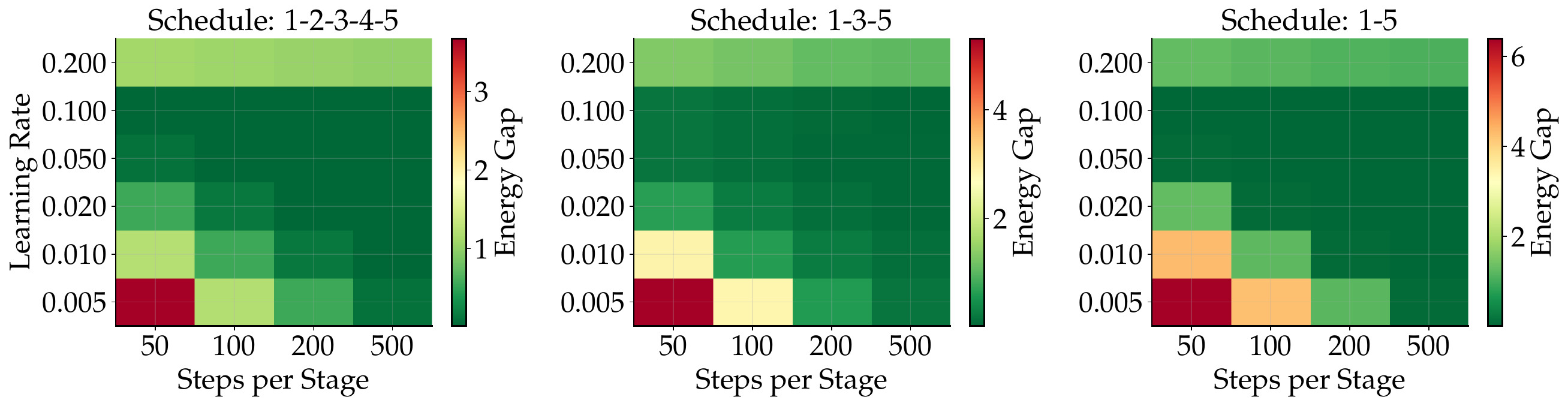}
\caption{Hyperparameter sensitivity: energy gap (color) as a function
of learning rate ($y$-axis) and steps per stage ($x$-axis).  Lower
(greener) is better.  IP-PDT is robust across a wide range of
$\eta$ and $T$.}
\label{fig:hyperparams}
\end{figure}
\subsection{Ablations}\label{sec:ablation}
Two loose ends remain, and we tie them off here.\label{sec:exp_jitter}
The first is the suspicion, raised repeatedly above, that the small
parameter jitter in our initialization is incidental rather than
essential.  To settle it we sweep the jitter scale
$\sigma\in\{0,0.01,0.05,0.1,0.5\}$ and compare against
zero-initialization, in which every new parameter is set exactly to zero,
and against scratch-$D = 5$ (\Cref{fig:jitter}).  All the progressive
variants perform nearly identically regardless of $\sigma$---including the
$\sigma = 0$ case---and all of them beat scratch.  The conclusion: it is the progressive depth curriculum, not the jitter, that
drives IP-PDT's advantage, so the method has essentially no tuning to get
wrong here.

The second loose end is a practical one: on real hardware we measure
energy, not fidelity, so can the energy gap itself certify how close we
are to the true ground state?  \Cref{tab:fidelity} places the fidelity
$F=|\langle \mathbf{v}_1|\boldsymbol{\psi}\rangle|^2$ next to the gap $\varepsilon = E -
\lambda_1$, and the gap--fidelity sandwich of \Cref{thm:gap_fidelity},
which bounds $\Delta(1{-}F)\leq\varepsilon$, holds on all $30$ runs ($10$
per Hamiltonian).  The certificate is tight where the ansatz is
expressive enough---Tilted Ising reaches $F=0.997$---and loose where it is
not, with TFIM and Random Ising showing lower fidelity that traces
straight back to the SE ceiling.  
\begin{table}[h]
\centering\small
\caption{Ground-state fidelity and energy gap (mean $\pm$ std, $10$
seeds).  The gap--fidelity sandwich (\Cref{thm:gap_fidelity}) is
satisfied on all seeds.}
\label{tab:fidelity}
\begin{tabular}{@{}lccc@{}}
\toprule
\textbf{Hamiltonian} & $F$ & $\varepsilon=E-\lambda_1$ & \textbf{Sandwich valid} \\
\midrule
Tilted Ising & $0.997\pm 0.002$ & $0.020\pm 0.009$ & 10/10 \\
TFIM         & $0.437\pm 0.313$ & $0.602\pm 0.386$ & 10/10 \\
Random Ising & $0.295\pm 0.448$ & $0.486\pm 0.397$ & 10/10 \\
\bottomrule
\end{tabular}
\end{table}

\begin{figure}[t]
\centering
\includegraphics[width=\textwidth]{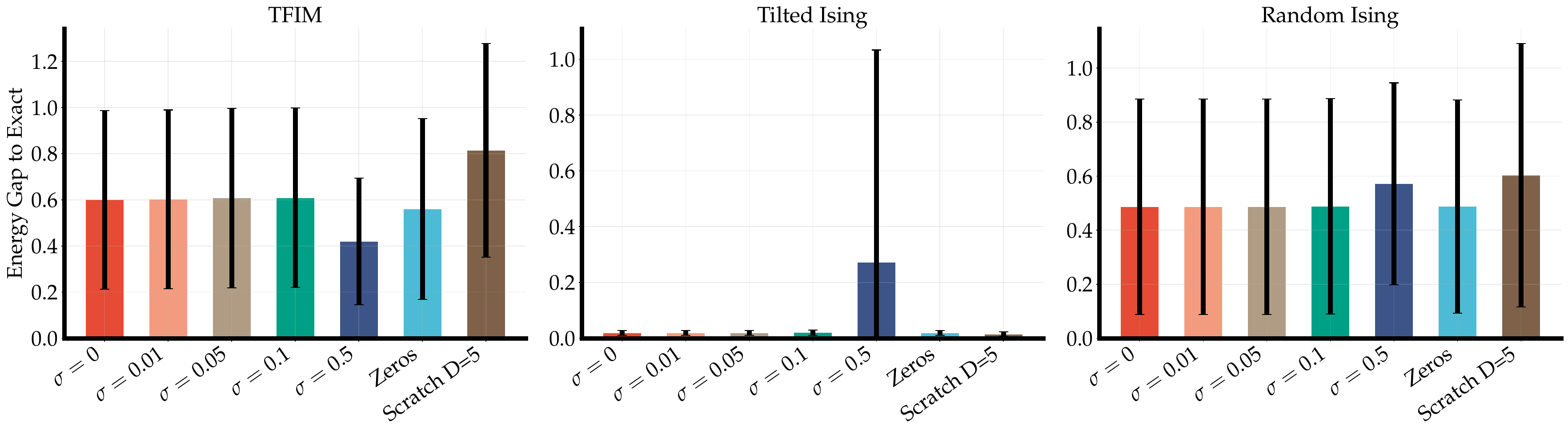}
\caption{Jitter ablation: energy gap for different $\sigma$ values and initialization strategies.  All progressive variants perform similarly regardless of $\sigma$; scratch-$D = 5$ is consistently worse.  Mean $\pm$ std, $10$ seeds.} \label{fig:jitter}
\end{figure}

\subsection{Scaling to $n=16$}\label{sec:exp_n16}
To assess whether IP-PDT scales beyond the $n=6$ setting used in the
preceding experiments, we re-run the extended Hamiltonian benchmark at
$n=16$.  The Hilbert space dimension grows from $d=64$ to $d=65{,}536$,
so the dense $\mathbf{H}\in\bbC^{d\times d}$ representation used at $n=6$ is no
longer practical; we use a Pauli-sum representation for $\mathbf{H}$ and the
Lanczos algorithm (\texttt{scipy.sparse.linalg.eigsh}) to compute exact
ground-state energies.  All other protocol parameters match
\Cref{sec:exp_extended}: stage depths $\{1,3,5\}$, $T=200$ steps per
stage, $\eta=0.02$, $\sigma=0.05$, $10$ seeds.
\textbf{Three findings emerge.}
\emph{HEA fails to scale.}  On every Hamiltonian, the depth-$5$
hardware-efficient ansatz produces energy gaps that are $2$--$10\times$
larger than the single-entangler methods (\Cref{tab:extended_n16}).
At $n=6$, HEA was competitive with IP-PDT on most Hamiltonians; at
$n=16$, it lags every other method on every Hamiltonian.  This is
consistent with the well-documented difficulty of training deep
hardware-efficient circuits as the qubit count grows---the regime in
which barren plateaus~\cite{mcclean2018barren} and trap-dominated
landscapes~\cite{anschuetz2022beyond} are expected to appear.
\emph{Progressive methods remain competitive with full-depth SE.}
IP-PDT and Na\"ive PDT achieve comparable gaps to the full-depth SE
baseline on 6 of 9 Hamiltonians, despite starting from a shallower
circuit.  The progressive curriculum scales.
\emph{The Eq.~\eqref{eq:ippdt_expand1} initialization and the na\"ive
copy initialization converge at scale.}  At $n=6$, IP-PDT (with the
trained-first-block initialization, Eq.~\ref{eq:ippdt_expand1}) and
Na\"ive PDT (with \textcolor{black}{copy} initialization) produce nearly identical
numerics on most Hamiltonians; at $n=16$ this convergence is even
tighter, with IP-PDT wins of $\le 0.1$ on three Hamiltonians (Spin
Glass, Power-law, XXZ) and Na\"ive-PDT-wins of $\le 0.3$ on three others.
This suggests that the jitter centre matters less at scale.

\textcolor{black}{We note one cautionary observation on the error bars.
On Tilted Ising at $n=16$, IP-PDT's larger mean and variance relative to
Na\"ive PDT in the $10$-seed table above are \emph{not} the artifact of a
single outlier seed: a $40$-seed robustness check finds IP-PDT falling into
a poor basin (final gap $>1$) on $4/40 \approx 10\%$ of seeds, a rate that
copy-initialized Na\"ive PDT essentially matches ($3/40 \approx 7.5\%$),
with \emph{identical} medians ($0.06$).  An occasional poor basin is
therefore a property of the single-entangler \emph{landscape} at this
scale, shared by both progressive methods, not an IP-PDT-specific
disadvantage.   Both progressive methods,
moreover, beat from-scratch SE, which fails on $16/40 \approx 40\%$ of
seeds---a $\approx 4\times$ lower failure rate.  On Spin Glass the large
spread is likewise a rugged-landscape property shared by all methods
(per-seed std $\approx 3$--$4$ across IP-PDT, Na\"ive PDT, and the SE
baseline alike); there IP-PDT's median gap ($\approx 0.65$) is in fact
\emph{better} than both Na\"ive PDT's ($\approx 3.3$) and the from-scratch
SE baseline's ($\approx 4.5$).}
Conversely, on Power-law Ising at $n=16$, IP-PDT
achieves $0.025\pm 0.001$---the tightest standard deviation of any
method on any Hamiltonian in our benchmark.
See also Figure \ref{fig:extended_n16}.

\begin{table}[t]
\centering\small
\caption{Energy gap $E-\lambda_1$ at $n=16$ (mean $\pm$ std, $10$ seeds).
\textbf{Bold} = best per Hamiltonian.}
\label{tab:extended_n16}
\begin{tabular}{@{}lcccc@{}}
\toprule
\textbf{Hamiltonian} & \textbf{IP-PDT} & \textbf{Na\"ive PDT} & \textbf{SE Baseline} & \textbf{HEA} \\
\midrule
Tilted Ising     & $0.356{\pm .955}$ & $\mathbf{0.055}{\pm .017}$ & $0.624{\pm 1.27}$ & $3.684{\pm 2.28}$ \\
Mixed-field      & $\mathbf{0.065}{\pm .020}^{\dagger}$ & $0.065{\pm .020}$ & $0.622{\pm 1.28}$ & $3.357{\pm 2.22}$ \\
Spin Glass       & $\mathbf{3.449}{\pm 4.12}$ & $4.012{\pm 3.89}$ & $4.814{\pm 2.93}$ & $6.003{\pm 4.78}$ \\
Random Ising     & $2.565{\pm 1.76}$ & $2.777{\pm 1.75}$ & $\mathbf{2.102}{\pm .81}$ & $5.861{\pm 2.94}$ \\
Power-law        & $\mathbf{0.025}{\pm .001}^{\dagger}$ & $0.025{\pm .001}$ & $1.500{\pm 1.10}$ & $3.103{\pm 3.41}$ \\
Random XY        & $1.917{\pm .69}$ & $\mathbf{1.633}{\pm .48}$ & $2.163{\pm 1.43}$ & $11.44{\pm 5.76}$ \\
XXZ              & $\mathbf{1.548}{\pm .82}$ & $1.763{\pm .97}$ & $1.553{\pm .73}$ & $4.184{\pm 2.35}$ \\
Heisenberg       & $5.140{\pm 1.69}$ & $4.892{\pm 1.49}$ & $\mathbf{4.204}{\pm 1.71}$ & $9.444{\pm 4.72}$ \\
TFIM             & $1.654{\pm .66}$ & $1.710{\pm .70}$ & $\mathbf{1.218}{\pm .59}$ & $1.924{\pm .77}$ \\
\bottomrule
\end{tabular}
\end{table}
\begin{figure}[t]
\centering
\includegraphics[width=\textwidth]{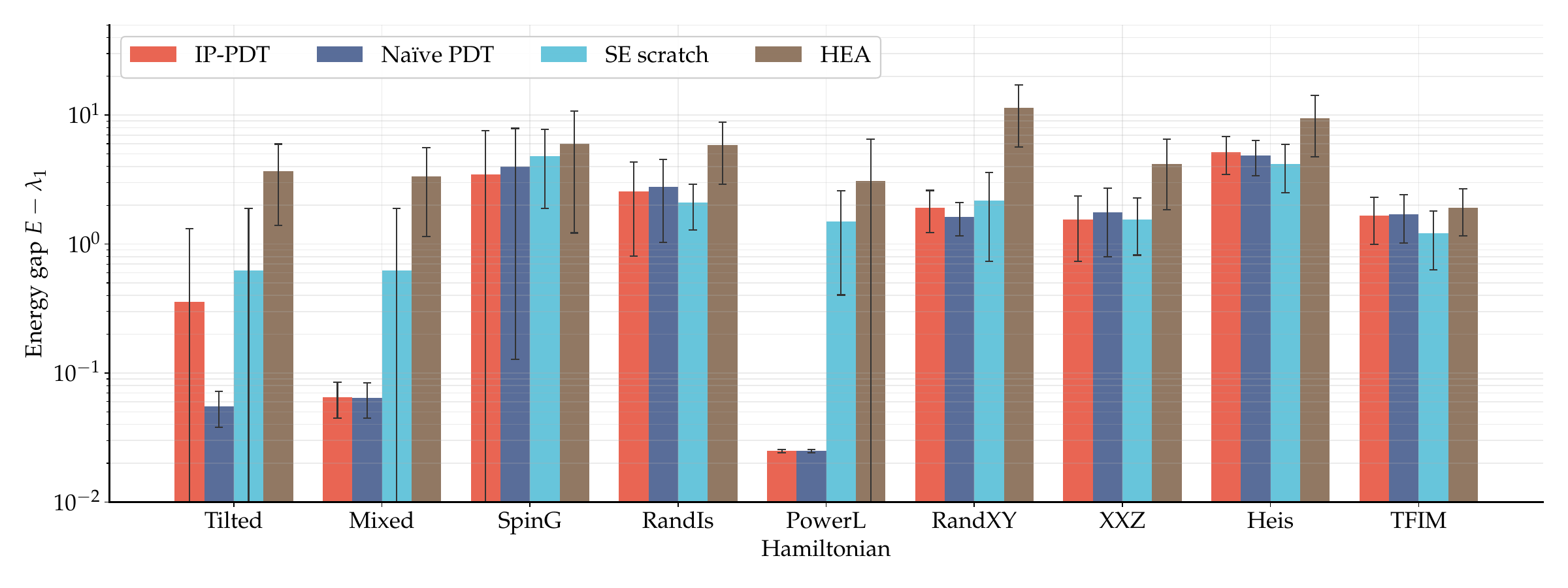}
\caption{Extended Hamiltonian benchmark at $n=16$.  HEA (brown) lags
every other method on every Hamiltonian, consistent with barren-plateau
onset.  IP-PDT (red) wins or ties on $4/9$; Na\"ive PDT on $2/9$;
SE-scratch baseline on $3/9$.  Error bars: $\pm 1$ std, $10$ seeds.
Y-axis is log-scale to accommodate the HEA range.}
\label{fig:extended_n16}
\end{figure}

\section{Conclusion}\label{sec:conclusion}
We introduced IP-PDT, a layerwise growth strategy that eliminates
initialization shock by appending forward/reverse block pairs whose CNOT
gates cancel.  The Reachable Set Saturation Theorem reveals that the
resulting circuit saturates its variational manifold after one expansion;
all subsequent depth increases are pure overparameterization of local
rotations---yet progressive training continues to yield optimization
benefits.  A resource analysis shows that IP-PDT achieves the lowest
total gate cost among the three protocols.  Experiments on three 6-qubit
Ising Hamiltonians show that IP-PDT matches or outperforms full-depth
HEA baselines despite using $5\times$ fewer CNOT gates, with the
strongest gains under limited budgets and on Hamiltonians whose ground
states are well-approximated by the single-entangler reachable set.
For such Hamiltonians, \emph{trainability gains can persist beyond
expressibility gains}, and structured initialization compensates for
reduced entangling power.  When the expressibility ceiling is itself
the binding constraint (as may occur for TFIM), deeper entangling
architectures remain necessary.
An immediate open question is whether the identity-pair construction
can be extended to $K>1$ backbone entangling layers, interpolating
between the single-entangler regime studied here and the full HEA,
and whether the Saturation Theorem generalizes to a $K$-dependent
hierarchy of reachable sets.
A second, orthogonal direction is to compose our warm-started
continuation in \emph{depth} space with continuation in \emph{Hamiltonian}
space: \v{Z}unkovi\v{c} et al.~\cite{zunkovic2026adiabatic} warm-start VQE
along an adiabatic path of Hamiltonians, so growing circuit depth while
adiabatically annealing the Hamiltonian would combine two warm-started
paths on independent axes---a combination we leave to future work.

Some limitations we have identified:
\begin{enumerate}[leftmargin=1.5em, nosep]
  \item \textbf{No general barren-plateau escape.}  The enlarged
    single-entangler gradient variance is confined to \emph{local} cost
    functions and is partly a shallow-depth effect shared with a shallow
    HEA; for \emph{global} cost functions all ans\"atze flatten
    identically and the single-entangler circuit confers no
    advantage~\cite{cerezo2021cost}.  IP-PDT is therefore not a universal
    remedy for barren plateaus.
  \item \textbf{Expressibility ceiling.}  For strongly correlated systems
    whose ground state is not $\cM_2$-structured, the SE circuit may be
    insufficient.  Our TFIM results confirm this: SE wins when the target
    lies in $\cM_2$ (including the highly entangled GHZ-like cat at
    $h_x = 0.5$) and loses near criticality where it does not.  On
    dense-coupling Hamiltonians (random spin glass with transverse field,
    random Heisenberg) the full-depth baseline can outperform the
    progressive strategies.
  \item \textbf{Single entangling layer.}  $\cM_3$ is limited by one
    CNOT ring.  A natural extension: $K>1$ ``backbone'' entangling
    layers with rotation-only growth.
\end{enumerate}

\section*{Acknowledgements}
This research was funded in part by: Rice University (Faculty Initiative
award); NSF FET:Small (award no.\ 1907936); NSF CAREER (award no.\ 2145629); a
Microsoft Research Award. AK would also like to thank the Ken Kennedy Institute
at Rice University for its support through the Research Cluster ``QuanTAS''.

\bibliographystyle{plain}
\nocite{quetschlich2023mqtbench,sakurai2020modern,arora2018optimization,zhou2020qaoa}
\bibliography{references}
\clearpage
\appendix
\section{Formal Proofs}\label{app:proofs}
\subsection{$\SU(2)$ coverage}\label{app:su2}
\begin{lemma}\label{lem:su2_cover}
  The map $\bbR^3\ni(\alpha,\beta,\gamma)\mapsto
  \mathbf{R}_X(-\gamma) \mathbf{R}_Z(-\beta) \mathbf{R}_Y(-\alpha)$ is surjective onto $\SU(2)$.
\end{lemma}
\begin{proof}
We argue in two parts: (a) the forward $Y ZX$ axis sequence covers all
of $\SU(2)$; (b) negating the angles
covers the same set, which yields the stated $\mathbf{R}^{\mathrm{inv}}$ form.

\emph{(a) Surjectivity of the $Y ZX$ Euler sequence.}
We start from the standard $ZY Z$ Euler decomposition: for every
$\mathbf{U}\in\SU(2)$ there exist angles $\psi,\theta,\varphi$ with
\begin{equation}\label{eq:zyz}
  \mathbf{U} = \mathbf{R}_Z(\psi) \mathbf{R}_Y(\theta) \mathbf{R}_Z(\varphi),
  \qquad \psi,\varphi\in\bbR,\ \theta\in[0,\pi],
\end{equation}
which is surjective because writing
$\mathbf{R}_Z(\xi)=\mathrm{diag}(e^{-i\xi/2},e^{i\xi/2})$ and
$\mathbf{R}_Y(\theta)=\bigl(\begin{smallmatrix}
\cos\frac\theta2 & -\sin\frac\theta2\\
\sin\frac\theta2 & \cos\frac\theta2\end{smallmatrix}\bigr)$
and equating with a general
$\mathbf{U}=\bigl(\begin{smallmatrix} u & -\bar v\\ v & \bar u
\end{smallmatrix}\bigr)$, $|u|^2+|v|^2=1$, the four real degrees of
freedom of $\SU(2)$ are matched by choosing $\theta$ from
$\cos\frac\theta2=|u|$, then $\psi,\varphi$ from the phases of $u,v$:
the half-angle combinations $(\psi{+}\varphi)/2$ and
$(\psi{-}\varphi)/2$ must equal $-\arg u$ and $\arg v$ modulo $2\pi$,
which is always solvable as $\psi,\varphi$ range over $\bbR$ (each
angle over a full $4\pi$ period; over $[0,2\pi)^2$ only half the phase
pairs are attainable, by the same double-cover sign discussed below). This is the standard Euler-angle construction for
$\SU(2)$; see \cite{nielsen2010quantum}, Theorem~4.1.

To convert~\eqref{eq:zyz} into a $Y ZX$ sequence we conjugate single
axes by the fixed Clifford rotations that permute the Pauli axes. Let
$\mathbf{S}=\mathbf{R}_X(\pi/2)$; from $\mathbf{R}_\sigma(\xi)=\exp(-i\tfrac\xi2\sigma)$ and
$\mathbf{S} \mathbf{Z} \mathbf{S}^\dagger=\mathbf{Y}$ (a direct $2\times2$ computation, since
conjugation by a $\tfrac\pi2$ $X$-rotation rotates the Bloch axis
$\hat z\mapsto\hat y$) we get, for all $\xi$,
\begin{equation}\label{eq:conj_zy}
  \mathbf{S} \mathbf{R}_Z(\xi) \mathbf{S}^\dagger
  = \exp \bigl(-i\tfrac\xi2 \mathbf{S}\mathbf{Z}\mathbf{S}^\dagger\bigr)
  = \mathbf{R}_Y(\xi),
\end{equation}
using that conjugation commutes with the matrix exponential. Likewise,
with $\mathbf{T}=\mathbf{R}_Y(\pi/2)$ and $\mathbf{T} \mathbf{Z} \mathbf{T}^\dagger=\mathbf{X}$ we obtain
$\mathbf{T} \mathbf{R}_Z(\xi) \mathbf{T}^\dagger=\mathbf{R}_X(\xi)$, so the three coordinate
one-parameter rotation subgroups are mutually conjugate in $\SU(2)$.

The product $\mathbf{R}_Y(\alpha) \mathbf{R}_Z(\beta) \mathbf{R}_X(\gamma)$ is built from
rotations about three \emph{distinct} axes whose \emph{consecutive} pairs
are non-parallel: $\hat y\neq\pm\hat z$ and $\hat z\neq\pm\hat x$.  This is
precisely a \emph{Tait--Bryan} (generalized-Euler) axis sequence, and the
classical generalized-Euler-angle theorem states that \emph{any} such
sequence parametrizes the entire rotation group $\mathrm{SO}(3)$.  Via the
surjective double cover $\SU(2)\twoheadrightarrow\mathrm{SO}(3)$, under which
$\mathbf{R}_{\hat n}(\xi)$ covers the rotation by angle $\xi$ about axis $\hat n$,
a Tait--Bryan solution therefore reaches any given $\mathbf{U}\in\SU(2)$ up to
the covering sign, $\mathbf{R}_Y(\alpha)\mathbf{R}_Z(\beta)\mathbf{R}_X(\gamma)=\pm\mathbf{U}$;
since $\mathbf{R}_\sigma(\theta+2\pi)=-\mathbf{R}_\sigma(\theta)$ (rotation gates are
$4\pi$-periodic), shifting any one angle by $2\pi$ absorbs the sign
\cite{nielsen2010quantum}.  The conjugation identities~\eqref{eq:conj_zy}
serve only to verify the non-parallel-axis hypothesis (the three subgroups
are genuinely distinct one-parameter groups); we do \emph{not} claim the
fixed conjugators $\mathbf{S},\mathbf{T}$ telescope through the product---they do
not---and we do not rely on a degree-of-freedom count, which would
establish only local surjectivity.  Consequently
\begin{equation}\label{eq:yzx_surj}
 \{\mathbf{R}_Y(\alpha) \mathbf{R}_Z(\beta) \mathbf{R}_X(\gamma)
    : (\alpha,\beta,\gamma)\in\bbR^3\} = \SU(2);
\end{equation}
that is, the forward map $\mathbf{R}^{\mathrm{fwd}}$ of~\eqref{eq:fwd_rot} is
surjective onto $\SU(2)$.  (The angle triple is unique only away from the
measure-zero gimbal locus $\theta\in\{0,\pi\}$; surjectivity, which is all
we need, holds on all of $\SU(2)$.)

\emph{(b) Negating the angles.}
The inverse map of~\eqref{eq:inv_rot} satisfies, by
$\mathbf{R}_\sigma(\xi)^\dagger=\mathbf{R}_\sigma(-\xi)$ and reversal of a product
under the adjoint,
\begin{equation}\label{eq:inv_is_adjoint}
  \mathbf{R}^{\mathrm{inv}}(\alpha,\beta,\gamma)
  = \mathbf{R}_X(-\gamma) \mathbf{R}_Z(-\beta) \mathbf{R}_Y(-\alpha)
  = \bigl[\mathbf{R}_Y(\alpha) \mathbf{R}_Z(\beta) \mathbf{R}_X(\gamma)\bigr]^\dagger
  = \bigl[\mathbf{R}^{\mathrm{fwd}}(\alpha,\beta,\gamma)\bigr]^\dagger.
\end{equation}
Thus the image of the $\mathbf{R}^{\mathrm{inv}}$ map is
$\{\mathbf{U}^\dagger:\mathbf{U}\in\mathrm{image}(\mathbf{R}^{\mathrm{fwd}})\}
=\{\mathbf{U}^\dagger:\mathbf{U}\in\SU(2)\}=\SU(2)$, where the last equality holds
because $\SU(2)$ is a group, hence closed under the (bijective)
adjoint/inverse operation $\mathbf{U}\mapsto\mathbf{U}^\dagger=\mathbf{U}^{-1}$. Therefore
the map $(\alpha,\beta,\gamma)\mapsto\mathbf{R}^{\mathrm{inv}}(\alpha,\beta,\gamma)$
is surjective onto $\SU(2)$, as claimed. (Equivalently: replacing
$(\alpha,\beta,\gamma)$ by $(-\alpha,-\beta,-\gamma)$ in
$\mathbf{R}^{\mathrm{fwd}}$ and reversing the factor order produces every
adjoint, and the adjoints of $\SU(2)$ are again all of $\SU(2)$.)
\end{proof}
\subsection{Perturbation bound derivation}\label{app:perturb}
We derive the energy perturbation bound~\eqref{eq:perturb} stated in
\Cref{sec:continuation}.  Throughout, $\|\cdot\|$ denotes the spectral
(operator) norm, which is submultiplicative, satisfies the triangle
inequality, and is invariant under conjugation by a unitary
($\|\mathbf{V}\mathbf{M}\mathbf{V}^\dagger\|=\|\mathbf{M}\|$ for unitary
$\mathbf{V}$).

\emph{Setup.}
At depth~$D$, the SE circuit prepares the unit-norm state
$|\boldsymbol{\psi}_D\rangle=\mathbf{U}_D^{\mathrm{SE}}(\bm\theta)|0^n\rangle$ with energy
$E_D(\bm\theta)=\langle\boldsymbol{\psi}_D|\mathbf{H}|\boldsymbol{\psi}_D\rangle$.  At depth~$D{+}2$
with new-layer parameters $\bm\phi$, the state becomes
$|\boldsymbol{\psi}_{D+2}\rangle=\mathbf{W}(\bm\phi)|\boldsymbol{\psi}_D\rangle$, where
$\mathbf{W}(\bm\phi)=\exp(-i\mathbf{G}(\bm\phi))$ is the unitary of the two new
rotation layers and $\mathbf{G}(\bm\phi)$ is its Hermitian generator
(constructed in \Cref{app:homotopy_extended}).  Because the energy
depends on $\bm\phi$ only through the conjugated Hamiltonian,
\begin{equation}\label{eq:energy_as_conj}
  E_{D+2}(\bm\theta,\bm\phi)
  =\langle\boldsymbol{\psi}_{D+2}|\mathbf{H}|\boldsymbol{\psi}_{D+2}\rangle
  =\langle\boldsymbol{\psi}_D| \mathbf{W}(\bm\phi)^\dagger\mathbf{H} \mathbf{W}(\bm\phi) |\boldsymbol{\psi}_D\rangle .
\end{equation}

\emph{Step~1 (Hadamard/BCH expansion).}
We use the \emph{Hadamard lemma} (the operator identity obtained by
Taylor-expanding $t\mapsto e^{i\mathbf{A}t}\mathbf{H}e^{-i\mathbf{A}t}$ about $t=0$ and
identifying the $k$-th derivative with the $k$-fold nested commutator).
For a Hermitian operator $\mathbf{A}$ and a real parameter $t$,
\begin{equation}\label{eq:hadamard_lemma}
  e^{i\mathbf{A}t} \mathbf{H} e^{-i\mathbf{A}t}
  =\sum_{k=0}^{\infty}\frac{(it)^k}{k!} \mathrm{ad}_{\mathbf{A}}^{k}(\mathbf{H})
  =\mathbf{H}+it[\mathbf{A},\mathbf{H}]+\frac{(it)^2}{2!}[\mathbf{A},[\mathbf{A},\mathbf{H}]]+\cdots,
\end{equation}
where $[\mathbf{A},\mathbf{H}]=\mathbf{A}\mathbf{H}-\mathbf{H}\mathbf{A}$ is the commutator and
$\mathrm{ad}_{\mathbf{A}}(\mathbf{H})\coloneqq[\mathbf{A},\mathbf{H}]$,
$\mathrm{ad}_{\mathbf{A}}^{k}=\mathrm{ad}_{\mathbf{A}}\circ\mathrm{ad}_{\mathbf{A}}^{k-1}$;
the series converges in operator norm for every bounded $\mathbf{A}$ because
$\|\mathrm{ad}_{\mathbf{A}}^{k}(\mathbf{H})\|\le(2\|\mathbf{A}\|)^k\|\mathbf{H}\|$.  Setting
$\mathbf{A}=\mathbf{G}(\bm\phi)$ and $t=1$ in~\eqref{eq:hadamard_lemma},
\begin{equation}\label{eq:WHW_expand}
  \mathbf{W}(\bm\phi)^\dagger\mathbf{H} \mathbf{W}(\bm\phi)
  =e^{i\mathbf{G}(\bm\phi)} \mathbf{H} e^{-i\mathbf{G}(\bm\phi)}
  =\mathbf{H}+i[\mathbf{G}(\bm\phi),\mathbf{H}]
   +\tfrac{i^2}{2}[\mathbf{G}(\bm\phi),[\mathbf{G}(\bm\phi),\mathbf{H}]]+\cdots,
\end{equation}
using $\mathbf{W}(\bm\phi)^\dagger=e^{+i\mathbf{G}(\bm\phi)}$ since $\mathbf{G}(\bm\phi)$ is
Hermitian.  Subtracting $\mathbf{H}$ and taking the expectation in the unit
state $|\boldsymbol{\psi}_D\rangle$ via~\eqref{eq:energy_as_conj} yields the
first-order form
\begin{equation}\label{eq:energy_diff_expand}
  E_{D+2}(\bm\theta,\bm\phi)-E_D(\bm\theta)
  =i\langle\boldsymbol{\psi}_D|[\mathbf{G}(\bm\phi),\mathbf{H}]|\boldsymbol{\psi}_D\rangle
   +O(\|\mathbf{G}(\bm\phi)\|^2),
\end{equation}
where the remainder collects all $k\ge 2$ terms
of~\eqref{eq:WHW_expand}, whose norms are bounded by
$\sum_{k\ge2}(2\|\mathbf{G}(\bm\phi)\|)^k\|\mathbf{H}\|/k!=O(\|\mathbf{G}(\bm\phi)\|^2)$ as
$\bm\phi\to\bm0$.  Since $\mathbf{G}(\bm\phi)$ and $\mathbf{H}$ are both Hermitian,
\begin{equation}\label{eq:comm_hermitian}
  \bigl(i[\mathbf{G},\mathbf{H}]\bigr)^\dagger
  =-i (\mathbf{G}\mathbf{H}-\mathbf{H}\mathbf{G})^\dagger
  =-i (\mathbf{H}\mathbf{G}-\mathbf{G}\mathbf{H})
  =i (\mathbf{G}\mathbf{H}-\mathbf{H}\mathbf{G})
  =i[\mathbf{G},\mathbf{H}],
\end{equation}
so $i[\mathbf{G}(\bm\phi),\mathbf{H}]$ is Hermitian and the leading term
in~\eqref{eq:energy_diff_expand} is a real expectation value,
consistent with $E_{D+2}-E_D\in\bbR$.

\emph{Step~2 (Exact bound via the mean-value inequality).}
Rather than truncate~\eqref{eq:WHW_expand}, we bound the full
(resummed) deviation.  Define the operator-valued curve
\begin{equation}\label{eq:f_curve}
  f(t)\coloneqq e^{it\mathbf{G}(\bm\phi)} \mathbf{H} e^{-it\mathbf{G}(\bm\phi)},
  \qquad t\in[0,1],
\end{equation}
so that $f(0)=\mathbf{H}$ and $f(1)=\mathbf{W}(\bm\phi)^\dagger\mathbf{H} \mathbf{W}(\bm\phi)$
by~\eqref{eq:WHW_expand}.  Differentiating~\eqref{eq:f_curve} by the
product rule, and using
$\tfrac{d}{dt}e^{\pm it\mathbf{G}}=\pm i\mathbf{G} e^{\pm it\mathbf{G}}=\pm i e^{\pm it\mathbf{G}}\mathbf{G}$
(the factor $\mathbf{G}$ commutes with its own exponential),
\begin{equation}\label{eq:f_deriv}
  f'(t)
  =i\mathbf{G} e^{it\mathbf{G}}\mathbf{H}e^{-it\mathbf{G}}
   -e^{it\mathbf{G}}\mathbf{H}e^{-it\mathbf{G}} i\mathbf{G}
  =i e^{it\mathbf{G}(\bm\phi)} [\mathbf{G}(\bm\phi),\mathbf{H}] e^{-it\mathbf{G}(\bm\phi)},
\end{equation}
where the last equality factors $e^{it\mathbf{G}}$ to the left and
$e^{-it\mathbf{G}}$ to the right of the commutator
$[\mathbf{G},\mathbf{H}]=\mathbf{G}\mathbf{H}-\mathbf{H}\mathbf{G}$.  Taking norms in~\eqref{eq:f_deriv},
since $e^{\pm it\mathbf{G}(\bm\phi)}$ is unitary and the spectral norm is
conjugation-invariant,
\begin{equation}\label{eq:fprime_bound}
  \|f'(t)\|
  =\|[\mathbf{G}(\bm\phi),\mathbf{H}]\|
  \overset{(\ast)}{\le}\|\mathbf{G}(\bm\phi)\mathbf{H}\|+\|\mathbf{H}\mathbf{G}(\bm\phi)\|
  \overset{(\ast\ast)}{\le}2 \|\mathbf{G}(\bm\phi)\| \|\mathbf{H}\|,
\end{equation}
where $(\ast)$ is the triangle inequality applied to
$\mathbf{G}\mathbf{H}-\mathbf{H}\mathbf{G}$ and $(\ast\ast)$ is submultiplicativity of the
spectral norm.  The bound~\eqref{eq:fprime_bound} is uniform in
$t\in[0,1]$.  By the mean-value inequality for the
Banach-space-valued $C^1$ map $f$ (i.e.\
$\|f(1)-f(0)\|\le\sup_{t\in[0,1]}\|f'(t)\| |1-0|$, obtained by
integrating $f(1)-f(0)=\int_0^1 f'(t) dt$ and applying the triangle
inequality to the integral),
\begin{equation}\label{eq:mvt_bound}
  \bigl\|\mathbf{W}(\bm\phi)^\dagger\mathbf{H} \mathbf{W}(\bm\phi)-\mathbf{H}\bigr\|
  =\|f(1)-f(0)\|
  \le 2 \|\mathbf{G}(\bm\phi)\| \|\mathbf{H}\|.
\end{equation}

\emph{Step~3 (From operator deviation to energy deviation).}
Subtract $E_D(\bm\theta)=\langle\boldsymbol{\psi}_D|\mathbf{H}|\boldsymbol{\psi}_D\rangle$
from~\eqref{eq:energy_as_conj} and apply Cauchy--Schwarz together with
the definition of the operator norm
($|\langle\boldsymbol{\psi}|\mathbf{M}|\boldsymbol{\psi}\rangle|\le\|\mathbf{M}\|$ for unit
$|\boldsymbol{\psi}\rangle$):
\begin{equation}\label{eq:energy_le_opnorm}
  |E_{D+2}(\bm\theta,\bm\phi)-E_D(\bm\theta)|
  =\bigl|\langle\boldsymbol{\psi}_D|
     \bigl(\mathbf{W}(\bm\phi)^\dagger\mathbf{H} \mathbf{W}(\bm\phi)-\mathbf{H}\bigr)
     |\boldsymbol{\psi}_D\rangle\bigr|
  \le\bigl\|\mathbf{W}(\bm\phi)^\dagger\mathbf{H} \mathbf{W}(\bm\phi)-\mathbf{H}\bigr\|.
\end{equation}

\emph{Step~4 (First-order expansion of the generator norm).}
By the generator decomposition derived in \Cref{app:homotopy_extended}
(eq.~\eqref{eq:W_tensor_app} and the lines following it),
$\mathbf{G}(\bm\phi)=\sum_{j=1}^{6n}\phi_j\mathbf{G}_j+\mathbf{R}(\bm\phi)$, where the
remainder $\mathbf{R}(\bm\phi)$ is the second- and higher-order part of the
Baker--Campbell--Hausdorff/Taylor expansion of
$\mathbf{G}(\bm\phi)=i\log\mathbf{W}(\bm\phi)$ about $\bm\phi=\bm0$, hence
$\|\mathbf{R}(\bm\phi)\|\le C_R\|\bm\phi\|^2$ for a finite constant $C_R<\infty$
determined by the generators $\{\mathbf{G}_j\}$ (the series for $i\log\mathbf{W}$
converges in operator norm because each $\|\mathbf{G}_j\|=\tfrac12$ is bounded).
By the triangle inequality and absolute homogeneity of the norm,
\begin{equation}\label{eq:Gnorm_expand}
  \|\mathbf{G}(\bm\phi)\|
  \le\sum_{j=1}^{6n}|\phi_j| \|\mathbf{G}_j\|+\|\mathbf{R}(\bm\phi)\|
  \le\sum_{j=1}^{6n}|\phi_j| \|\mathbf{G}_j\|+C_R\|\bm\phi\|^2.
\end{equation}

\emph{Conclusion.}
Chaining~\eqref{eq:energy_le_opnorm}, \eqref{eq:mvt_bound},
and~\eqref{eq:Gnorm_expand},
\begin{equation}\label{eq:perturb_app_final}
  |E_{D+2}(\bm\theta,\bm\phi)-E_D(\bm\theta)|
  \le 2 \|\mathbf{H}\| \|\mathbf{G}(\bm\phi)\|
  \le 2 \|\mathbf{H}\|\Bigl(\sum_{j=1}^{6n}|\phi_j| \|\mathbf{G}_j\|
     +C_R\|\bm\phi\|^2\Bigr),
\end{equation}
which is exactly~\eqref{eq:perturb} with the second-order remainder
$\le 2\|\mathbf{H}\|C_R\|\bm\phi\|^2$ retained.  The bound holds for every fixed
$\bm\theta$ (the only $\bm\theta$-dependence is through the unit state
$|\boldsymbol{\psi}_D(\bm\theta)\rangle$, which~\eqref{eq:energy_le_opnorm}
bounds out uniformly).  In the degenerate limit $\bm\phi\to\bm0$ one
has $\mathbf{G}(\bm\phi)\to\mathbf{0}$ by~\eqref{eq:Gnorm_expand}, so the
right-hand side of~\eqref{eq:perturb_app_final} tends to $0$ and
$E_{D+2}\to E_D$, consistent with the identity embedding
(\Cref{def:iep}).  For the jittered initialization with
$\sigma=0.05$ and $\|\mathbf{G}_j\|=1/2$, the linear term dominates the retained
quadratic correction $2\|\mathbf{H}\|C_R\|\bm\phi\|^2$ in magnitude; the quadratic
term is nonetheless carried through every downstream bound, never dropped.
\subsection{Gradient and Hessian perturbation bounds}\label{app:grad_hess}
We justify the bounds~\eqref{eq:grad_diff}--\eqref{eq:hess_diff}
used in the proof of \Cref{thm:basin}.  The quantifier is uniform over
the neighbourhood $\cN_D$ of \Cref{ass:local_reg}: every bound below
holds for \emph{all} $\bm\theta\in\cN_D$ with constants that do not
depend on $\bm\theta$ (they depend only on $\|\mathbf{H}\|$, on the
generator norms $\le 1/2$, and---for the Euclidean gradient norm and
the Hessian operator norm---on the parameter dimension $3nD$, as made
explicit below).
\begin{proof}
\emph{Setup.}
Write $\widetilde{E}(\bm\theta)\coloneqq E_{D+2}(\bm\theta,\bm\phi^\star)$
and
$\Delta\mathbf{H}\coloneqq\mathbf{W}(\bm\phi^\star)^\dagger\mathbf{H} \mathbf{W}(\bm\phi^\star)-\mathbf{H}$.
The operator $\Delta\mathbf{H}$ is Hermitian (a difference of Hermitian
operators) and is \emph{independent of $\bm\theta$}, since
$\bm\phi^\star$ is fixed.  By~\eqref{eq:energy_as_conj},
\begin{equation}\label{eq:gh_diff_state}
  \widetilde{E}(\bm\theta)-E_D(\bm\theta)
  =\langle\boldsymbol{\psi}_D(\bm\theta)| \Delta\mathbf{H} |\boldsymbol{\psi}_D(\bm\theta)\rangle ,
  \qquad
  |\boldsymbol{\psi}_D(\bm\theta)\rangle=\mathbf{U}_D^{\mathrm{SE}}(\bm\theta)|0^n\rangle .
\end{equation}
Applying the perturbation bound~\eqref{eq:mvt_bound} (equivalently
\eqref{eq:perturb_app_final}) with $\bm\phi=\bm\phi^\star$ and
$\|\bm\phi^\star\|=O(\sigma)$, $\|\mathbf{G}_j\|=1/2$,
\begin{equation}\label{eq:dH_bound}
  \|\Delta\mathbf{H}\|
  \le 2 \|\mathbf{H}\| \|\mathbf{G}(\bm\phi^\star)\|
  \le 2 \|\mathbf{H}\|\Bigl(\tfrac12\sum_{j=1}^{6n}|\phi_j^\star|
     +O(\|\bm\phi^\star\|^2)\Bigr)
  =O(\sigma),
\end{equation}
where the implied constant is proportional to $n\|\mathbf{H}\|$ (there are
$6n$ summands, each $O(\sigma)$).

\emph{Step~1 (Parameter-derivative bound).}
Each angle $\theta_j$ enters $\mathbf{U}_D^{\mathrm{SE}}(\bm\theta)$ through a
single Pauli rotation $e^{-i\theta_j\mathbf{G}_j^{\mathrm{gen}}}$ with a
local generator $\mathbf{G}_j^{\mathrm{gen}}$ (a half-Pauli, so
$\|\mathbf{G}_j^{\mathrm{gen}}\|=1/2$); write
$\mathbf{U}_D^{\mathrm{SE}}=\mathbf{U}_{>j} e^{-i\theta_j\mathbf{G}_j^{\mathrm{gen}}} \mathbf{U}_{<j}$
with $\mathbf{U}_{>j},\mathbf{U}_{<j}$ unitary and independent of $\theta_j$.
Differentiating,
\begin{equation}\label{eq:dpsi}
  \partial_{\theta_j}|\boldsymbol{\psi}_D(\bm\theta)\rangle
  =-i \mathbf{U}_{>j} \mathbf{G}_j^{\mathrm{gen}} 
    e^{-i\theta_j\mathbf{G}_j^{\mathrm{gen}}} \mathbf{U}_{<j}|0^n\rangle ,
\end{equation}
so, using $\||{-i}\mathbf{V}\boldsymbol{\xi}\rangle\|=\|\boldsymbol{\xi}\|$ for unitary
$\mathbf{V}$, submultiplicativity, and $\||0^n\rangle\|=1$,
\begin{equation}\label{eq:dpsi_bound}
  \bigl\|\partial_{\theta_j}|\boldsymbol{\psi}_D(\bm\theta)\rangle\bigr\|
  =\bigl\|\mathbf{G}_j^{\mathrm{gen}} 
     e^{-i\theta_j\mathbf{G}_j^{\mathrm{gen}}} \mathbf{U}_{<j}|0^n\rangle\bigr\|
  \le\|\mathbf{G}_j^{\mathrm{gen}}\|
  =\tfrac12 ,
\end{equation}
uniformly in $\bm\theta$.  (The standard Pauli generators are
involutory up to a factor, so $\|\mathbf{G}_j^{\mathrm{gen}}\|=1/2$; the
bound $\le 1/2$ is all we use.)

\emph{Step~2 (Gradient identity and per-component bound).}
Since $\widetilde{E}-E_D=\langle\boldsymbol{\psi}_D|\Delta\mathbf{H}|\boldsymbol{\psi}_D\rangle$
with $\Delta\mathbf{H}$ Hermitian and $\bm\theta$-independent, the product
rule gives the real, manifestly symmetric form
\begin{equation}\label{eq:grad_identity}
  \partial_{\theta_j}\bigl[\widetilde{E}-E_D\bigr]
  =\langle\partial_{\theta_j}\boldsymbol{\psi}_D|\Delta\mathbf{H}|\boldsymbol{\psi}_D\rangle
   +\langle\boldsymbol{\psi}_D|\Delta\mathbf{H}|\partial_{\theta_j}\boldsymbol{\psi}_D\rangle
  =2 \mathrm{Re} 
    \langle\partial_{\theta_j}\boldsymbol{\psi}_D| \Delta\mathbf{H} |\boldsymbol{\psi}_D\rangle ,
\end{equation}
where the second equality uses that the two terms are complex
conjugates ($\Delta\mathbf{H}^\dagger=\Delta\mathbf{H}$).  Bounding by
Cauchy--Schwarz, then~\eqref{eq:dpsi_bound} and~\eqref{eq:dH_bound}
(with $\||\boldsymbol{\psi}_D\rangle\|=1$),
\begin{equation}\label{eq:grad_comp_bound}
  \bigl|\partial_{\theta_j}[\widetilde{E}-E_D]\bigr|
  \le 2 \bigl\|\partial_{\theta_j}|\boldsymbol{\psi}_D\rangle\bigr\| 
        \|\Delta\mathbf{H}\| \||\boldsymbol{\psi}_D\rangle\|
  \le 2\cdot\tfrac12\cdot\|\Delta\mathbf{H}\|\cdot 1
  =\|\Delta\mathbf{H}\|=O(\sigma),
\end{equation}
uniformly in $j\in\{1,\dots,3nD\}$ and $\bm\theta\in\cN_D$.

\emph{Step~3 (Euclidean gradient bound).}
Aggregating the $3nD$ components~\eqref{eq:grad_comp_bound},
\begin{equation}\label{eq:grad_l2}
  \|\nabla\widetilde{E}(\bm\theta)-\nabla E_D(\bm\theta)\|_2
  =\Bigl(\sum_{j=1}^{3nD}
     \bigl|\partial_{\theta_j}[\widetilde{E}-E_D]\bigr|^2\Bigr)^{1/2}
  \le\sqrt{3nD} \|\Delta\mathbf{H}\|=O(\sqrt{nD} \sigma),
\end{equation}
which is~\eqref{eq:grad_diff}.  We flag the dimension factor honestly:
the \emph{per-coordinate} deviation is $O(\sigma)$ with a constant
$\propto n\|\mathbf{H}\|$, while the \emph{Euclidean} gradient norm carries
an additional $\sqrt{3nD}$ from summing $3nD$ coordinates; the $O(\cdot)$
constant in~\eqref{eq:grad_diff} is therefore not dimension-free but
scales like $n^{3/2}D^{1/2}\|\mathbf{H}\|$.  For a fixed problem size this
is a $\sigma$-independent constant times $\sigma$, which is all the
descent/PL argument requires; the dependence does \emph{not} affect the
contraction factor $1-\eta\mu_{D+2}$ beyond shifting
$\mu_{D+2}=\mu_D-O(\sigma)$ (with the $n,D$ scaling of that constant
recorded at~\eqref{eq:basin_c_scaling}).

\emph{Step~4 (Hessian bound).}
Differentiating~\eqref{eq:grad_identity} once more in $\theta_k$ and
using that $\Delta\mathbf{H}$ is Hermitian and $\bm\theta$-independent,
\begin{equation}\label{eq:hess_identity}
  \partial^2_{\theta_k\theta_j}\bigl[\widetilde{E}-E_D\bigr]
  =2 \mathrm{Re}\Bigl(
     \langle\partial^2_{\theta_k\theta_j}\boldsymbol{\psi}_D|\Delta\mathbf{H}|\boldsymbol{\psi}_D\rangle
     +\langle\partial_{\theta_j}\boldsymbol{\psi}_D|
        \Delta\mathbf{H}|\partial_{\theta_k}\boldsymbol{\psi}_D\rangle\Bigr).
\end{equation}
By the same single-rotation argument as~\eqref{eq:dpsi_bound}, the
second mixed derivative satisfies
$\|\partial^2_{\theta_k\theta_j}|\boldsymbol{\psi}_D\rangle\|
\le\|\mathbf{G}_k^{\mathrm{gen}}\| \|\mathbf{G}_j^{\mathrm{gen}}\|\le 1/4$
(two factors of a half-Pauli; for $k=j$ the same bound holds with
$\|\mathbf{G}_j^{\mathrm{gen}}\|^2=1/4$).  Hence, by Cauchy--Schwarz on each
inner product in~\eqref{eq:hess_identity} together
with~\eqref{eq:dpsi_bound} and~\eqref{eq:dH_bound},
\begin{equation}\label{eq:hess_comp_bound}
  \bigl|\partial^2_{\theta_k\theta_j}[\widetilde{E}-E_D]\bigr|
  \le 2\Bigl(\tfrac14+\tfrac12\cdot\tfrac12\Bigr)\|\Delta\mathbf{H}\|
  =\|\Delta\mathbf{H}\|=O(\sigma),
\end{equation}
uniformly in $k,j$ and $\bm\theta\in\cN_D$.  The Hessian difference is
a symmetric $3nD\times 3nD$ matrix whose entries are all $O(\sigma)$,
so its spectral norm obeys
\begin{equation}\label{eq:hess_op_bound}
  \bigl\|\nabla^2\widetilde{E}(\bm\theta)-\nabla^2 E_D(\bm\theta)\bigr\|
  \le\bigl\|\nabla^2\widetilde{E}-\nabla^2 E_D\bigr\|_F
  =\Bigl(\sum_{k,j}\bigl|\partial^2_{\theta_k\theta_j}[\widetilde{E}-E_D]\bigr|^2
   \Bigr)^{1/2}
  \le 3nD \|\Delta\mathbf{H}\|=O(nD \sigma),
\end{equation}
which is~\eqref{eq:hess_diff}.  Here we used $\|\cdot\|\le\|\cdot\|_F$
for the spectral-vs-Frobenius comparison.  As in Step~3, the bound is
$O(\sigma)$ for fixed problem size, but its constant scales with the
dimension (here as $nD$, with an extra $n\|\mathbf{H}\|$ from
$\|\Delta\mathbf{H}\|$); we record this rather than absorb it silently.

\emph{Sanity check ($\sigma\to0$).}
As $\sigma\to0$, $\|\bm\phi^\star\|=O(\sigma)\to0$, so $\Delta\mathbf{H}\to\mathbf{0}$
by~\eqref{eq:dH_bound}, and~\eqref{eq:grad_l2}--\eqref{eq:hess_op_bound}
force $\nabla\widetilde{E}\to\nabla E_D$ and
$\nabla^2\widetilde{E}\to\nabla^2 E_D$; in particular
$\mu_{D+2}\ge\mu_D-c\sigma\to\mu_D$ and $L_{D+2}\le L_D+c\sigma\to L_D$,
recovering the depth-$D$ regularity constants.
\end{proof}
\subsection{Monotone acceptance under sampling noise}\label{app:sampling}
On real hardware the energies are not evaluated exactly but estimated
from a finite number of measurement shots, so the exact acceptance rule
of \Cref{thm:monotone} must be replaced by one that is robust to the
estimation error.  We model the estimator through a two-sided confidence
interval of half-width $\varepsilon_{\mathrm{CI}}$ and show that
widening the acceptance margin to $2\varepsilon_{\mathrm{CI}}$ preserves
true (noise-free) monotonicity with high probability.
\begin{proposition}\label{prop:finite_shots}
  With $|\hat{E}(\bm\theta)-E(\bm\theta)|\le\varepsilon_{\mathrm{CI}}$
  w.h.p., accepting only if
  $\hat{E}(\bm\theta')\le\hat{E}(\bm\theta)-2\varepsilon_{\mathrm{CI}}$
  guarantees $E(\bm\theta')\le E(\bm\theta)$.
\end{proposition}
\begin{proof}
Let $E(\cdot)$ denote the exact energy (the quantity that
\Cref{thm:monotone} controls) and $\hat{E}(\cdot)$ its shot-based
estimate.  Fix the two parameter vectors $\bm\theta$ (the incumbent) and
$\bm\theta'$ (the candidate) being compared at a given expansion.  The
hypothesis is a two-sided confidence interval at each point: there is a
failure probability $\delta_{\mathrm{CI}}\in(0,1)$ (for example,
$\delta_{\mathrm{CI}}=2\exp \bigl(-2N\varepsilon_{\mathrm{CI}}^2/(\lambda_d-\lambda_1)^2\bigr)$
for $N$ \emph{independent} shots, by Hoeffding's inequality applied to the
i.i.d.\ per-shot energy outcomes, which lie in the spectral range
$[\lambda_1,\lambda_d]$ of $\mathbf{H}$, so the per-shot range is
$\lambda_d-\lambda_1\le 2\|\mathbf{H}\|$) such that, individually,
\begin{equation}\label{eq:ci_individual}
  \bbP\bigl(|\hat{E}(\bm\theta)-E(\bm\theta)|
            \le\varepsilon_{\mathrm{CI}}\bigr)\ge 1-\delta_{\mathrm{CI}}
  \quad\text{and}\quad
  \bbP\bigl(|\hat{E}(\bm\theta')-E(\bm\theta')|
            \le\varepsilon_{\mathrm{CI}}\bigr)\ge 1-\delta_{\mathrm{CI}} .
\end{equation}
By the union bound, the two events in~\eqref{eq:ci_individual} hold
\emph{simultaneously} with probability at least
$1-2\delta_{\mathrm{CI}}$.  We argue deterministically on this
good event, on which both estimates are within
$\varepsilon_{\mathrm{CI}}$ of their true values.  Suppose the
candidate is accepted, i.e.
\begin{equation}\label{eq:ci_accept}
  \hat{E}(\bm\theta')\le\hat{E}(\bm\theta)-2\varepsilon_{\mathrm{CI}} .
\end{equation}
Then the true candidate energy is controlled by the following chain of
three inequalities, each justified to its right:
\begin{align}
  E(\bm\theta')
  &\le\hat{E}(\bm\theta')+\varepsilon_{\mathrm{CI}}
  &&\text{(upper confidence bound at $\bm\theta'$, good event)}
  \label{eq:ci_chain1}\\
  &\le\bigl(\hat{E}(\bm\theta)-2\varepsilon_{\mathrm{CI}}\bigr)
      +\varepsilon_{\mathrm{CI}}
   =\hat{E}(\bm\theta)-\varepsilon_{\mathrm{CI}}
  &&\text{(acceptance rule~\eqref{eq:ci_accept})}
  \label{eq:ci_chain2}\\
  &\le\bigl(E(\bm\theta)+\varepsilon_{\mathrm{CI}}\bigr)
      -\varepsilon_{\mathrm{CI}}
   =E(\bm\theta)
  &&\text{(upper confidence bound at $\bm\theta$, good event).}
  \label{eq:ci_chain3}
\end{align}
The first inequality uses
$E(\bm\theta')\le\hat{E}(\bm\theta')+\varepsilon_{\mathrm{CI}}$ (the
left half of the confidence interval at $\bm\theta'$); the second
substitutes the acceptance condition~\eqref{eq:ci_accept}; the third
uses $\hat{E}(\bm\theta)\le E(\bm\theta)+\varepsilon_{\mathrm{CI}}$ (the
right half of the confidence interval at $\bm\theta$).  Chaining
\eqref{eq:ci_chain1}--\eqref{eq:ci_chain3} gives
$E(\bm\theta')\le E(\bm\theta)$ on the good event, which holds with
probability at least $1-2\delta_{\mathrm{CI}}$.
\end{proof}
\begin{remark}[High-probability qualifier and what the noise affects]
\label{rem:sampling_qualifier}
  The guarantee is high-probability, not deterministic: it can fail only
  if one of the two confidence intervals in~\eqref{eq:ci_individual} is
  violated, an event of probability $\le2\delta_{\mathrm{CI}}$ per
  accept/reject comparison.  Over a run with $K$ such comparisons (one
  per expansion stage that invokes the rule), a second union bound caps
  the total failure probability at $2K\delta_{\mathrm{CI}}$, so choosing
  $\delta_{\mathrm{CI}}\le\delta_{\mathrm{tot}}/(2K)$ (equivalently
  $N=\Theta(\varepsilon_{\mathrm{CI}}^{-2}\|\mathbf{H}\|^2
  \log(K/\delta_{\mathrm{tot}}))$ shots per estimate) certifies the
  whole produced sequence is monotone with probability
  $\ge1-\delta_{\mathrm{tot}}$.  As $\varepsilon_{\mathrm{CI}}\to0$ the
  margin $2\varepsilon_{\mathrm{CI}}\to0$ and the rule reduces to the
  exact acceptance rule of \Cref{thm:monotone}.  We emphasise that this
  proposition concerns \emph{energy} estimation noise only: it makes the
  \emph{acceptance decision} robust and thereby preserves monotonicity
  of the produced energies.  It says nothing about the \emph{gradient}
  noise that enters the optimizer producing the candidate
  $\tilde{\bm\theta}^{(D+2)}$; gradient shot noise affects \emph{which}
  candidate is found and hence the convergence rate of
  \Cref{thm:basin}, but---because the acceptance test is applied to the
  candidate after it is produced---it cannot break the monotonicity
  guarantee, which rests solely on the energy comparison above and on
  the lossless fallback~\eqref{eq:energy_preserve_exact}.
\end{remark}
\section{Extended Theoretical Discussion}\label{app:theory_extended}
This appendix provides detailed derivations, motivations, and
interpretations for the theoretical results stated in
\Cref{sec:theory}.  Each subsection corresponds to a specific result
in the main text and is cross-referenced from there.
\subsection{Identity-embedding property: extended discussion}
\label{app:iep_extended}
The IEP (\Cref{def:iep}) is the foundational mechanism that makes
progressive depth training safe.  The core idea is simple: if we can
add new rotation layers to the circuit in such a way that those layers
``do nothing'' (implement the identity), then the deeper circuit can
reproduce \emph{every} state that the shallower circuit can, simply by
keeping the extra parameters at their identity-pair configuration.
The identity pairing of \Cref{def:id_pair} provides a concrete
$\iota_D$: set
$\iota_D(\bm\theta)=(\bm\theta,\bm\theta_{\mathrm{ref}},
\bm\theta_{\mathrm{ref}})$, where the two appended copies of
$\bm\theta_{\mathrm{ref}}$ form an identity pair in the
single-entangler circuit (the inverse rotation layer followed by the
forward rotation layer with the same parameters cancels to $\mathbf{I}_{2^n}$
by \Cref{prop:block_inv}).
In the language of reachable sets, if $\cM_D$ denotes all states
producible by the depth-$D$ SE circuit (over all parameter choices),
then the IEP guarantees $\cM_D\subseteq\cM_{D+2}$: every state
reachable at depth~$D$ is also reachable at depth~$D{+}2$ (by
freezing the new layers at identity).  This monotonicity
is the key ingredient for both \Cref{cor:energy_monotone} (optimal
energy cannot increase with depth) and \Cref{thm:monotone} (practical
monotonicity via the acceptance rule).
\subsection{Homotopy construction: detailed derivation}
\label{app:homotopy_extended}
We provide the full derivation of the homotopy
formulation~\eqref{eq:homotopy} used in \Cref{sec:continuation}.

\emph{Step~1 ($\mathbf{W}(\bm\phi)$ is a tensor product of single-qubit
unitaries).}
The two appended layers form an identity pair: by
construction~\eqref{eq:UD}, the depth-$(D{+}2)$ SE circuit places the
two new \emph{rotation} layers $\mathbf{W}_D,\mathbf{W}_{D+1}$ in the leftmost
(last-applied) factors, with the shared entangler $\mathbf{U}_{\mathrm{ent}}$
already accounted for at lower depth; the CNOT ring contributed by the
expansion cancels through \Cref{prop:block_inv}, exactly as
in~\eqref{eq:nested_cancel}--\eqref{eq:nested_collapse}.  What survives
acting on $|\boldsymbol{\psi}_D\rangle$ is therefore the product of the two
remaining single-qubit rotation layers,
$\mathbf{W}(\bm\phi)=\mathbf{W}_{D+1}\mathbf{W}_D$, \emph{with no entangling gate between
them}.  Each rotation layer is by definition a tensor product over
qubits, $\mathbf{W}_D=\bigotimes_q r_q^{(D)}$ and
$\mathbf{W}_{D+1}=\bigotimes_q r_q^{(D+1)}$ with $r_q^{(\cdot)}\in\SU(2)$.
Since $(\mathbf{A}_1\otimes\cdots\otimes\mathbf{A}_n)(\mathbf{B}_1\otimes\cdots\otimes\mathbf{B}_n)
=(\mathbf{A}_1\mathbf{B}_1)\otimes\cdots\otimes(\mathbf{A}_n\mathbf{B}_n)$ (the mixed-product
property of the Kronecker product), the layer product is again a tensor
product of single-qubit unitaries,
\begin{equation}\label{eq:W_tensor_app}
  \mathbf{W}(\bm\phi)=\bigotimes_{q=1}^n w_q(\bm\phi_q),
  \qquad w_q(\bm\phi_q)=r_q^{(D+1)} r_q^{(D)}\in\SU(2),
\end{equation}
where $\bm\phi_q\in\bbR^6$ collects the six parameters of qubit~$q$
across the two new layers (three Euler angles per layer), and
$w_q(\bm 0)=\mathbf{I}_2$ for each qubit because the reference configuration
makes the pair the identity (\Cref{def:id_pair}).

\emph{Step~2 (Matrix-exponential representation; the generator is a sum
of local terms).}
The exponential map $\exp:\su(2)\to\SU(2)$ is a local diffeomorphism at
$\bm0$ (its differential at the origin is the identity, so the inverse
function theorem applies); hence there is a neighbourhood of $\mathbf{I}_2$ on
which the matrix logarithm $\log:\SU(2)\to\su(2)$ is a well-defined
smooth inverse.  On the explicit domain
$\{\bm\phi:\|\bm\phi_q\|_1<\pi\ \text{for every }q\}$---where each
single-qubit factor $w_q(\bm\phi_q)$ has $\SU(2)$ rotation angle below
$\pi$ (the angle of a product of rotations is at most the sum of the
individual angles), hence no eigenvalue at $-1$ and so lies in the domain
of the principal logarithm (the jitter regime $\sigma=0.05$ satisfies this
with room to spare)---define the Hermitian single-qubit generators
$g_q(\bm\phi_q)\coloneqq i\log w_q(\bm\phi_q)$,
so that $w_q(\bm\phi_q)=\exp(-i g_q(\bm\phi_q))$ with
$g_q(\bm 0)=\mathbf{0}$.  Using
$\exp(\mathbf{A}_1)\otimes\exp(\mathbf{A}_2)=\exp(\mathbf{A}_1\otimes\mathbf{I}+\mathbf{I}\otimes\mathbf{A}_2)$
(valid because $\mathbf{A}_1\otimes\mathbf{I}$ and $\mathbf{I}\otimes\mathbf{A}_2$ commute, so the
two single-qubit exponentials combine without a
Baker--Campbell--Hausdorff correction \emph{across} qubits), the global
unitary~\eqref{eq:W_tensor_app} is
\begin{equation}\label{eq:W_global_exp}
  \mathbf{W}(\bm\phi)=\exp\bigl(-i \mathbf{G}(\bm\phi)\bigr),
  \qquad
  \mathbf{G}(\bm\phi)=\sum_{q=1}^n
    \mathbf{I}^{\otimes(q-1)}\otimes g_q(\bm\phi_q)\otimes\mathbf{I}^{\otimes(n-q)} ,
\end{equation}
a sum of $n$ mutually commuting, single-qubit-supported Hermitian
terms.  This $\mathbf{G}(\bm\phi)$ is the Hermitian generator referenced
in~\eqref{eq:homotopy} and~\eqref{eq:WHW_expand}.

\emph{Step~3 (First-order decomposition into local Paulis with
$\|\mathbf{G}_j\|=1/2$).}
Each $g_q(\bm\phi_q)$ is smooth in $\bm\phi_q$ with $g_q(\bm0)=\mathbf{0}$, so
its first-order Taylor expansion is
$g_q(\bm\phi_q)=\sum_{\ell=1}^{6}\phi_{q,\ell} 
\partial_{\phi_{q,\ell}}g_q(\bm0)+O(\|\bm\phi_q\|^2)$.  Collecting the
$6n$ parameters into the single index
$j\leftrightarrow(q,\ell)$ and substituting
into~\eqref{eq:W_global_exp},
\begin{equation}\label{eq:G_first_order_app}
  \mathbf{G}(\bm\phi)=\sum_{j=1}^{6n}\phi_j \mathbf{G}_j+\mathbf{R}(\bm\phi),
  \qquad
  \mathbf{G}_j\coloneqq i \frac{\partial\mathbf{W}}{\partial\phi_j}\Big|_{\bm\phi=\bm0}
        =\frac{\partial\mathbf{G}}{\partial\phi_j}\Big|_{\bm\phi=\bm0},
  \qquad \|\mathbf{R}(\bm\phi)\|=O(\|\bm\phi\|^2),
\end{equation}
where the identity
$\mathbf{G}_j=i \partial_{\phi_j}\mathbf{W}|_{\bm0}=\partial_{\phi_j}\mathbf{G}|_{\bm0}$
follows by differentiating $\mathbf{W}=\exp(-i\mathbf{G})$ at $\bm\phi=\bm0$ (where
$\mathbf{G}(\bm0)=\mathbf{0}$, so the chain rule for the matrix exponential gives
$\partial_{\phi_j}\mathbf{W}|_{\bm0}=-i \partial_{\phi_j}\mathbf{G}|_{\bm0}$), and the
$O(\|\bm\phi\|^2)$ remainder $\mathbf{R}(\bm\phi)$ is precisely the
second- and higher-order Taylor/BCH tail used
in~\eqref{eq:Gnorm_expand} and~\eqref{eq:G_first_order_app}.  Each
$\partial_{\phi_{q,\ell}}g_q(\bm0)$ is the generator of one elementary
Pauli rotation: for the SE ansatz the single-qubit factors are
$\mathbf{R}_P(\phi)=\exp(-i\tfrac{\phi}{2}\mathbf{P})$ with
$\mathbf{P}\in\{\mathbf{X},\mathbf{Y},\mathbf{Z}\}$, so
$\partial_\phi\bigl(\tfrac{\phi}{2}\mathbf{P}\bigr)=\tfrac12\mathbf{P}$.  Hence, for
each $j$, $\mathbf{G}_j=\tfrac12\mathbf{P}_q$ is a single-qubit Pauli on the
corresponding qubit~$q$ tensored with identities; for example, if
$\phi_j$ is the angle of an $\mathbf{R}_Z$ gate on qubit~$q$ then
$\mathbf{G}_j=\tfrac12\mathbf{Z}_q$.  Because every Pauli matrix has eigenvalues
$\pm1$, its spectral norm is $1$, and therefore
\begin{equation}\label{eq:Gj_norm_half}
  \|\mathbf{G}_j\|=\bigl\|\tfrac12\mathbf{P}_q\bigr\|=\tfrac12 ,
  \qquad j=1,\dots,6n,
\end{equation}
using $\|\mathbf{A}\otimes\mathbf{B}\|=\|\mathbf{A}\| \|\mathbf{B}\|$ (the spectral norm is
multiplicative under tensor products) so that tensoring with identities
does not change the norm.  This is the $\|\mathbf{G}_j\|=1/2$ fact invoked
throughout \Cref{app:perturb,app:grad_hess}.

\emph{Step~4 (Homotopy interpretation).}
The homotopy parameter $\lambda\in[0,1]$ in~\eqref{eq:homotopy}
smoothly ``turns on'' these new layers via
$\exp(-i\lambda\mathbf{G}(\bm\phi))$.  This is the standard device of
numerical continuation theory~\cite{allgower2003introduction}: instead
of a discrete jump from depth~$D$ to depth~$D{+}2$, we study the
one-parameter family $F(\bm\theta,\bm\phi;\lambda)$
of~\eqref{eq:F_homotopy} that interpolates between
$F(\cdot;0)=E_D$ and $F(\cdot;1)=E_{D+2}$, to which the implicit
function theorem applies (\Cref{app:continuation_complete}) to track
how critical points of the energy landscape evolve as $\lambda$
increases.  In the degenerate case $\bm\phi=\bm0$, \eqref{eq:W_global_exp}
gives $\mathbf{G}(\bm0)=\mathbf{0}$, hence $\exp(-i\lambda\mathbf{G}(\bm0))=\mathbf{I}$ for all
$\lambda$ and $\mathbf{W}(\bm0)=\mathbf{I}$: the homotopy is the constant identity
deformation, consistent with \Cref{rem:degenerate_phi_zero_app}.
\subsection{Continuation theorem: complete proof}
\label{app:continuation_complete}
We provide the full five-step proof of \Cref{thm:continuation}, along
with remarks on genericity and the degenerate case.
\begin{proof}[Complete proof of \Cref{thm:continuation}]
We apply the Implicit Function Theorem (IFT) to the stationarity map
\begin{equation}\label{eq:Phi_def_app}
  \Phi:\bbR^{3nD}\times\bbR\to\bbR^{3nD},
  \qquad
  \Phi(\bm\theta,\lambda)
  \coloneqq\nabla_{\bm\theta}F(\bm\theta,\bm\phi^\star;\lambda),
\end{equation}
where $\bm\phi^\star$ is the fixed nonzero target from condition~(iv).
By condition~(ii), $F$ is $C^2$ near
$(\bm\theta_D^\star,\bm\phi^\star,0)$, so $\Phi=\nabla_{\bm\theta}F$ is
$C^1$ there; both partial Jacobians $\partial_{\bm\theta}\Phi$ and
$\partial_\lambda\Phi$ exist and are continuous, as the IFT requires.

\emph{Step~1: Verify the IFT hypotheses at $(\bm\theta_D^\star,0)$.}
\begin{enumerate}[label=(\alph*)]
  \item \textbf{Root condition} $\Phi(\bm\theta_D^\star,0)=\bm 0$:
    By condition~(iii), $F(\bm\theta,\bm\phi^\star;0)=E_D(\bm\theta)$
    identically in $\bm\theta$, so
    $\Phi(\bm\theta,0)=\nabla_{\bm\theta}F(\bm\theta,\bm\phi^\star;0)
    =\nabla E_D(\bm\theta)$; evaluating at the local minimizer
    $\bm\theta_D^\star$ and invoking first-order optimality
    (\Cref{ass:nondegen} gives a local minimizer, at which the
    gradient vanishes),
    $\Phi(\bm\theta_D^\star,0)=\nabla E_D(\bm\theta_D^\star)=\bm 0$.
  \item \textbf{Invertibility of the $\bm\theta$-Jacobian:}
    Differentiating $\Phi(\bm\theta,0)=\nabla E_D(\bm\theta)$ in
    $\bm\theta$,
    \begin{equation}\label{eq:Phi_theta_jac}
      \partial_{\bm\theta}\Phi(\bm\theta_D^\star,0)
      =\nabla^2_{\bm\theta\bm\theta}E_D(\bm\theta_D^\star)
      \succeq\mu\mathbf{I}\succ\mathbf 0
    \end{equation}
    by condition~(i)/\Cref{ass:nondegen}.  A symmetric matrix with
    $\nabla^2 E_D\succeq\mu\mathbf{I}$ has every eigenvalue $\ge\mu>0$, so
    it is positive definite, hence invertible, and its inverse has
    spectral norm equal to the reciprocal of its smallest eigenvalue:
    \begin{equation}\label{eq:hess_inv_bound}
      \bigl\|[\nabla^2_{\bm\theta\bm\theta}E_D(\bm\theta_D^\star)]^{-1}\bigr\|
      =\frac{1}{\lambda_{\min}(\nabla^2 E_D(\bm\theta_D^\star))}
      \le\frac{1}{\mu}.
    \end{equation}
\end{enumerate}

\emph{Step~2: Confirm the homotopy is non-trivial.}
With $\mathbf{U}_{D+2}^{(\lambda)}=e^{-i\lambda\mathbf{G}(\bm\phi^\star)}
\mathbf{U}_D^{\mathrm{SE}}(\bm\theta)$ and
$|\boldsymbol{\psi}_D(\bm\theta)\rangle=\mathbf{U}_D^{\mathrm{SE}}(\bm\theta)|0^n\rangle$,
the homotopy energy~\eqref{eq:F_homotopy} reads
$F(\bm\theta,\bm\phi^\star;\lambda)
=\langle\boldsymbol{\psi}_D(\bm\theta)| 
e^{+i\lambda\mathbf{G}(\bm\phi^\star)}\mathbf{H} e^{-i\lambda\mathbf{G}(\bm\phi^\star)} 
|\boldsymbol{\psi}_D(\bm\theta)\rangle$.
Differentiating in $\lambda$ and evaluating at $\lambda=0$, using
$\partial_\lambda e^{\mp i\lambda\mathbf{G}}|_{\lambda=0}=\mp i\mathbf{G}$,
\begin{equation}\label{eq:lambda_deriv_app}
  \partial_\lambda F(\bm\theta,\bm\phi^\star;\lambda)\big|_{\lambda=0}
  =\langle\boldsymbol{\psi}_D(\bm\theta)|
    \bigl(i\mathbf{G}(\bm\phi^\star)\mathbf{H}-\mathbf{H} i\mathbf{G}(\bm\phi^\star)\bigr)
    |\boldsymbol{\psi}_D(\bm\theta)\rangle
  =i\langle\boldsymbol{\psi}_D(\bm\theta)|[\mathbf{G}(\bm\phi^\star),\mathbf{H}]|\boldsymbol{\psi}_D(\bm\theta)\rangle,
\end{equation}
which is real because $i[\mathbf{G},\mathbf{H}]$ is Hermitian (Step~1
of~\Cref{app:perturb}).  The IFT does not require this $\lambda$-derivative
to be nonzero, but the \emph{content} of the theorem---that the minimum
genuinely \emph{moves}---requires the mixed second derivative
$\partial_\lambda\Phi(\bm\theta_D^\star,0)
=\partial_\lambda\nabla_{\bm\theta}F(\bm\theta_D^\star,\bm\phi^\star;\lambda)|_{\lambda=0}
=\nabla_{\bm\theta}\bigl(i\langle\boldsymbol{\psi}_D|[\mathbf{G}(\bm\phi^\star),\mathbf{H}]
|\boldsymbol{\psi}_D\rangle\bigr)\big|_{\bm\theta_D^\star}$
to be nonzero, which is exactly condition~(v).  (Exchanging the
$\bm\theta$- and $\lambda$-derivatives is justified by the
$C^2$ smoothness of condition~(ii) via Clairaut/Schwarz.)

\emph{Step~3: Apply the IFT.}
By Step~1, $\Phi$ is $C^1$ near $(\bm\theta_D^\star,0)$ with
$\Phi(\bm\theta_D^\star,0)=\bm0$ and $\partial_{\bm\theta}\Phi$
invertible there.  The Implicit Function Theorem
(see \cite[Theorem~1.3.1]{allgower2003introduction}) therefore provides
open neighbourhoods $\cU\ni\bm\theta_D^\star$ in $\bbR^{3nD}$ and
$\cV\ni 0$ in $\bbR$, and a unique $C^1$ map $\bm\theta:\cV\to\cU$ with
$\bm\theta(0)=\bm\theta_D^\star$ and
\begin{equation}\label{eq:ift_root_app}
  \Phi(\bm\theta(\lambda),\lambda)
  =\nabla_{\bm\theta}F(\bm\theta(\lambda),\bm\phi^\star;\lambda)=\bm0
  \qquad\text{for all }\lambda\in\cV.
\end{equation}
Choosing $\lambda_{\max}\in(0,1]$ with $[0,\lambda_{\max}]\subseteq\cV$
(intersecting $\cV$ with $[0,1]$, possible since $\cV$ is an open
neighbourhood of $0$) yields the curve of the statement.  Here $\cU$ is
a neighbourhood (a \emph{set}), distinct from the circuit unitary
$\mathbf{U}_D^{\mathrm{SE}}$.

\emph{Step~4: Derive the initial velocity~\eqref{eq:curve_deriv}.}
The identity~\eqref{eq:ift_root_app} holds for all $\lambda\in\cV$, so
we may differentiate it in $\lambda$.  By the chain rule applied to
$\lambda\mapsto\Phi(\bm\theta(\lambda),\lambda)$,
\begin{equation}\label{eq:chain_rule_app}
  \partial_{\bm\theta}\Phi(\bm\theta(\lambda),\lambda) \bm\theta'(\lambda)
  +\partial_\lambda\Phi(\bm\theta(\lambda),\lambda)=\bm0 .
\end{equation}
Evaluating~\eqref{eq:chain_rule_app} at $\lambda=0$ (where
$\bm\theta(0)=\bm\theta_D^\star$) and
substituting~\eqref{eq:Phi_theta_jac},
\begin{equation}\label{eq:velocity_eqn_app}
  \nabla^2_{\bm\theta\bm\theta}E_D(\bm\theta_D^\star) \bm\theta'(0)
  +\partial_\lambda\nabla_{\bm\theta}
   F(\bm\theta_D^\star,\bm\phi^\star;\lambda)\big|_{\lambda=0}=\bm0 .
\end{equation}
Since $\nabla^2 E_D(\bm\theta_D^\star)$ is invertible
by~\eqref{eq:hess_inv_bound}, we may left-multiply
by its inverse to obtain exactly~\eqref{eq:curve_deriv}:
\begin{equation}\label{eq:velocity_solved_app}
  \bm\theta'(0)
  =-\bigl[\nabla^2_{\bm\theta\bm\theta}E_D(\bm\theta_D^\star)\bigr]^{-1}
     \partial_\lambda\nabla_{\bm\theta}
     F(\bm\theta_D^\star,\bm\phi^\star;\lambda)\big|_{\lambda=0} .
\end{equation}
By condition~(v) the second factor is nonzero, and an invertible matrix
maps a nonzero vector to a nonzero vector, so $\bm\theta'(0)\neq\bm0$:
the minimum genuinely moves as $\lambda$ increases.  Taking norms
in~\eqref{eq:velocity_solved_app} with~\eqref{eq:hess_inv_bound} gives
the quantitative bound
$\|\bm\theta'(0)\|\le\mu^{-1} 
\|\partial_\lambda\nabla_{\bm\theta}F|_{\lambda=0}\|
=O(\|\bm\phi^\star\|/\mu)=O(\sigma/\mu)$ quoted in the main text, since
$\partial_\lambda\nabla_{\bm\theta}F|_{\lambda=0}$ is linear in
$\mathbf{G}(\bm\phi^\star)$ by~\eqref{eq:lambda_deriv_app} and
$\|\mathbf{G}(\bm\phi^\star)\|=O(\|\bm\phi^\star\|)$
by~\eqref{eq:Gnorm_expand}.

\emph{Step~5: The continued critical point remains a local minimum.}
The map $\lambda\mapsto\nabla^2_{\bm\theta\bm\theta}
F(\bm\theta(\lambda),\bm\phi^\star;\lambda)$ is continuous on $\cV$:
$F$ is $C^2$ (condition~(ii)) so $\nabla^2_{\bm\theta\bm\theta}F$ is
continuous in $(\bm\theta,\lambda)$, and $\bm\theta(\cdot)$ is $C^1$ by
Step~3, so the composition is continuous.  At $\lambda=0$ it equals
$\nabla^2_{\bm\theta\bm\theta}E_D(\bm\theta_D^\star)\succeq\mu\mathbf{I}$
by~\eqref{eq:Phi_theta_jac}.  By continuity of the matrix-valued map in
the spectral norm, there is $\lambda_1\in(0,\lambda_{\max}]$ such that
\begin{equation}\label{eq:hess_close_app}
  \bigl\|\nabla^2_{\bm\theta\bm\theta}
    F(\bm\theta(\lambda),\bm\phi^\star;\lambda)
    -\nabla^2_{\bm\theta\bm\theta}E_D(\bm\theta_D^\star)\bigr\|
  \le\frac{\mu}{2}
  \qquad\text{for all }\lambda\in[0,\lambda_1].
\end{equation}
Writing $\mathbf{A}=\nabla^2 E_D(\bm\theta_D^\star)$ and $\mathbf{E}(\lambda)=
\nabla^2_{\bm\theta\bm\theta}F(\bm\theta(\lambda),\bm\phi^\star;\lambda)-\mathbf{A}$,
both symmetric, Weyl's inequality
$\lambda_{\min}(\mathbf{A}+\mathbf{E})\ge\lambda_{\min}(\mathbf{A})-\|\mathbf{E}\|$ gives, for
$\lambda\in[0,\lambda_1]$,
\begin{equation}\label{eq:weyl_app}
  \lambda_{\min}\bigl(\nabla^2_{\bm\theta\bm\theta}
    F(\bm\theta(\lambda),\bm\phi^\star;\lambda)\bigr)
  \ge\lambda_{\min}(\mathbf{A})-\|\mathbf{E}(\lambda)\|
  \ge\mu-\frac{\mu}{2}=\frac{\mu}{2}>0 ,
\end{equation}
i.e.\ $\nabla^2_{\bm\theta\bm\theta}F(\bm\theta(\lambda),\bm\phi^\star;\lambda)
\succeq(\mu/2)\mathbf{I}\succ\mathbf 0$.  A critical point
(by~\eqref{eq:ift_root_app}) with positive-definite Hessian is a
strict local minimum (second-order sufficiency), so the tracked point
$\bm\theta(\lambda)$ is a nondegenerate local minimizer of
$F(\cdot,\bm\phi^\star;\lambda)$ throughout $[0,\lambda_1]$; replacing
$\lambda_{\max}$ by $\min\{\lambda_{\max},\lambda_1\}$ if necessary
completes the proof.
\end{proof}
\begin{remark}[Genericity of condition~(v)]
\label{rem:nontrivial_genericity_app}
  Condition~(v) fails only when $[\mathbf{G}(\bm\phi^\star),\mathbf{H}]=0$ or when the
  expectation $\langle\boldsymbol{\psi}_D|[\mathbf{G}(\bm\phi^\star),\mathbf{H}]|\boldsymbol{\psi}_D\rangle$ is
  stationary at $\bm\theta_D^\star$ by a symmetry accident.  For
  Hamiltonians with no special commutation relations with single-qubit
  Pauli generators (our Tilted Ising and Random Ising benchmarks),
  this holds for all nonzero $\bm\phi^\star$.  For the TFIM, the
  $Z_2$ symmetry $[\mathbf{H},\prod_q \mathbf{X}_q]=0$ could create cancellations for
  specific $\bm\phi^\star$, but these form a measure-zero set.
\end{remark}
\begin{remark}[Contrast with the degenerate case $\bm\phi=\bm 0$]
\label{rem:degenerate_phi_zero_app}
  If $\bm\phi=\bm 0$, then $\mathbf{G}(\bm 0)=0$ and $\exp(-i\lambda\cdot 0)=\mathbf{I}$
  for all $\lambda$.  The homotopy energy reduces to
  $F(\bm\theta,\bm 0;\lambda)=E_D(\bm\theta)$ independently of
  $\lambda$, the mixed derivative is zero, and $\bm\theta'(0)=\bm 0$.
  The IFT yields only the trivial constant curve.  Condition~(iv)
  ($\bm\phi^\star\neq\bm 0$) is essential for non-trivial content.
\end{remark}
\begin{remark}[Interpretation]
\label{rem:continuation_interp}
  \Cref{thm:continuation} guarantees \emph{continuity with genuine
  deformation}: a good local minimum at depth~$D$ deforms smoothly
  into a nearby local minimum of the deeper landscape.  The driving
  force is the commutator expectation
  $\langle\boldsymbol{\psi}_D^\star|[\mathbf{G}(\bm\phi^\star),\mathbf{H}]|\boldsymbol{\psi}_D^\star\rangle$:
  the new rotation layers perturb the landscape precisely to the
  extent that their generators fail to commute with $\mathbf{H}$.
  For the jittered identity-pair initialization,
  $\|\bm\phi^\star\|=O(\sigma)$, so $\|\mathbf{G}(\bm\phi^\star)\|=O(\sigma)$
  and $\|\bm\theta'(0)\|=O(\sigma/\mu)$.  The old parameters shift by
  only a small amount---exactly the stability property we seek.
\end{remark}
\begin{remark}[Relationship to the classical IFT]
\label{rem:ift_scope_app}
  At its core, \Cref{thm:continuation} applies the standard IFT to a
  smooth family of energy functions.  The quantum-specific content
  enters through the perturbation bound~\eqref{eq:perturb}, which
  exploits the tensor-product structure of the new rotation layers to
  show that the perturbation is controlled by
  $\|\bm\phi^\star\|\cdot\sum_j\|\mathbf{G}_j\|$ rather than by the operator
  norm of a generic unitary.
\end{remark}
\subsection{Monotone energy guarantees: extended discussion}
\label{app:monotone_extended}
\Cref{cor:energy_monotone} established that the \emph{globally
optimal} energy is non-increasing with depth:
$E_{D+2}^\star\le E_D^\star$.  However, this is a statement about
the best possible parameters, which we cannot generally find in
polynomial time.  In practice, gradient descent produces a
\emph{local} minimum $\hat{\bm\theta}^{(D)}$, and there is no a
priori reason why the local minimum at depth~$D{+}2$ should have
lower energy than the one at depth~$D$.
This gap between theory and practice is the central challenge of
progressive depth training.  Na\"ive PDT can (and in our experiments,
does) exhibit energy \emph{increases} at expansion: the randomly
initialized new HEA blocks, each containing a CNOT ring, disrupt
the optimized state, and gradient descent may converge to a worse
local minimum.
IP-PDT addresses this through two mechanisms:
\begin{enumerate}[leftmargin=1.5em]
  \item \textbf{Identity-pair initialization}: the new rotation layers
    start at $E_{D+2}\approx E_D(\hat{\bm\theta}^{(D)})+O(\sigma)$.
  \item \textbf{Acceptance rule}: if optimization fails to improve,
    fall back to the embedded parameters
    $\iota_D(\hat{\bm\theta}^{(D)})$.
\end{enumerate}
The fallback is only possible because of the IEP.  In the standard HEA,
no identity embedding exists because every new block introduces an
irremovable CNOT ring.  In practice, the acceptance rule rarely triggers
the fallback: gradient descent from the identity-pair initialization
almost always improves the energy, because the optimizer starts in a
basin of the previous-stage minimum.
\subsection{Local strong convexity and the PL inequality: motivation and
validity}\label{app:pl_motivation}
\Cref{ass:local_reg} posits that $E_D$ is locally $\mu_D$-strongly
convex.  This is stronger than the Polyak--{\L}ojasiewicz (PL) condition
of Polyak~\cite{polyak1963gradient}---one of the weakest conditions
under which gradient descent converges linearly---but it is the
hypothesis a warm-started expansion actually needs, for two reasons
developed below: it \emph{implies} the PL inequality with a sharp
constant, and, unlike PL, it is stable under the $O(\sigma)$ depth
perturbation.

\medskip\noindent\textbf{Strong convexity near a nondegenerate minimum.}
If $\bm\theta_D^\star$ is a nondegenerate local minimum with
$\nabla^2 E_D(\bm\theta_D^\star)\succeq\mu \mathbf{I}$, then by continuity of
the Hessian there is a neighbourhood $\cN_D$ on which
$\nabla^2 E_D(\bm\theta)\succeq\tfrac{\mu}{2}\mathbf{I}$; on $\cN_D$ the exact
second-order Taylor expansion with mean-value remainder,
\[
  E_D(\bm\theta)= E_D(\bm\theta_D^\star)
  +\tfrac{1}{2}(\bm\theta-\bm\theta_D^\star)^T
  \nabla^2 E_D(\bm\xi) (\bm\theta-\bm\theta_D^\star)
\]
for some $\bm\xi$ on the segment $[\bm\theta_D^\star,\bm\theta]$ (using
$\nabla E_D(\bm\theta_D^\star)=\bm 0$), certifies local strong convexity
with $\mu_D$ of order the smallest Hessian eigenvalue at the minimum.

\medskip\noindent\textbf{Strong convexity $\Rightarrow$ PL, sharply.}
For any $\mu$-strongly convex $f$, minimizing both sides of
$f(\bm y)\ge f(\bm x)+\langle\nabla f(\bm x),\bm y-\bm x\rangle
+\tfrac{\mu}{2}\|\bm y-\bm x\|^2$ over $\bm y$ (the minimizer is
$\bm y=\bm x-\tfrac1\mu\nabla f(\bm x)$) gives
$f^\star\ge f(\bm x)-\tfrac{1}{2\mu}\|\nabla f(\bm x)\|^2$, i.e.\ the PL
inequality $\|\nabla f(\bm x)\|^2\ge 2\mu (f(\bm x)-f^\star)$ with
constant $\mu$.  The sharp constant matters: the cruder chaining
$\|\nabla f\|^2\ge\mu^2\|\bm x\|^2$ with $f\le\tfrac{L}{2}\|\bm x\|^2$
yields only $\|\nabla f\|^2\ge 2(\mu^2/L)f=2(\mu/\kappa)f$, weaker by the
condition number $\kappa=L/\mu$.  It is this strong-convexity argument,
not the chaining, that supplies the constant $\mu_{D+2}$ used
in~\eqref{eq:basin_PL}.

\medskip\noindent\textbf{When it fails, and why locality helps.}
Local strong convexity (hence PL) can fail (i)~far from the minimum,
where the landscape may have flat (barren) regions with near-zero
gradient but large suboptimality, or (ii)~at degenerate minima with a
singular Hessian.  Our assumption is explicitly \emph{local}: it need
only hold in a neighbourhood $\cN_D$ where the IP-PDT optimizer
operates, and since IP-PDT initializes each stage near the previous
stage's minimum, that neighbourhood can be taken small, increasing
plausibility.  In the barren-plateau regime, however, $\mu_D$ can be
exponentially small in $n$, giving an exponentially large condition
number $\kappa_D=L_D/\mu_D$ and an exponentially small admissible jitter
$\sigma_0=\mu_D/(2c)$; IP-PDT's shallow starting depth is designed to
mitigate this, but we do not claim the regime is escaped
unconditionally.
\subsection{Basin preservation: complete proof}
\label{app:basin_complete}
We provide the complete five-step proof of \Cref{thm:basin}.  Throughout,
$\bm\theta^+\coloneqq\bm\theta-\eta\nabla\widetilde{E}(\bm\theta)$ denotes
one gradient-descent step on $\widetilde{E}$ at $\bm\theta\in\cN_D$ with
step size $\eta\in(0,1/L_{D+2}]$, and all displayed constants are tracked
to the precision needed to certify $\mu_{D+2}>0$.
\begin{proof}[Complete proof of \Cref{thm:basin}]
Write $\widetilde{E}(\bm\theta)\coloneqq E_{D+2}(\bm\theta,\bm\phi^\star)$
and let $g\coloneqq 3nD$ be the number of old-layer coordinates.

\emph{Step~1: Relate to the depth-$D$ landscape and transfer the
regularity constants.}
By the perturbation bound~\eqref{eq:perturb} evaluated at the fixed
new-layer parameters $\bm\phi^\star$ with $\|\bm\phi^\star\|=O(\sigma)$
and $\|\mathbf{G}_j\|=1/2$,
\begin{equation}\label{eq:basin_energy_prox}
  |\widetilde{E}(\bm\theta)-E_D(\bm\theta)|
  \le 2\|\mathbf{H}\| \|\mathbf{G}(\bm\phi^\star)\|
  \eqqcolon\delta(\sigma),
  \qquad\text{uniformly for }\bm\theta\in\cN_D,
\end{equation}
the \emph{exact} bound~\eqref{eq:perturb} (first inequality, no
truncation); expanding the generator gives
$\delta(\sigma)\le 2\|\mathbf{H}\|\bigl(\sum_{j=1}^{6n}|\phi_j^\star| \|\mathbf{G}_j\|
+C_R\|\bm\phi^\star\|^2\bigr)$, which scales linearly in the jitter
$\sigma$ with the quadratic term retained.
The gradient and Hessian of the difference $\widetilde{E}-E_D$ obey, by
the component-wise bounds~\eqref{eq:grad_l2} and~\eqref{eq:hess_op_bound}
of \Cref{app:grad_hess},
\begin{equation}\label{eq:basin_grad_hess}
  \|\nabla\widetilde{E}(\bm\theta)-\nabla E_D(\bm\theta)\|
  \le\sqrt{g} \|\Delta\mathbf{H}\|,
  \qquad
  \|\nabla^2\widetilde{E}(\bm\theta)-\nabla^2 E_D(\bm\theta)\|
  \le g \|\Delta\mathbf{H}\|,
\end{equation}
where $\Delta\mathbf{H}=\mathbf{W}(\bm\phi^\star)^\dagger\mathbf{H} \mathbf{W}(\bm\phi^\star)-\mathbf{H}$
and $\|\Delta\mathbf{H}\|=O(\sigma)$ by~\eqref{eq:dH_bound}; both bounds hold
uniformly on $\cN_D$.  We name the resulting deviation constants
\begin{equation}\label{eq:basin_deviation_constants}
  a(\sigma)\coloneqq\sqrt{g} \|\Delta\mathbf{H}\|=O(\sigma),
  \qquad
  b(\sigma)\coloneqq g \|\Delta\mathbf{H}\|=O(\sigma).
\end{equation}

\textbf{Smoothness transfer.}  Fix $\bm\theta,\bm\theta'\in\cN_D$.
Adding and subtracting $\nabla E_D$ and applying the triangle
inequality,
\begin{align}
  \|\nabla\widetilde{E}(\bm\theta)
  -\nabla\widetilde{E}(\bm\theta')\|
  &\le\|\nabla E_D(\bm\theta)-\nabla E_D(\bm\theta')\|
  +\|(\nabla\widetilde{E}-\nabla E_D)(\bm\theta)
  -(\nabla\widetilde{E}-\nabla E_D)(\bm\theta')\|
  \notag\\
  &\le L_D\|\bm\theta-\bm\theta'\|
  +\Bigl(\sup_{\bm\theta''\in\cN_D}
     \|\nabla^2\widetilde{E}(\bm\theta'')-\nabla^2 E_D(\bm\theta'')\|\Bigr)
     \|\bm\theta-\bm\theta'\|
  \label{eq:basin_smooth_transfer}\\
  &\le\bigl(L_D+b(\sigma)\bigr)\|\bm\theta-\bm\theta'\| ,
  \notag
\end{align}
where the first norm uses the $L_D$-smoothness of $E_D$
(\Cref{ass:local_reg}), and the second uses the mean-value inequality
applied to the map $\bm\theta\mapsto(\nabla\widetilde{E}-\nabla E_D)
(\bm\theta)$ (whose Jacobian is the Hessian
difference~\eqref{eq:basin_grad_hess}) on the convex
set\footnote{If $\cN_D$ is not convex, replace it by an open ball
$\subseteq\cN_D$ around $\bm\theta_D^\star$; \Cref{ass:local_reg} holds
on any such sub-neighbourhood, and the basin argument only ever uses
points in this ball.} $\cN_D$.  Hence $\widetilde{E}$ is
$L_{D+2}$-smooth on $\cN_D$ with
\begin{equation}\label{eq:basin_L}
  L_{D+2}=L_D+b(\sigma)=L_D+O(\sigma).
\end{equation}

\textbf{Strong-convexity (PL) constant transfer.}  By
\Cref{ass:local_reg}, $\nabla^2 E_D(\bm\theta)\succeq\mu_D\mathbf{I}$ for all
$\bm\theta\in\cN_D$, i.e.\ the smallest eigenvalue satisfies
$\lambda_{\min}(\nabla^2 E_D(\bm\theta))\ge\mu_D$.  The two Hessians
$\nabla^2\widetilde{E}(\bm\theta)$ and $\nabla^2 E_D(\bm\theta)$ are
real symmetric, and their difference has spectral norm at most
$b(\sigma)$ by~\eqref{eq:basin_grad_hess}.  Weyl's eigenvalue
perturbation inequality, $|\lambda_{\min}(\mathbf{A}+\mathbf{E})
-\lambda_{\min}(\mathbf{A})|\le\|\mathbf{E}\|$ for symmetric $\mathbf{A},\mathbf{E}$, applied with
$\mathbf{A}=\nabla^2 E_D(\bm\theta)$ and
$\mathbf{E}=\nabla^2\widetilde{E}(\bm\theta)-\nabla^2 E_D(\bm\theta)$, gives
\begin{equation}\label{eq:basin_weyl}
  \lambda_{\min}\bigl(\nabla^2\widetilde{E}(\bm\theta)\bigr)
  \ge\lambda_{\min}\bigl(\nabla^2 E_D(\bm\theta)\bigr)-b(\sigma)
  \ge\mu_D-b(\sigma),
  \qquad\bm\theta\in\cN_D.
\end{equation}
Because $b(\sigma)=g \|\Delta\mathbf{H}\|$ is proportional to $\sigma$
(by~\eqref{eq:dH_bound}, $\|\Delta\mathbf{H}\|\le 2\|\mathbf{H}\|
\sum_{j}|\phi_j^\star|$ with $\|\bm\phi^\star\|=O(\sigma)$), write
$b(\sigma)=c \sigma+O(\sigma^2)$ where
\begin{equation}\label{eq:basin_c_scaling}
  c=O \bigl(g \|\mathbf{H}\|\bigr)=O \bigl(n^{3/2}D \|\mathbf{H}\|\bigr)
\end{equation}
is a \emph{problem-dependent} constant: it is not absolute but grows with
the old-parameter count $g=3nD$ and the Hamiltonian norm (and depends on
the realized jitter through $\sum_j|\phi_j^\star|/\sigma$).  Consequently
the admissible-jitter threshold $\sigma_0=\mu_D/(2c)$ shrinks both with
the (possibly exponentially small) curvature $\mu_D$ \emph{and}
polynomially in $n,D$---a limitation we record explicitly.  Setting
\begin{equation}\label{eq:basin_mu}
  \mu_{D+2}\coloneqq\mu_D-b(\sigma)=\mu_D-c \sigma+O(\sigma^2),
\end{equation}
the lower bound~\eqref{eq:basin_weyl} reads
$\nabla^2\widetilde{E}\succeq\mu_{D+2}\mathbf{I}$ on $\cN_D$.  This is a
\emph{positive} bound precisely under the quantitative smallness
condition
\begin{equation}\label{eq:basin_sigma_threshold}
  \sigma\le\sigma_0\coloneqq\mu_D/(2c)
  \qquad\Bigl(\text{equivalently, }b(\sigma)\le\mu_D/2\Bigr),
\end{equation}
which we assume henceforth; under it
$\mu_{D+2}=\mu_D-b(\sigma)\ge\mu_D/2>0$, the explicit form of the
statement's clause $\mu_{D+2}>0$.
Then $\widetilde{E}$ is
$\mu_{D+2}$-strongly convex on the convex neighbourhood $\cN_D$, and
$\mu_{D+2}$-strong convexity implies the $\mu_{D+2}$-PL inequality on
$\cN_D$ (a standard implication: strong convexity $\Rightarrow$
$\widetilde{E}(\bm\theta)-\widetilde{E}^\star
\le\tfrac{1}{2\mu_{D+2}}\|\nabla\widetilde{E}(\bm\theta)\|^2$, which is
exactly PL with constant $\mu_{D+2}$), where
$\widetilde{E}^\star\coloneqq\widetilde{E}(\bm\theta_{D+2}^\star)
=\min_{\bm\theta\in\cN_D}\widetilde{E}(\bm\theta)$ and
$\bm\theta_{D+2}^\star\in\cN_D$ is the (unique, by strong convexity)
minimizer.  The PL inequality used in Step~4 is therefore
\begin{equation}\label{eq:basin_PL}
  \|\nabla\widetilde{E}(\bm\theta)\|^2
  \ge 2\mu_{D+2}\bigl(\widetilde{E}(\bm\theta)-\widetilde{E}^\star\bigr),
  \qquad\bm\theta\in\cN_D.
\end{equation}

\emph{Interior stationarity of $\bm\theta_{D+2}^\star$.}
Strong convexity makes the minimizer of $\widetilde{E}$ over $\cN_D$
unique; we check it is an \emph{interior} stationary point, so that
$\nabla\widetilde{E}(\bm\theta_{D+2}^\star)=\bm 0$ (used in the
invariance argument and in Step~5).  The previous-stage optimum
$\bm\theta_D^\star$ is interior to $\cN_D$, and
$\|\nabla\widetilde{E}(\bm\theta_D^\star)\|
=\|\nabla\widetilde{E}(\bm\theta_D^\star)-\nabla E_D(\bm\theta_D^\star)\|
\le a(\sigma)$ by~\eqref{eq:basin_grad_hess} (using
$\nabla E_D(\bm\theta_D^\star)=\bm 0$), where $a(\sigma)=\sqrt{g} \|\Delta\mathbf{H}\|$
is the gradient-deviation bound, continuous and strictly increasing in
$\sigma$ with $a(0)=0$.  Strong convexity then places the stationary point
within $\|\bm\theta_{D+2}^\star-\bm\theta_D^\star\|\le a(\sigma)/\mu_{D+2}$
of $\bm\theta_D^\star$.  Let $r_D:=\mathrm{dist}(\bm\theta_D^\star,\partial\cN_D)>0$
be the in-radius of $\cN_D$ at $\bm\theta_D^\star$.  The displacement stays
below $r_D$---so $\bm\theta_{D+2}^\star$ is interior---whenever
$a(\sigma)/\mu_{D+2}<r_D$; using $\mu_{D+2}\ge\mu_D/2$ this holds for
$\sigma\le\sigma_1$, where $\sigma_1$ is the explicit threshold defined by
$a(\sigma_1)=\mu_D r_D/2$ (unique since $a$ is increasing).  We henceforth
take $\sigma\le\min(\sigma_0,\sigma_1)$ (still abbreviated $\sigma_0$), giving
$\nabla\widetilde{E}(\bm\theta_{D+2}^\star)=\bm 0$.

\emph{Step~2: Descent lemma from $L_{D+2}$-smoothness.}
A function that is $L_{D+2}$-smooth on the convex set $\cN_D$ satisfies,
for all $\bm\theta,\bm y\in\cN_D$, the quadratic upper bound
$\widetilde{E}(\bm y)\le\widetilde{E}(\bm\theta)
+\langle\nabla\widetilde{E}(\bm\theta),\bm y-\bm\theta\rangle
+\tfrac{L_{D+2}}{2}\|\bm y-\bm\theta\|^2$ (the descent lemma, obtained
by integrating $\nabla\widetilde{E}$ along the segment $[\bm\theta,\bm y]$
and bounding the integrand via Lipschitzness of the gradient).  Taking
$\bm y=\bm\theta^+=\bm\theta-\eta\nabla\widetilde{E}(\bm\theta)$, so that
$\bm\theta^+-\bm\theta=-\eta\nabla\widetilde{E}(\bm\theta)$,
\begin{align}
  \widetilde{E}(\bm\theta^+)
  &\le\widetilde{E}(\bm\theta)
  +\langle\nabla\widetilde{E}(\bm\theta),
  -\eta\nabla\widetilde{E}(\bm\theta)\rangle
  +\frac{L_{D+2}}{2}\|{-}\eta\nabla\widetilde{E}(\bm\theta)\|^2
  \notag\\
  &=\widetilde{E}(\bm\theta)
  -\eta\|\nabla\widetilde{E}(\bm\theta)\|^2
  +\frac{L_{D+2}\eta^2}{2}\|\nabla\widetilde{E}(\bm\theta)\|^2
  =\widetilde{E}(\bm\theta)
  -\eta\Bigl(1-\tfrac{L_{D+2}\eta}{2}\Bigr)
   \|\nabla\widetilde{E}(\bm\theta)\|^2 .
  \label{eq:basin_descent}
\end{align}
(The descent lemma requires the segment
$[\bm\theta,\bm\theta^+]\subseteq\cN_D$.  This is supplied by the
iterate-invariance argument below, which shows---from strong convexity
alone, \emph{independently} of the contraction derived here---that every
iterate remains in $\cN_D$.)

\emph{Step~3: The step restriction $\eta\le 1/L_{D+2}$ gives a clean
$\tfrac12$ factor.}
For $\eta\le 1/L_{D+2}$ we have
$\tfrac{L_{D+2}\eta}{2}\le\tfrac12$, hence
$1-\tfrac{L_{D+2}\eta}{2}\ge\tfrac12>0$.  The coefficient
in~\eqref{eq:basin_descent} is therefore at least $\eta/2$, and because
$\|\nabla\widetilde{E}(\bm\theta)\|^2\ge0$,
\begin{equation}\label{eq:basin_half}
  \widetilde{E}(\bm\theta^+)
  \le\widetilde{E}(\bm\theta)
  -\frac{\eta}{2}\|\nabla\widetilde{E}(\bm\theta)\|^2 .
\end{equation}

\emph{Step~4: Apply the PL inequality to obtain the contraction.}
Subtract $\widetilde{E}^\star$ from both sides
of~\eqref{eq:basin_half} and substitute the PL
inequality~\eqref{eq:basin_PL} into the negative term:
\begin{align}
  \widetilde{E}(\bm\theta^+)-\widetilde{E}^\star
  &\le\bigl(\widetilde{E}(\bm\theta)-\widetilde{E}^\star\bigr)
  -\frac{\eta}{2}\|\nabla\widetilde{E}(\bm\theta)\|^2
  \notag\\
  &\le\bigl(\widetilde{E}(\bm\theta)-\widetilde{E}^\star\bigr)
  -\frac{\eta}{2}\cdot 2\mu_{D+2}
    \bigl(\widetilde{E}(\bm\theta)-\widetilde{E}^\star\bigr)
  =(1-\eta\mu_{D+2})
  \bigl(\widetilde{E}(\bm\theta)-\widetilde{E}^\star\bigr).
  \label{eq:basin_contraction}
\end{align}
This is~\eqref{eq:linear_rate}.  The contraction factor lies in
$[0,1)$: it is $<1$ because $\eta\mu_{D+2}>0$
(using~\eqref{eq:basin_sigma_threshold}), and it is $\ge0$ because
$\eta\le 1/L_{D+2}\le 1/\mu_{D+2}$ (as $\mu_{D+2}\le\lambda_{\min}
\le\lambda_{\max}\le L_{D+2}$), so $\eta\mu_{D+2}\le1$.

\emph{Iterate invariance (non-circular).}
Before iterating we must check that every iterate stays in $\cN_D$, so
that the one-step bound applies along the whole trajectory; we establish
this from the Hessian lower bound alone, \emph{without} invoking the
value contraction~\eqref{eq:basin_contraction}.  Since $\widetilde{E}$ is
$\mu_{D+2}$-strongly convex and $L_{D+2}$-smooth on the convex set
$\cN_D$, the gradient-descent map
$\Phi(\bm\theta)\coloneqq\bm\theta-\eta\nabla\widetilde{E}(\bm\theta)$ with
$\eta\le 1/L_{D+2}$ contracts toward its fixed point
$\bm\theta_{D+2}^\star$ in Euclidean distance,
\begin{equation}\label{eq:basin_iterate_contraction}
  \|\Phi(\bm\theta)-\bm\theta_{D+2}^\star\|
  \le(1-\eta\mu_{D+2}) \|\bm\theta-\bm\theta_{D+2}^\star\|
  \le\|\bm\theta-\bm\theta_{D+2}^\star\| ,
  \qquad\bm\theta\in\cN_D,
\end{equation}
the standard contraction bound for gradient descent on a smooth strongly
convex function with $\eta\le 1/L_{D+2}$
(see \cite[\S3.4.2]{bubeck2015convex}).  Writing
$\bm u=\bm\theta-\bm\theta_{D+2}^\star$ and
$\bm g=\nabla\widetilde{E}(\bm\theta)$ (so that
$\nabla\widetilde{E}(\bm\theta_{D+2}^\star)=\bm 0$ and
$\Phi(\bm\theta)-\bm\theta_{D+2}^\star=\bm u-\eta\bm g$),
\[
  \|\bm u-\eta\bm g\|^2
  =\|\bm u\|^2-2\eta\langle\bm g,\bm u\rangle+\eta^2\|\bm g\|^2 ;
\]
the strong-convexity/smoothness co-coercivity inequality
$\langle\bm g,\bm u\rangle\ge
\tfrac{\mu_{D+2}L_{D+2}}{\mu_{D+2}+L_{D+2}}\|\bm u\|^2
+\tfrac{1}{\mu_{D+2}+L_{D+2}}\|\bm g\|^2$ makes the $\|\bm g\|^2$
coefficient $\eta^2-\tfrac{2\eta}{\mu_{D+2}+L_{D+2}}\le0$ for
$\eta\le 1/L_{D+2}\le 2/(\mu_{D+2}+L_{D+2})$; lower-bounding that
non-positive term via $\|\bm g\|\ge\mu_{D+2}\|\bm u\|$ (strong monotonicity)
and recombining leaves
$\|\bm u-\eta\bm g\|^2\le(1-\eta\mu_{D+2})^2\|\bm u\|^2$, i.e.\ contraction
by the factor $(1-\eta\mu_{D+2})$.  (Strong monotonicity alone would not
suffice---the combined co-coercivity bound is required.)  This uses only
strong convexity and smoothness, \emph{not} the energy
contraction~\eqref{eq:basin_contraction}.  By the hypothesis of
\Cref{thm:basin}, the closed ball
$\overline{B}(\bm\theta_{D+2}^\star,\|\bm\theta^{(0)}
-\bm\theta_{D+2}^\star\|)\subseteq\cN_D$; induction on $k$
using~\eqref{eq:basin_iterate_contraction} gives
$\|\bm\theta^{(k)}-\bm\theta_{D+2}^\star\|
\le\|\bm\theta^{(0)}-\bm\theta_{D+2}^\star\|$ for every $k$, so each iterate
remains in that ball and hence in $\cN_D$.  The descent lemma therefore
applies at every step, and iterating the one-step
contraction~\eqref{eq:basin_contraction} over the $T$ updates
$\bm\theta^{(0)}\mapsto\bm\theta^{(1)}\mapsto\cdots\mapsto\bm\theta^{(T)}$
yields the $T$-step bound~\eqref{eq:T_step_convergence}.

\emph{Step~5: Energy proximity of the two minimizers.}
Finally we quantify how close the depth-$(D{+}2)$ minimum value is to
the depth-$D$ minimum value, justifying that the inherited basin is the
\emph{right} one.  By the triangle inequality,
\begin{equation}\label{eq:basin_step5_split}
  |\widetilde{E}(\bm\theta_{D+2}^\star)-E_D(\bm\theta_D^\star)|
  \le\underbrace{|\widetilde{E}(\bm\theta_{D+2}^\star)
     -E_D(\bm\theta_{D+2}^\star)|}_{(\mathrm{i})}
  +\underbrace{|E_D(\bm\theta_{D+2}^\star)
     -E_D(\bm\theta_D^\star)|}_{(\mathrm{ii})} .
\end{equation}
Term~(i) is bounded by $\delta(\sigma)=O(\sigma)$ directly
from~\eqref{eq:basin_energy_prox} (with
$\bm\theta=\bm\theta_{D+2}^\star\in\cN_D$).  For term~(ii), $\bm\theta_D^\star$
is the minimizer of the $L_D$-smooth function $E_D$ on $\cN_D$, so
$\nabla E_D(\bm\theta_D^\star)=\bm 0$; the descent-lemma upper bound for
$E_D$ at $\bm y=\bm\theta_{D+2}^\star$ then gives
\begin{equation}\label{eq:basin_step5_ii}
  E_D(\bm\theta_{D+2}^\star)-E_D(\bm\theta_D^\star)
  \le\langle\underbrace{\nabla E_D(\bm\theta_D^\star)}_{= \bm 0},
       \bm\theta_{D+2}^\star-\bm\theta_D^\star\rangle
  +\frac{L_D}{2}\|\bm\theta_{D+2}^\star-\bm\theta_D^\star\|^2
  =\frac{L_D}{2}\|\bm\theta_{D+2}^\star-\bm\theta_D^\star\|^2 .
\end{equation}
The minimizer displacement is itself $O(\sigma)$, by a self-contained
strong-convexity argument that needs no appeal to the continuation
theorem.  Both points are stationary,
$\nabla\widetilde{E}(\bm\theta_{D+2}^\star)=\bm 0$ and
$\nabla E_D(\bm\theta_D^\star)=\bm 0$.  Strong convexity of
$\widetilde{E}$ gives the monotonicity lower bound
$\|\nabla\widetilde{E}(\bm\theta_D^\star)
-\nabla\widetilde{E}(\bm\theta_{D+2}^\star)\|
\ge\mu_{D+2}\|\bm\theta_D^\star-\bm\theta_{D+2}^\star\|$, while its
left-hand side equals the gradient deviation
$\nabla\widetilde{E}(\bm\theta_D^\star)-\nabla E_D(\bm\theta_D^\star)$
(using $\nabla E_D(\bm\theta_D^\star)=\bm 0$), of norm at most
$a(\sigma)=O(\sigma)$ by~\eqref{eq:basin_grad_hess}.  Hence
$\|\bm\theta_{D+2}^\star-\bm\theta_D^\star\|
\le a(\sigma)/\mu_{D+2}=O(\sigma/\mu_D)$, so the right-hand side
of~\eqref{eq:basin_step5_ii} is $O(\sigma^2)$.  Combining,
\begin{equation}\label{eq:basin_step5_final}
  |\widetilde{E}(\bm\theta_{D+2}^\star)-E_D(\bm\theta_D^\star)|
  \le\delta(\sigma)+\frac{L_D}{2} O(\sigma^2/\mu_D^2)
  =O(\sigma),
\end{equation}
so the value of the inherited basin differs from the previous-stage
optimum by only $O(\sigma)$, confirming that the contraction
in~\eqref{eq:basin_contraction} drives the energy toward a minimizer
$O(\sigma)$-close to $E_D(\bm\theta_D^\star)$.
\end{proof}
\begin{remark}[Degenerate case $\bm\phi=\bm 0$]
\label{rem:degenerate_basin_app}
  Setting $\bm\phi=\bm 0$ gives
  $E_{D+2}(\bm\theta,\bm 0)=E_D(\bm\theta)$ identically, and the
  bound reduces to the standard PL convergence theorem on $E_D$.  The
  nonzero $\bm\phi^\star$ is essential for new content.
\end{remark}
\begin{remark}[Contrast with na\"ive PDT]
\label{rem:naive_contrast_app}
  In na\"ive PDT, the energy after expansion is at some random point
  in the depth-$(D{+}2)$ landscape.  The optimizer may land outside
  the basin of any good minimum, requiring an unpredictable ``recovery''
  phase.  IP-PDT avoids this by starting each stage inside the basin
  of the previous-stage minimum.
\end{remark}
\subsection{Gap--fidelity sandwich: complete proof and applications}
\label{app:gap_fidelity_extended}
\begin{proof}[Complete proof of \Cref{thm:gap_fidelity}]
\emph{Step~1: Expand $|\boldsymbol{\psi}\rangle$ in the eigenbasis.}
Write $|\boldsymbol{\psi}\rangle=\sum_{i=1}^d c_i|\mathbf{v}_i\rangle$ with $c_i\in\bbC$,
$\sum_i|c_i|^2=1$.  The fidelity is $F(\boldsymbol{\psi})=|c_1|^2$, so
$\sum_{i\ge 2}|c_i|^2=1-F(\boldsymbol{\psi})$.
\emph{Step~2: Express the energy suboptimality.}
\begin{align*}
  \langle\boldsymbol{\psi}|\mathbf{H}|\boldsymbol{\psi}\rangle-\lambda_1
  &=\sum_i\lambda_i|c_i|^2-\lambda_1\sum_i|c_i|^2
  =\sum_{i\ge 2}(\lambda_i-\lambda_1)|c_i|^2.
\end{align*}
\emph{Step~3: Lower bound.}
$\lambda_i-\lambda_1\ge\lambda_2-\lambda_1=\Delta$ for $i\ge 2$, so
$\sum_{i\ge 2}(\lambda_i-\lambda_1)|c_i|^2
\ge\Delta\sum_{i\ge 2}|c_i|^2=\Delta(1-F(\boldsymbol{\psi}))$.
\emph{Step~4: Upper bound.}
$\lambda_i-\lambda_1\le\lambda_d-\lambda_1$ for all $i$, so
$\sum_{i\ge 2}(\lambda_i-\lambda_1)|c_i|^2
\le(\lambda_d-\lambda_1)(1-F(\boldsymbol{\psi}))$.
\emph{Step~5: Combine.}
$\Delta(1-F(\boldsymbol{\psi}))
\le\langle\boldsymbol{\psi}|\mathbf{H}|\boldsymbol{\psi}\rangle-\lambda_1
\le(\lambda_d-\lambda_1)(1-F(\boldsymbol{\psi}))$.
\end{proof}

\medskip\noindent\textbf{Applications to our benchmarks.}
\noindent\textbf{1.\ Energy--fidelity certificate.}
If $\langle\boldsymbol{\psi}|\mathbf{H}|\boldsymbol{\psi}\rangle-\lambda_1\le\varepsilon$, then
$F(\boldsymbol{\psi})\ge 1-\varepsilon/\Delta$.  For our Tilted Ising Hamiltonian,
IP-PDT achieves $\varepsilon\approx 0.02$ with $\Delta\approx 0.44$,
giving $F\ge 1-0.02/0.44\approx 0.955$.
\noindent\textbf{2.\ Condition number interpretation.}
The ratio $(\lambda_d-\lambda_1)/\Delta$ quantifies how efficiently
energy improvements translate to fidelity gains:
\begin{itemize}[leftmargin=1.5em]
  \item TFIM: $\Delta\approx 0.27$,
    $\lambda_d-\lambda_1\approx 12.7$, ratio $\approx 47$.
  \item Tilted Ising: $\Delta\approx 0.44$,
    $\lambda_d-\lambda_1\approx 18.6$, ratio $\approx 42$.
  \item Random Ising: $\Delta\approx 0.53$,
    $\lambda_d-\lambda_1\approx 5.8$, ratio $\approx 11$.
\end{itemize}
The Random Ising Hamiltonian has the most favorable ratio, consistent
with IP-PDT's strong performance on this benchmark.

\end{document}